\documentclass[%
 reprint,
 amsmath,amssymb,
 aps,
prb,
]{revtex4-2}

\usepackage{bm}

\usepackage{graphicx}
\usepackage{amsthm}
\usepackage{amsmath, amssymb, amsfonts}
\usepackage{enumitem}
\usepackage{bm}
\usepackage{comment}
\usepackage{multirow}
\usepackage{booktabs}
\usepackage{tikz}
\usetikzlibrary{arrows.meta,positioning,calc,backgrounds}
\usepackage{color}
\usepackage{mathtools}
\usepackage{needspace}
\usepackage[margin=20truemm]{geometry}
\usepackage[colorlinks=true,linkcolor=blue,citecolor=blue,urlcolor=blue]{hyperref}
\theoremstyle{plain}

\newtheorem*{thm*}{Theorem~\ref{t:main}}

\newtheorem{lemma}{Lemma}
\newtheorem{theorem}{Theorem}
\theoremstyle{definition}

\newtheorem*{dfn*}{Definition}

\begin{document}


\title{Proof of the absence of local conserved quantities in the Holstein model}

\author{Fuga Ishii}
\email{fuga@i-shi-i.name}
\affiliation{Department of Integrated Sciences, College of Arts and Sciences, The University of Tokyo, 3-8-1 Komaba, Meguro, Tokyo 153-8902, Japan}

\author{Mizuki Yamaguchi}
\email{yamaguchi-q@g.ecc.u-tokyo.ac.jp}
\affiliation{Department of Basic Science, Department of Multidisciplinary Sciences, Graduate School of Arts and Sciences, The University of Tokyo, 3-8-1 Komaba, Meguro, Tokyo 153-8902, Japan}

\date{\today}

\begin{abstract}
Absence of local conserved quantities, or \textit{nonintegrability}, is often assumed when discussing various phenomena in quantum many-body systems, such as thermalization and transport. However, no concrete proof of this property is known in electron--phonon coupled systems, a typical setting for condensed matter physics. 
In this paper, we show that the one-dimensional Holstein model has no nontrivial local conserved quantities other than the Hamiltonian itself and the total fermion number operator. We further show that the absence of nontrivial local conserved quantities also holds for the more general Holstein--Hubbard model.
We also establish the absence of nontrivial local conserved quantities in these models on higher-dimensional hypercubic lattices.
Our result has accomplished an
advance in nonintegrability proofs by expanding their scope to systems in which particles with different statistical properties are mixed.
\end{abstract}

\maketitle

\section{Introduction}
To understand quantum many-body dynamics, it is essential to identify the full set of local conserved quantities of a model. Local conservation laws strongly constrain the dynamics of a system, including how the system relaxes and how transport develops. Although isolated many-body systems are generally expected to thermalize under generic conditions~\cite{quantumchaosRev,ThermalizationRev,ETHDeutsch,ETHSrednicki,ETHnature}, this expectation can fail if relevant local conserved quantities are present but not considered~\cite{ThermalizationRev,ETHnature}. In integrable systems, where an extensive number of local conserved quantities exist, stationary states are constrained by infinitely many conditions in the thermodynamic limit and are described by generalized Gibbs ensembles rather than conventional thermal ensembles~\cite{GGE1,GGE2}. Accordingly, a complete characterization of local conserved quantities is not optional; it is a prerequisite for a correct description of long-time behavior.

A systematic understanding of conserved quantities is also important for transport theory. From the viewpoint of linear response, 
the Green--Kubo formula~\cite{Green,Kubo} can break down when the observable of interest has finite overlap with a conserved quantity~\cite{Mazur,Suzuki}. 
In particular, overlap between conserved quantities and current operators leads to ballistic transport rather than ordinary diffusion~\cite{ballistic1,ballistic2,TransprtRev}. 
Beyond linear response, the role of conservation laws becomes even more explicit in far-from-equilibrium settings, where one-dimensional systems with many local or quasi-local conservation laws motivate generalized hydrodynamic descriptions~\cite{GHDnature,GHD1,GHD2,GHD3,GHDrev}.

Motivated by these perspectives, rigorous approaches have recently been developed to exhaustively identify local conserved quantities~\cite{XYZh,HokkyoNI}. 
Such methods have been successfully applied to several classes of quantum systems, including spin, fermionic, and bosonic models.
As a result, many nonintegrable models have been shown to lack local conserved quantities~\cite{mixedIsing,ChibaYoneta,PXP1,PXP2,NNNHeisenberg,spinNNcomplete,NNcompleteclassification,BLBQ,BLBQPark,2DIsing,2DXY,NNNcomplete,compass,SU2,Fredkin,2DHubbard,NHboson,FanThreeSite,SengokuWatanabe}, and the possible forms of local conserved quantities have also been clarified in known integrable models~\cite{XYZ,XXXMPO,XYZMPO,1DHubbard,1DHubbard2}.
By contrast, systems involving couplings between different kinds of operators---most notably fermion--boson coupled systems---have remained largely unexplored. This is a significant gap, since fermion--boson couplings, such as electron--phonon interactions, arise broadly in condensed-matter physics.

In this paper, we investigate local conserved quantities
in the one-dimensional Holstein model, a paradigmatic electron--phonon coupled system. 
We rigorously prove that, away from the known integrable parameter regimes, the model possesses no nontrivial local conserved quantities other than the Hamiltonian itself and the total particle-number operator. This establishes the nonintegrability of the Holstein model in the above sense. We further show that a similar result extends to the Holstein--Hubbard model, which provides another important example of an electron--phonon coupled system. In this way, our result opens the way to rigorous nonintegrability proofs for a wider class of systems involving particles with different statistical properties.
We moreover extend the absence result to the Holstein and Holstein--Hubbard models on hypercubic lattices in arbitrary spatial dimension.

The rest of this paper is organized as follows. In Sec.~\ref{sec:main_result}, we define the models, introduce our notion of local conserved quantities, and state the main theorems. Before proceeding to the proof, we present the strategy in Sec.~\ref{sec:proof_strategy}, including the operator-basis expansion that we utilize in the subsequent sections. We then provide detailed proofs for the main result in Sec.~\ref{sec:proof}, followed by the analysis of short-support conserved quantities in Sec.~\ref{sec:short}. 
We also discuss the Holstein--Hubbard model in Sec.~\ref{sec:HH} and extend the result to higher-dimensional hypercubic lattices in Sec.~\ref{sec:higher_dimensions}. We summarize
our results and discuss implications and outlook in Sec.~\ref{sec:discussion}.
We further complement these results in the appendices. In
Appendix~\ref{sec:integrable_case}, we review the known integrable parameter
regimes. We also discuss the Holstein model with a fermionic chemical potential
in Appendix~\ref{sec:withchem}. The proofs for the Holstein--Hubbard and
higher-dimensional models are presented in Appendices~\ref{sec:HH_proof} and
\ref{sec:higher_dimensions_proof}, respectively.

\section{Setup and main results}\label{sec:main_result}

We consider the one-dimensional (spinless) Holstein model~\cite{original} on $L$ sites with periodic boundary conditions. The Hamiltonian is given by
\begin{align}
 \hat{H}   &=  \hat{H}_{\mathrm{hop}} + \hat{H}_{\mathrm{int}} + \hat{H}_{\mathrm{pho}},
\label{Eq.1D_H_model}
\end{align}
with
\begin{align}
\hat{H}_{\mathrm{hop}} &= t \sum_{i=1}^L (\hat{c}^\dagger_i \hat{c}_{i+1} + h.c.),
\end{align}
\begin{align}
\hat{H}_{\mathrm{int}} &= g \sum_{i=1}^L \hat{n}_i \bigl(\hat{b}_i^\dagger + \hat{b}_i \bigr),
\end{align}
\begin{align}
\hat{H}_{\mathrm{pho}} &= \omega \sum_{i=1}^L \hat{b}^\dagger_i \hat{b}_i,
\end{align}
where $\hat{c}^\dagger_i$ ($\hat{c}_i$) is the creation (annihilation) operator of a fermion at site $i$, and $\hat{b}^\dagger_i$ ($\hat{b}_i$) is the creation (annihilation) operator of a phonon at site $i$. Here $t$ is the hopping amplitude, $g$ is the electron--phonon coupling strength, and $\omega$ is the phonon frequency. We also write $\hat{n}_i = \hat{c}^\dagger_i \hat{c}_i$ for the fermion number operator. In some situations one may further add a chemical-potential term to this Hamiltonian; we discuss that case separately in Appendix~\ref{sec:withchem}.

To state our claim precisely, we first define what we mean by a local conserved quantity. An operator is called a \textit{$k$-support operator} if its minimal contiguous support is contained in $k$ sites. For example, $\hat{c}_i \hat{c}_{i+1} \hat{c}_{i+2}$ is a $3$-support operator, while $\hat{n}_i \hat{b}_{i+4}$ is a $5$-support operator. A \textit{$k$-local quantity} is then an operator that can be expressed as a sum of $k$-support operators, but cannot be expressed as a sum of $(k-1)$-support operators. A \textit{$k$-local conserved quantity} is a $k$-local quantity that commutes with the Hamiltonian.

Throughout this paper, we refer to $k$-local conserved quantities with $k=O(L^0)$ as local conserved quantities. We characterize nonintegrability by the absence of nontrivial local conserved quantities. More specifically, we call $k$-local conserved quantities with $k \geq 3$ and $k=O(L^0)$ \textit{nontrivial} local conserved quantities, in order to distinguish them from trivial ones such as the Hamiltonian itself. This notion of nonintegrability is consistent with that adopted in previous studies~\cite{ThermalizationRev,CauxMossel,XYZh,GrabowskiMathieuTest}.

The Holstein model is a paradigmatic electron--phonon coupled system. It captures essential aspects of lattice polaron formation through the local coupling between itinerant fermions and dispersionless phonons, and has been studied extensively in the context of polaron physics and related analytical and numerical approaches~\cite{HolsteinRev,DMFT,DMFTrelaxation,DMFT1D,Lanczos,LFS,localbasisoptimization,CDWmelting,ML,RNN,quantchem,dynamicaltypicality}.

From the viewpoint of local conserved quantities, the one-dimensional Holstein model is widely believed to be nonintegrable, although this belief has mostly been supported indirectly. Indeed, several results are consistent with nonintegrability, including relaxation toward Gibbs states after quenches~\cite{PolaronRelaxation,PolaronThermalization}, random-matrix like behavior in quantum-chaos indicators~\cite{HolsteinETH}, and ETH-type behavior inferred from the dependence of diagonal and off-diagonal matrix elements of observables on energy eigenvalues~\cite{HolsteinETH}. Transport studies provide further evidence consistent with this picture: optical, thermal, and dc conductivity calculations report diffusive or otherwise dissipative transport behavior~\cite{dynamicalproperty,MPS}, while real-time simulations of a Holstein polaron exhibit a crossover from short-time ballistic spreading to long-time diffusion~\cite{transportHEOM}.

At the same time, the model is analytically solvable in certain limiting cases~\cite{LangFirsov}. In particular, when any one of the three parameters $t$, $g$, or $\omega$ is set to zero, the model becomes integrable, in the sense that it possesses an extensive number of local conserved quantities. We discuss these solvable cases in Appendix~\ref{sec:integrable_case}. Away from these solvable limits, we show that the model is nonintegrable whenever $t \cdot g \cdot \omega \neq 0$.

\begin{theorem}[Main Result]\label{nontrivial}
The one-dimensional Holstein model~\eqref{Eq.1D_H_model} with $t \neq 0$, $g \neq 0$, and $\omega \neq 0$ has no $k$-local conserved quantity with $3 \leq k < L/2$.
\end{theorem}

Theorem~\ref{nontrivial} gives a rigorous formulation of the nonintegrability of the Holstein model. That is, away from the solvable limits, the model possesses no nontrivial local conserved quantities. Thus, our result provides a firm analytical foundation for the nonintegrable behavior that has previously been inferred mainly from dynamical, spectral, and transport properties.

We then complement Theorem~\ref{nontrivial} by analyzing the cases $k=1$ and $k=2$, which yields the following complete classification of local conserved quantities with $k < L/2$.

\begin{theorem}\label{withtrivial}
In the one-dimensional Holstein model~\eqref{Eq.1D_H_model} with $t \neq 0$, $g \neq 0$, and $\omega \neq 0$, $k$-local conserved quantities with $k < L/2$ are restricted to linear combinations of the following:
\begin{enumerate}[label={(\roman*)},ref={\roman*}]
    \item the Hamiltonian itself, $\hat{H}$; and
    \item the total fermion number,
    \begin{align}
    \hat{N} = \sum_{i=1}^L \hat{n}_i.
    \end{align}
\end{enumerate}
\end{theorem}

Together with the corresponding Holstein--Hubbard result proved later in Theorem~\ref{t:HHwithtrivial}, this classification shows that nontrivial local conserved quantities are absent in representative fermion--boson coupled systems, a class of models for which such rigorous results have rarely been available~\cite{LowenHolstein,FreericksLieb,MiyaoHolsteinHubbard,MiyaoSSH,MiyaoHolsteinHubbard2,Miyaolongrange,MiyaoNagaoka,MiyaoHolsteinHubbard3,MiyaoCDW,Miyaoetal,ForLieb}.

\section{Proof strategy}\label{sec:proof_strategy}
\subsection{Proof outline}
The proof of the absence of local conserved quantities proceeds as follows.

To begin with, we expand a $k$-local quantity in a suitable operator basis as
\begin{equation}
\label{expandinput}
\hat{Q}^k= \sum_{l=1}^k \sum_{i=1}^L\sum_{\mathbf{A}^l_i} q_{\mathbf{A}^l_i} \mathbf{A}^l_i ,   
\end{equation}
with coefficients $q_{\mathbf{A}^l_i} \in \mathbb{C}$. Here the sum over $\mathbf{A}^l$ runs over all possible sequences of basis operators of length $l$. Since the commutator of a $k$-local quantity $\hat{Q}^k$ with $\hat{H}$ has support on at most $(k+1)$ sites, we can also expand $[\hat{Q}^k, \hat{H}]$ as
\begin{equation}
\label{expandoutput}
[\hat{Q}^k,\hat{H}]= \sum_{l=0}^{k+1} \sum_{i=1}^L \sum_{\mathbf{B}^l_i} r_{\mathbf{B}^l_i} \mathbf{B}^l_i.
\end{equation}
We refer to a term $\mathbf{A}$ appearing in $\hat{Q}^k$ as an \textit{input} operator, and to a term $\mathbf{B}$ appearing in $[\hat{Q}^k,\hat{H}]$ as an \textit{output} operator.
Each coefficient $r_{\mathbf{B}^l}$ can be expressed as a linear combination of the coefficients $q_{\mathbf{A}}$. The conservation law $[\hat{Q}^k,\hat{H}]=0$ therefore implies that $r_{\mathbf{B}^l}=0$ for every $\mathbf{B}^l$, yielding a large set of linear constraints on the coefficients $q_{\mathbf{A}}$.
If these linear relations admit no solution other than $q_{\mathbf{A}^k}=0$ for all $k$-support inputs $\mathbf{A}^k$, then $\hat{Q}^k$ cannot be a $k$-local conserved quantity.

\subsection{Basis and commutation relations}
In this section, we introduce a convenient basis of operators that describes local conserved quantities of fermion-boson coupling systems.

For bosonic and fermionic degrees of freedom, operators are constructed using the creation/annihilation operators $\hat{b}^\dagger / \hat{b}$ and $\hat{c}^\dagger / \hat{c}$. 
Due to canonical commutation relations $[\hat{b},\hat{b}^\dagger]=1$, single site operator of boson systems can be expanded as a linear combination of $(\hat{b}^\dagger)^x\hat{b}^y$ where $x,y$ are non-negative integers. For example, $\hat{b}^\dagger\hat{b}\hat{b}^\dagger\hat{b}$ can be expanded to $(\hat{b}^\dagger)^2\hat{b}^2+\hat{b}^\dagger\hat{b}$.
For a single fermionic degree of freedom, the local operator algebra is finite dimensional. A single-site fermionic operator can be written as a linear combination of $I$, $\hat{c}$, $\hat{c}^\dagger$, and $\hat{n}=\hat{c}^\dagger\hat{c}$, where $I$ is the identity operator.
Fermion and boson creation/annihilation operators commute with each other, so products of these fermionic and bosonic single-site basis operators form a basis of observables on fermion-boson coupled chains.

Throughout this paper, we write an $l$-support input basis operator starting from the $i$-th site as
\begin{align}
\mathbf{A}^l_i&=  e_i e_{i+1}\dots e_{i+l-1},
\label{spinlessbasis_1}
\end{align}
\begin{align}
e_j \in \left\{ f_j(\hat{b}^\dagger_j)^x\hat{b}_j^y \mid
f_j\in\{I,\hat{c}_j,\hat{c}^\dagger_j,\hat{n}_j\},\,
x,y\in\mathbb{Z}_{\geq 0}\right\},
\label{spinlessbasis_2}
\end{align}
with $e_i,e_{i+l-1}\neq I(\hat{b}^\dagger_j)^0\hat{b}_j^0$. Output operators $\mathbf{B}^l_i$ are expanded in the same basis.

Also, we denote by $f(x,y)_j$ the pair of fermion operator $f_j \in \{ \hat{c}_j, \hat{c}^\dagger_
j, \hat{n}_j, I\}$ and boson operator $(\hat{b}^\dagger_j)^x\hat{b}^y_j$. 
When an operator acts trivially on bosons (i.e. when $x=y=0$), we often denote $f(0,0)$ simply as $f$.
Likewise, when an operator acts trivially on both the fermionic and bosonic degrees of freedom, we often omit the corresponding $I(0,0)=I$ whenever no confusion arises.

To evaluate the commutator $[\mathbf{A}^l_i, \hat{H}]$, we use the following commutation relations.
The fermion operators satisfy the \textit{anticommutation relations}
\begin{align}
    \{ \hat{c}^\dagger_i , \hat{c}_j \} &= \delta _ { i  j }, \label{anticommutation1}
\end{align}
and
\begin{align}
	\{ \hat{c}^\dagger_i , \hat{c}^\dagger_j \} = \{ \hat{c}_i , \hat{c}_j \} = 0 , \label{anticommutation2}
\end{align}
for any $ i ,j=1,2,\dots, L $, where $ \{ \hat{A} , \hat{B} \} := \hat{A} \hat{B} + \hat{B} \hat{A} $ and $\delta_{ij}$ is the Kronecker delta.

From the anticommutation relations of the fermionic operators above, we obtain the following useful commutation formulas, which will be used repeatedly in the subsequent analysis. For the fermionic part, for any $i\neq j$, we have

\begin{align}
[\hat{c}_i, \hat{c}^\dagger_i \hat{c}_j] &= \hat{c}_j ,
\label{commutator_c}\\
[\hat{c}_j^\dagger, \hat{c}^\dagger_i \hat{c}_j] &= -\hat{c}_i^\dagger,
\label{commutator_cdagger}\\
[\hat{n}_i, \hat{c}^\dagger_i \hat{c}_j] &= \hat{c}^\dagger_i \hat{c}_j,
\label{commutator_nc}\\
[\hat{n}_j, \hat{c}^\dagger_i \hat{c}_j] &= -\hat{c}^\dagger_i \hat{c}_j,
\label{commutator_ncdagger}\\
[\hat{c}_i, \hat{n}_i] &= \hat{c}_i ,
\label{commutator_cn}\\
[\hat{c}_i^\dagger, \hat{n}_i] &= -\hat{c}^\dagger_i ,
\label{commutator_cdaggern}
\end{align}
For the bosonic part, we only need commutators in which one of the two bosonic monomials comes from the Hamiltonian. Since the Holstein Hamiltonian~\eqref{Eq.1D_H_model} contains only $\hat{b}$, $\hat{b}^\dagger$, and $\hat{b}^\dagger\hat{b}$, this means that the corresponding powers satisfy $u,v\in\{0,1\}$ in all applications below. Thus, for $x,y\in\mathbb{Z}_{\geq0}$ and $u,v\in\{0,1\}$,
\begin{align}
[(\hat{b}^\dagger)^x_i\hat{b}^y_i,(\hat{b}^\dagger)^u_j\hat{b}^v_j]=(yu-xv)(\hat{b}^\dagger)^{x+u-1}_i\hat{b}^{y+v-1}_i\delta_{ij},
\label{commutator_boson}
\end{align}
Finally, fermionic and bosonic operators commute with each other. Hence, for $f_j\in \{\hat{c}^\dagger_j, \hat{c}_j,\hat{n}_j, I \}$,
\begin{align}
[(\hat{b}^\dagger)^x_i\hat{b}^y_i,f_j]=0.
\label{commutator_bosonspin}
\end{align}

The key observation is that the system has fermion parity symmetry, i.e., each term in the Hamiltonian contains an even number of fermion creation and annihilation operators. This observation partially restores the commutativity with other sites assumed in the proof of nonintegrability for quantum spin systems~\cite{XYZh,HokkyoNI}. That is, let $h_i$ be one of the terms in the Hamiltonian $\hat{H}$, let $A_i$ be an operator whose support overlaps with that of $h_i$, and let $B_j$ be an operator whose support does not overlap with that of $h_i$. Then
\begin{align}
[A_i B_j,h_i]&=[A_i,h_i]B_j
\label{fermioncommutation}
\end{align}
holds. For example, $[\hat{c}_i \hat{c}_{i+3}, \hat{c}^\dagger_i\hat{c}_{i+1}]
=[\hat{c}_i , \hat{c}^\dagger_i\hat{c}_{i+1}]\hat{c}_{i+3}$ holds. 
This property was implicitly used by Futami~\cite{2DHubbard} who discussed the Hubbard model on high-dimensional hypercubic lattices, but it generally holds for any system with fermion parity symmetry.

\section{Proof of Theorem~\ref{nontrivial}}\label{sec:proof}
In this section, we prove Theorem~\ref{nontrivial} given in Sec.~\ref{sec:main_result}. In the following, we assume that $3 \leq k < L/2$.

The proof consists of three steps: In step 1, we focus on the conditions $\{r_{\mathbf{B}_j^{k+1}} = 0\}$, and derive that most of $\mathbf{A}^k_i$ have zero coefficients.
In step~2, we turn to the conditions $\{ r_{\mathbf{B}_j^{k}} = 0\}$, and derive that some of the remaining $\mathbf{A}^k_i$ and most of $\mathbf{A}^{k-1}_i$ also have zero coefficients. 
In step 3, we turn to the conditions $\{ r_{\mathbf{B}_j^{k-1}} = 0\}$ and use the constraints on the associated $(k-1)$-support inputs $\mathbf{A}^{k-1}_i$ to show that the remaining coefficients of $\mathbf{A}^k_i$ also vanish.

Now we introduce some terminology to classify the inputs. Take an arbitrary $k$-support input $\mathbf{A}^k_i$. By focusing on its ends $e_i,e_{i+k-1}$, it can be classified into one of the following three types:
\begin{enumerate}[label={\roman*.},ref={\roman*},align=right,leftmargin=!]
    \item \label{enum:input_1}An input 
    that includes $\hat{c}^\dagger$, $\hat{c}$, or $\hat{n}$ on both ends
    \item \label{enum:input_2}An input 
    that includes $\hat{c}^\dagger$, $\hat{c}$, or $\hat{n}$ on one end but not on the other
    \item \label{enum:input_3}An input 
    that includes none of $\hat{c}^\dagger$, $\hat{c}$, or $\hat{n}$ on both ends
\end{enumerate}

\paragraph*{A brief illustration of the coefficient-comparison strategy.}
Before beginning the detailed proof, we illustrate the two basic ways in which
the conditions $r_{\mathbf B}=0$ constrain an input coefficient.  First,
consider the type~\ref{enum:input_2} input
$\mathbf A_i^k=\hat b_i\hat c_{i+k-1}$, where identity factors at the
intermediate sites are suppressed.  Its commutator with the hopping term
across the right boundary is
\begin{align}
[\hat b_i\hat c_{i+k-1},
 \hat c^\dagger_{i+k-1}\hat c_{i+k}]
 &=\hat b_i I_{i+k-1}\hat c_{i+k}.
\end{align}
We denote the $(k+1)$-support operator on the right-hand side by
$\mathbf B_i^{k+1}$.  This output has no competing $k$-support input: in
particular, an input supported on $[i+1,i+k]$ cannot reproduce the left-end
phonon operator $\hat b_i$ through a boundary hopping commutator.  Hence
\begin{align}
r_{\mathbf B_i^{k+1}}=tq_{\hat b_i\hat c_{i+k-1}}=0,
\end{align}
and the input coefficient vanishes.

By contrast, two type~\ref{enum:input_1} inputs may generate the same output.
For example,
\begin{align}
[\hat c^\dagger_i\hat c_{i+k-1},
 \hat c^\dagger_{i+k-1}\hat c_{i+k}]
 &=\hat c^\dagger_i\hat c_{i+k},\\
[\hat c^\dagger_{i+1}\hat c_{i+k},
 \hat c^\dagger_i\hat c_{i+1}]
 &=-\hat c^\dagger_i\hat c_{i+k}.
\end{align}
Denoting their common output by
$\mathbf B_i^{k+1}=\hat c^\dagger_i\hat c_{i+k}$, the condition
$r_{\mathbf B_i^{k+1}}=0$ instead gives
\begin{align}
q_{\hat c^\dagger_i\hat c_{i+k-1}}
=q_{\hat c^\dagger_{i+1}\hat c_{i+k}}.
\end{align}
Thus, a uniquely generated output eliminates an input coefficient, whereas a
shared output yields a relation among its contributing coefficients.  The
proof below repeatedly uses these two mechanisms for outputs of support
$k+1$, $k$, and $k-1$.  The same coefficient-comparison framework also
applies to $k=1,2$; in those cases, however, the surviving coefficients
reconstruct the total fermion number $\hat N$ and the Hamiltonian $\hat H$,
as shown in Sec.~\ref{sec:short}.

Before entering the detailed argument, Table~\ref{tab:proof_overview} and
Fig.~\ref{fig:proof_roadmap} summarize the structure of the proof of
Theorem~\ref{nontrivial}.  The table organizes the input classes according to
where they are treated, while the flowchart displays the logical dependence
among the three proof steps and the auxiliary $(k-1)$-support analysis.

\begin{table*}[t]
\caption{Candidate-oriented overview of the coefficient-elimination proof of
Theorem~\ref{nontrivial}.  The first block lists the $k$-support candidates,
while the second lists the auxiliary $(k-1)$-support inputs used to eliminate
the last candidates.  Identity factors at unshown intermediate sites are
suppressed.}
\label{tab:proof_overview}
\small
\centering
\begin{tabular}{p{0.12\textwidth}p{0.07\textwidth}p{0.23\textwidth}p{0.29\textwidth}p{0.23\textwidth}}
\toprule
Input support & Type & Detailed class & Representative input(s) & Treated in
\\ \midrule
\multirow{6}{*}{$k$-support}
& \multirow{3}{*}{\ref{enum:input_1}}
& Particle-nonconserving bare endpoint bilinears ($cc$ type)
& $\hat c_i\hat c_{i+k-1}$,
$\hat c_i^\dagger\hat c_{i+k-1}^\dagger$
& Step 1 (Lemma~\ref{lem:step1_1}): restricted.\newline
Step 2 (Lemma~\ref{lem:step2_1}): eliminated. \\
\cmidrule(lr){3-5}
&
& Particle-conserving bare endpoint bilinears ($c^\dagger c$ type)
& $\hat c_i^\dagger\hat c_{i+k-1}$,
$\hat c_i\hat c_{i+k-1}^\dagger$
& Step 1 (Lemma~\ref{lem:step1_1}): restricted.\newline
Step 2 (Lemma~\ref{lem:step2_1}): linked to $(k-1)$-support inputs.\newline
Step 3 (Lemmas~\ref{lem:step3_k3} and \ref{lem:step3_kgeneral}): eliminated using
Lemma~\ref{lem:constraints}. \\
\cmidrule(lr){3-5}
&
& All remaining type-\ref{enum:input_1} structures
& $\hat c_i^\dagger\hat b_{i+1}\hat c_{i+k-1}$,
$\hat n_i\hat c_{i+k-1}$
& Step 1 (Lemma~\ref{lem:step1_1}): eliminated. \\
\cmidrule(lr){2-5}
& \ref{enum:input_2}
& A fermionic operator on exactly one end
& $\hat b_i\hat c_{i+k-1}$
& Step 1 (Lemma~\ref{lem:step1_2}): eliminated. \\
\cmidrule(lr){2-5}
& \multirow{2}{*}{\ref{enum:input_3}}
& A fermionic or number factor at an intermediate site
& $\hat b_i\hat n_{i+1}\hat b^\dagger_{i+k-1}$
& Step 2 (Lemma~\ref{lem:step2_3}): reduced to purely bosonic products. \\
\cmidrule(lr){3-5}
&
& Purely bosonic products
& $\hat b_i\hat b^\dagger_{i+k-1}$,
$\hat b_i^\dagger\hat b_{i+1}\hat b_{i+k-1}$
& Step 2 (Lemma~\ref{lem:step2_3}): eliminated using the phonon-number term. \\
\midrule
\multirow{2}{*}{$(k-1)$-support}
&
& Bare particle-conserving endpoint products
& $\hat c_i^\dagger\hat c_{i+k-2}$,
$\hat c_i\hat c_{i+k-2}^\dagger$
& Step 2 (Lemma~\ref{lem:constraints}): restricted.\newline
Step 3 (Lemmas~\ref{lem:step3_k3} and \ref{lem:step3_kgeneral}): used in the final
elimination. \\
\cmidrule(lr){3-5}
&
& Particle-conserving endpoint products dressed by a single phonon operator
& $\hat c_i^\dagger\hat c_{i+k-2}\hat b_{i+k-2}$,
$\hat c_i^\dagger\hat b_i\hat c_{i+k-2}$
& Step 2 (Lemmas~\ref{lem:step2_1} and \ref{lem:constraints}): linked to the
$k$-support candidates.\newline
Step 3 (Lemmas~\ref{lem:step3_k3} and \ref{lem:step3_kgeneral}): used to close the
final system.
\\
\bottomrule
\end{tabular}
\end{table*}

\begin{figure*}[p]
\centering
\resizebox{0.96\textwidth}{!}{%
\begin{tikzpicture}[
  font=\footnotesize,
  >={Latex[length=2.1mm,width=1.5mm]},
  line/.style={draw,thick,-{Latex[length=2.1mm,width=1.5mm]}},
  dependency/.style={line,densely dashed},
  box/.style={draw,rounded corners=2pt,thick,align=center,
    inner sep=5pt,minimum height=9mm},
  framework/.style={box,fill=gray!14,text width=10.8cm},
  classify/.style={box,fill=gray!7,text width=7.4cm},
  type/.style={box,text width=4.95cm,minimum height=13mm},
  stagebox/.style={box,fill=gray!6,text width=4.15cm,minimum height=16mm},
  auxiliary/.style={box,densely dashed,fill=gray!10,text width=5.8cm,
    minimum height=17mm},
  conclusion/.style={box,fill=gray!18,text width=9.5cm,minimum height=17mm},
  typeframe/.style={draw=black!55,thick}
]

\node[framework] (framework) {
  \textbf{Input--output coefficient comparison}\\[-1pt]
  $\displaystyle
  \hat Q^k=\sum_{\mathbf A}q_{\mathbf A}\mathbf A,
  \qquad
  [\hat Q^k,\hat H]=\sum_{\mathbf B}r_{\mathbf B}\mathbf B,
  \qquad r_{\mathbf B}=0\ \text{for every }\mathbf B$
};

\node[classify,below=8mm of framework] (classify) {
  Assume that $\hat Q^k$ is $k$-local, $3\leq k<L/2$,\\
  and classify its $k$-support inputs by their two endpoints
};

\node[type,below left=10mm and 9mm of classify] (typei) {
  \textbf{Type \ref{enum:input_1}}\\
  Fermionic structure at both ends
};
\node[type,below=10mm of classify] (typeii) {
  \textbf{Type \ref{enum:input_2}}\\
  Fermionic structure at exactly one end
};
\node[type,below right=10mm and 9mm of classify] (typeiii) {
  \textbf{Type \ref{enum:input_3}}\\
  No fermionic structure at either end
};

\node[stagebox,below=7mm of typei] (step1i) {
  \textbf{Step 1: $r_{\mathbf B^{k+1}}=0$}\\
  Lemma \ref{lem:step1_1}\\
  Only four bare endpoint bilinears remain; their coefficients obey shift relations
};
\node[stagebox,below=7mm of typeii] (step1ii) {
  \textbf{Step 1: $r_{\mathbf B^{k+1}}=0$}\\
  Lemma \ref{lem:step1_2}\\
  Every type-\ref{enum:input_2} coefficient vanishes
};

\node[stagebox,below=8mm of step1i] (step2i) {
  \textbf{Step 2: $r_{\mathbf B^{k}}=0$}\\
  Lemma \ref{lem:step2_1}\\
  Particle-nonconserving candidates vanish; particle-conserving candidates are linked to $(k-1)$-support coefficients
};
\node[stagebox] (step2iii) at (step2i -| typeiii) {
  \textbf{Step 2: $r_{\mathbf B^{k}}=0$}\\
  Lemma \ref{lem:step2_3}\\
  Every type-\ref{enum:input_3} coefficient vanishes
};

\node[auxiliary] (auxiliary) at (step2i -| typeii) {
  \textbf{Separate auxiliary analysis in Step 2}\\
  \textit{Selected $(k-1)$-support inputs only}\\
  Lemma \ref{lem:constraints}: derive the translation and endpoint-transfer
  relations needed in Step 3
};

\node[stagebox,below=8mm of step2i] (step3i) {
  \textbf{Step 3: $r_{\mathbf B^{k-1}}=0$}\\
  Lemmas \ref{lem:step3_k3} and \ref{lem:step3_kgeneral}\\
  The remaining type-\ref{enum:input_1} coefficients vanish using the auxiliary
  $(k-1)$-support relations
};

\coordinate (conclusionanchor) at (classify |- step3i.south);
\node[conclusion,below=13mm of conclusionanchor] (conclusion) {
  \textbf{All $k$-support coefficients vanish}\\
  Hence $\hat Q^k$ can be expressed using support at most $k-1$,
  contradicting the assumption that it is $k$-local\\[-1pt]
  $\Longrightarrow$ no $k$-local conserved quantity for $3\leq k<L/2$
};

\draw[line] (framework) -- (classify);
\draw[line] (classify) -- (typei);
\draw[line] (classify) -- (typeii);
\draw[line] (classify) -- (typeiii);
\draw[dependency] (step1i) -- (step2i);
\draw[dependency] (step2i) -- (step3i);
\draw[dependency] (step1i.south east) -- (auxiliary.north west);
\draw[dependency] ([xshift=-3mm]step1ii.south) --
  ([xshift=-3mm]auxiliary.north);
\draw[dependency] (step1ii.east) -- (step2iii.west);
\draw[dependency] (auxiliary.south west) -- (step3i.north east);

\coordinate (typeibottom) at ([yshift=-4mm]step3i.south);
\coordinate (typeibl) at (typei.south west |- typeibottom);
\coordinate (typeibr) at (typei.south east |- typeibottom);
\draw[typeframe,rounded corners=2pt]
  (typei.south west) -- (typeibl) -- (typeibr) -- (typei.south east);

\coordinate (typeiibottom) at ([yshift=-4mm]step1ii.south);
\coordinate (typeiibl) at (typeii.south west |- typeiibottom);
\coordinate (typeiibr) at (typeii.south east |- typeiibottom);
\draw[typeframe,rounded corners=2pt]
  (typeii.south west) -- (typeiibl) -- (typeiibr) -- (typeii.south east);

\coordinate (typeiiibottom) at ([yshift=-4mm]step2iii.south);
\coordinate (typeiiibl) at (typeiii.south west |- typeiiibottom);
\coordinate (typeiiibr) at (typeiii.south east |- typeiiibottom);
\draw[typeframe,rounded corners=2pt]
  (typeiii.south west) -- (typeiiibl) -- (typeiiibr) -- (typeiii.south east);

\coordinate (typeiout) at ($(typeibl)!0.5!(typeibr)$);
\coordinate (typeiiout) at ($(typeiibl)!0.5!(typeiibr)$);
\coordinate (typeiiiout) at ($(typeiiibl)!0.5!(typeiiibr)$);
\begin{scope}[on background layer]
  \draw[line] (typeiout) -- (conclusion.north west);
  \draw[line] (typeiiout) -- (conclusion.north);
  \draw[line] (typeiiiout) -- (conclusion.north east);
\end{scope}

\end{tikzpicture}%
}
\caption{Roadmap for the proof of Theorem~\ref{nontrivial}.  Solid arrows
organize the initial classification and connect the three completed type
analyses to the final conclusion, whereas dashed arrows indicate that an
earlier result is used to derive a later restriction.  The three solid outer
frames show how far the analysis of each input type continues before that
type is completely eliminated.  The dashed box is a separate, auxiliary
analysis of selected $(k-1)$-support inputs needed to close Step 3.}
\label{fig:proof_roadmap}
\end{figure*}

\clearpage

\subsection{Proof step 1: Basic relations for products with width $k + 1$.}

Using the condition that $r_{\mathbf{B}_i^{k+1}} = 0$ for all $\mathbf{B}_i^{k+1}$, we narrow down the candidates for $\mathbf{A}_i^k$ whose coefficients $q_{\mathbf{A}_i^k}$ may be nonzero. The key observation is that a $(k+1)$-support output $\mathbf{B}_i^{k+1}$ can be generated only by $\mathbf{A}_j^k$ from $\hat{Q}^k$ together with the hopping terms in $\hat{H}$. Using this fact, we derive simple constraints on coefficients of $\mathbf{A}_j^k$. More specifically, in the following three subsections we show that the possible forms of type~\ref{enum:input_1} inputs are severely restricted, that all type~\ref{enum:input_2} inputs are excluded, and that step 1 alone does not impose a useful restriction on type~\ref{enum:input_3} inputs.

\subsubsection{Inputs of type~\ref{enum:input_1}}

In this section, we treat $k$-support inputs $\mathbf{A}^k$ of type~\ref{enum:input_1}. 
The conclusion of step 1 analysis for this type is summarized in the following Lemma:
\begin{lemma}[step 1 analysis for type~\ref{enum:input_1}]
Assume $\hat{Q}^k$ is a $k$-local conserved quantity of the one-dimensional Holstein model~\eqref{Eq.1D_H_model} with $t\neq 0$, $g\neq 0$, and $\omega \neq0$. 
The coefficients of $\mathbf{A}^k_i$ of type~\ref{enum:input_1} satisfy the following relation for any $i$:
\begin{align}
q_{\hat{c}^\dagger_i 
\hat{c}_{i+k-1}}&=
q_{\hat{c}^\dagger_{i+1} 
\hat{c}_{i+k}},\\
q_{\hat{c}_i \hat{c}^\dagger_{i+k-1} 
}&=
q_{\hat{c}_{i+1} \hat{c}^\dagger_{i+k} 
},\\
q_{\hat{c}^\dagger_i 
\hat{c}^\dagger_{i+k-1}}&=
-q_{\hat{c}^\dagger_{i+1} 
\hat{c}^\dagger_{i+k}},\\
q_{\hat{c}_i 
\hat{c}_{i+k-1}}&=
-q_{\hat{c}_{i+1} 
\hat{c}_{i+k}},\\
q_{\mathbf{A}^k_i}&=0. \quad (\textnormal{otherwise})
\end{align}
\label{lem:step1_1}
\end{lemma}

\begin{proof}[Proof of Lemma \ref{lem:step1_1}]
First, we show that these $\mathbf{A}^k$ have zero coefficients unless both endpoints are $\hat{c}^\dagger(0,0),\hat{c}(0,0)$. 
Take a $k$-support input $\mathbf{A}^k_i=e^1_ie^2_{i+1}\cdots e^k_{i+k-1}$ and assume $e^1\neq \hat{c}^\dagger(0,0),\hat{c}(0,0)$. 
We first consider the case $e^k=\hat{c}(x_k,y_k)$.  Commuting this input with
the hopping term across its right boundary gives
\begin{align}
&[\big( \prod^{k-1}_{j=1} e^j_{i+j-1} \big) 
\hat{c}(x_k,y_k)_{i+k-1},\hat{c}^\dagger_{i+k-1} \hat{c}_{i+k}]\nonumber \\
=&\big( \prod^{k-1}_{j=1} e^j_{i+j-1} \big) 
I(x_k,y_k)_{i+k-1}\hat{c}(0,0)_{i+k},    
\end{align}
The operator on the right-hand side has $(k+1)$-support; we denote it by
$\mathbf B_i^{k+1}$.  There is no other
$k$-support input $\mathbf{A}^k$ whose commutator with $\hat{H}$ generates
this output.  Then,
\begin{align}
    r_{\mathbf{B}^{k+1}_i}=t q_{\mathbf{A}^k_i}=0
\end{align}
holds for all $i$. 
The cases $e^k=\hat{n}(x_k,y_k)$ and $e^k=\hat{c}^\dagger(x_k,y_k)$ can be treated similarly by choosing the outputs
\begin{align*}
\big( \prod^{k-1}_{j=1} e^j_{i+j-1} \big) \hat{c}(x_k,y_k)_{i+k-1}\hat{c}(0,0)_{i+k}
\end{align*}
and
\begin{align*}
\big( \prod^{k-1}_{j=1} e^j_{i+j-1} \big) I(x_k,y_k)_{i+k-1}\hat{c}^\dagger(0,0)_{i+k},
\end{align*}
respectively; an analogous uniqueness argument gives $q_{\mathbf{A}^k_i}=0$ in those cases as well.
Since $t \neq 0$, $\mathbf{A}^k$ has a zero coefficient when $e^1\neq \hat{c}^\dagger(0,0),\hat{c}(0,0)$. An analogous argument can be made for the other endpoint $e^k$, so type~\ref{enum:input_1} inputs $\mathbf{A}^k$ have a zero coefficient unless both endpoints are $\hat{c}^\dagger(0,0),\hat{c}(0,0)$.

Next, we show that type~\ref{enum:input_1} inputs $\mathbf{A}^k$ have zero coefficients unless middle terms $e^2,\dots,e^{k-1}$ are $I(0,0)$. Take a $k$-support input
\begin{align*}
\mathbf{A}^k_i=\hat{c}^\dagger_i \big( \prod^{k-1}_{j=2} e^j_{i+j-1} \big)\hat{c}_{i+k-1}.
\end{align*}
Its commutator with the hopping term across the right boundary is
\begin{align}
&[\hat{c}^\dagger_i\big( \prod^{k-1}_{j=2} e^j_{i+j-1} \big) 
\hat{c}_{i+k-1},\hat{c}^\dagger_{i+k-1} \hat{c}_{i+k}] \nonumber \\
=&\hat{c}^\dagger_i\big( \prod^{k-1}_{j=2} e^j_{i+j-1} \big) 
I_{i+k-1}\hat{c}_{i+k}.  
\end{align}
We denote the $(k+1)$-support operator on the right-hand side by
$\mathbf{B}^{k+1}_i$, since here we must compare all inputs that generate the
same output.
We now ask whether the same output can be generated from any other input whose coefficient has not already been shown to vanish. By the first part of the proof, such an input must have $\hat{c}^\dagger$ or $\hat{c}$ at both ends. Since $\mathbf{B}^{k+1}_i$ has support from $i$ to $i+k$, the only possible competing $k$-support input with nonzero coefficient must be supported on the interval $[i+1,i+k]$ and must use the hopping term $\hat{c}^\dagger_i\hat{c}_{i+1}$ to create the left endpoint. This is possible only if $e^2=I$, in which case the competing input is
\begin{align*}
\hat{c}^\dagger_{i+1}\big( \prod^{k-1}_{j=3} e^j_{i+j-1} \big)I_{i+k-1}\hat{c}_{i+k}.
\end{align*}
Indeed, it gives the same output through
\begin{align}
&[\hat{c}^\dagger_{i+1}\big( \prod^{k-1}_{j=3} e^j_{i+j-1} \big) 
I_{i+k-1}\hat{c}_{i+k},\hat{c}^\dagger_{i} \hat{c}_{i+1}] \nonumber \\
=&-\hat{c}^\dagger_iI_{i+1}\big( \prod^{k-1}_{j=3} e^j_{i+j-1} \big) 
I_{i+k-1}\hat{c}_{i+k}.
\end{align}
If $e^2\neq I$, no such competing input exists; all other possible inputs have already been excluded by the end-point analysis above. Hence $\mathbf{B}^{k+1}_i$ is generated uniquely by $\mathbf{A}^k_i$, and the condition $r_{\mathbf{B}^{k+1}_i}=0$ implies $q_{\mathbf{A}^k_i}=0$. Thus, this $\mathbf{A}^k_i$ can have a nonzero coefficient only when $e^2=I$. The remaining cases $(e^1,e^k)=(\hat{c},\hat{c}^\dagger),(\hat{c}^\dagger,\hat{c}^\dagger),(\hat{c},\hat{c})$ follow similarly. 
By repeating similar arguments for $e^3,\dots,e^{k-1}$, we find that $q_{\mathbf{A}^k_i}=0$ holds for all $i$
except for $\mathbf{A}^k$ such that middle terms $e^2,\dots,e^{k-1}$ are identity $I$. 
Therefore, $\mathbf{A}^k_i=\hat{c}^\dagger_i \hat{c}_{i+k-1},\hat{c}_i \hat{c}^\dagger_{i+k-1},\hat{c}^\dagger_i \hat{c}^\dagger_{i+k-1}$, and $\hat{c}_i \hat{c}_{i+k-1}$ remains as candidates for non-zero coefficient $\mathbf{A}^k$ of type~\ref{enum:input_1}.

Finally, we derive relations among the coefficients of the type~\ref{enum:input_1} candidates $\mathbf{A}^k$.  The two boundary commutators
\begin{align}
[\hat{c}^\dagger_i 
\hat{c}_{i+k-1},\hat{c}^\dagger_{i+k-1} \hat{c}_{i+k}]&=\hat{c}^\dagger_i\hat{c}_{i+k}, \\   
[\hat{c}^\dagger_{i+1} 
\hat{c}_{i+k},\hat{c}^\dagger_{i} \hat{c}_{i+1}]&= -\hat{c}^\dagger_i\hat{c}_{i+k},    
\end{align}
generate the same $(k+1)$-support output, which we denote locally by
$\mathbf{B}^{k+1}_i=\hat{c}^\dagger_i\hat{c}_{i+k}$.  Its conservation
condition leads to
\begin{align}
q_{\hat{c}^\dagger_i 
\hat{c}_{i+k-1}}
-
q_{\hat{c}^\dagger_{i+1} 
\hat{c}_{i+k}}&=0.
\end{align}
Similarly, by considering the commutators that contribute to $\mathbf{B}^{k+1}_i=\hat{c}_i\hat{c}^\dagger_{i+k} ,\hat{c}^\dagger_i \hat{c}^\dagger_{i+k},\hat{c}_i \hat{c}_{i+k}$, we obtain $q_{\hat{c}^\dagger_i 
\hat{c}_{i+k-1}}=
q_{\hat{c}^\dagger_{i+1} 
\hat{c}_{i+k}},
q_{\hat{c}_i \hat{c}^\dagger_{i+k-1} 
}=
q_{\hat{c}_{i+1} \hat{c}^\dagger_{i+k} 
},
q_{\hat{c}^\dagger_i 
\hat{c}^\dagger_{i+k-1}}=
-q_{\hat{c}^\dagger_{i+1} 
\hat{c}^\dagger_{i+k}},
q_{\hat{c}_i 
\hat{c}_{i+k-1}}=
-q_{\hat{c}_{i+1} 
\hat{c}_{i+k}}$. 
\end{proof}

From Lemma~\ref{lem:step1_1}, for $\mathbf{A}^k_i=\hat{c}^\dagger_i\hat{c}_{i+k-1}, \hat{c}_i\hat{c}_{i+k-1}^\dagger, \hat{c}_i\hat{c}_{i+k-1}$, and $\hat{c}^{\dagger}_i\hat{c}^{\dagger}_{i+k-1}$,
\begin{align}
q_{\hat{c}^\dagger_{i} 
\hat{c}_{i+k-1}}&=
q_{\hat{c}^\dagger_{i+k-1}\hat{c}_{i+2k-2}},
\label{cdc_k1}\\
q_{\hat{c}_{i} 
\hat{c}^\dagger_{i+k-1}}&=
q_{\hat{c}_{i+k-1}\hat{c}^\dagger_{i+2k-2}},
\end{align}
\begin{align}
q_{\hat{c}_{i} 
\hat{c}_{i+k-1}}&=
\begin{cases}
    q_{\hat{c}_{i+k-1}\hat{c}_{i+2k-2}}, & (k\colon\textnormal{odd});\\
    -q_{\hat{c}_{i+k-1}\hat{c}_{i+2k-2}}, & (k\colon\textnormal{even})
\end{cases}
\label{cc_k1}\\
q_{\hat{c}^{\dagger}_{i} 
\hat{c}^{\dagger}_{i+k-1}}&=
\begin{cases}
    q_{\hat{c}^{\dagger}_{i+k-1}\hat{c}^{\dagger}_{i+2k-2}}, & (k\colon\textnormal{odd});\\
    -q_{\hat{c}^{\dagger}_{i+k-1}\hat{c}^{\dagger}_{i+2k-2}}, & (k\colon\textnormal{even})
\end{cases}
\end{align}
hold respectively.

\subsubsection{Inputs of type~\ref{enum:input_2}}
Next, we treat $k$-support inputs $\mathbf{A}^k_i$ of type~\ref{enum:input_2}.
The conclusion of step 1 analysis for this type is summarized in the following Lemma:
\begin{lemma}[Step 1 analysis for type~\ref{enum:input_2}]
Assume $\hat{Q}^k$ is a $k$-local conserved quantity of the one-dimensional Holstein model~\eqref{Eq.1D_H_model} with $t\neq 0$, $g\neq 0$, and $\omega \neq0$.
$q_{\mathbf{A}^k_i}=0$ holds for any $\mathbf{A}^k_i$ of type~\ref{enum:input_2}.
\label{lem:step1_2}
\end{lemma}

\begin{proof}
Owing to the inversion symmetry of the Hamiltonian, it is sufficient to consider the cases $\mathbf{A}_i^{k} = I(x_1,y_1)_i \big( \prod^{k-1}_{j=2} e^j_{i+j-1} \big) \hat{A}(x_k,y_k)_{i+k-1}$ with $\hat{A}=\hat{c},\hat{n}, \hat{c}^\dagger$. 
If $\hat{A}=\hat{c}$, commuting the input with the hopping term across its
right boundary gives
\begin{align}
&[I(x_1,y_1)_i \big( \prod^{k-1}_{j=2} e^j_{i+j-1} \big)
\hat c(x_k,y_k)_{i+k-1},
\hat c^\dagger_{i+k-1}\hat c_{i+k}]\nonumber\\
&\quad=I(x_1,y_1)_i \big( \prod^{k-1}_{j=2} e^j_{i+j-1} \big)
I(x_k,y_k)_{i+k-1}\hat c_{i+k}.
\end{align}
We denote the $(k+1)$-support operator on the right-hand side by
$\mathbf{B}^{k+1}_i$.  No other input of width at most $k$ generates it.  Then,
\begin{align}
r_{\mathbf{B}^{k+1}_i}=t q_{\mathbf{A}^k_i}=0
\end{align}
holds.
The remaining cases follow similarly.
Therefore, type~\ref{enum:input_2} inputs $\mathbf{A}^k$ have zero coefficients.
\end{proof}

\subsubsection{Inputs of type~\ref{enum:input_3}}
Finally, we examine $k$-support inputs $\mathbf{A}^k_i$ of type~\ref{enum:input_3}. 
None of this type of input generates $\mathbf{B}^{k+1}$. 
Then, we cannot obtain any valid relation for type~\ref{enum:input_3} inputs $\mathbf{A}^k_i$ at this point.

\subsection{Proof step~2: Basic relations for products with width $k$.}
We next derive further constraints from the conditions on $k$-support outputs, i.e., $r_{\mathbf{B}^k_i} = 0$ for all $\mathbf{B}^k_i$. We note that only $k$-support operators $\{\mathbf{A}^k_j\}$ and $(k-1)$-support operators $\{\mathbf{A}^{k-1}_j\}$ may contribute to $\mathbf{B}^k_i$. In this step, we show that among the type~\ref{enum:input_1} inputs only a small family of candidates can remain, and that their coefficients are explicitly related to certain $(k-1)$-support coefficients; we then derive the corresponding constraints on those $(k-1)$-support inputs, which will be used in step~3; finally, we prove that all type~\ref{enum:input_3} inputs have zero coefficients.

\subsubsection{Inputs of type~\ref{enum:input_1}}
Again, we treat $\mathbf{A}^k$ of type~\ref{enum:input_1}. The conclusion of step~2 analysis for this type is summarized in the following Lemma: 
\begin{lemma}[Step 2 analysis for type~\ref{enum:input_1}]
Assume $\hat{Q}^k$ is a $k$-local conserved quantity of the one-dimensional Holstein model~\eqref{Eq.1D_H_model} with $t\neq 0$, $g\neq 0$, and $\omega \neq0$. 
The coefficients of $\mathbf{A}^k_i$ of type~\ref{enum:input_1} satisfy the following relation for any $i$:
\begin{align}
    q_{\hat{c}^\dagger_i 
\hat{c}_{i+k-1}}
    &=\frac{t}{g}q_{\hat{c}^\dagger_{i+1} 
\hat{c}(0,1)_{i+k-1}},\label{cdcb1}\\
    q_{\hat{c}^\dagger_{i+k-1} 
\hat{c}_{i+2k-2}}
    &=\frac{t}{g}q_{\hat{c}^\dagger(0,1)_{i+k-1} 
\hat{c}_{i+2k-3}},\label{cdcb2}\\
q_{\hat{c}_i 
\hat{c}^\dagger_{i+k-1}}
    &=\frac{t}{g}q_{\hat{c}_{i+1} 
\hat{c}^\dagger(0,1)_{i+k-1}},\\
    q_{\hat{c}_{i+k-1} 
\hat{c}^\dagger_{i+2k-2}}
    &=\frac{t}{g}q_{\hat{c}(0,1)_{i+k-1} 
\hat{c}^\dagger_{i+2k-3}},\\
q_{\hat{c}_i \hat{c}_{i+k-1}}&=q_{\hat{c}^\dagger_i \hat{c}^\dagger_{i+k-1}}=0.
\end{align}
\label{lem:step2_1}
\end{lemma}

\begin{proof}[Proof of Lemma \ref{lem:step2_1}]
By Lemma~\ref{lem:step1_1}, $k$-support input $\mathbf{A}^k$ of type~\ref{enum:input_1} has zero coefficients in most cases.
Below, we further divide the remaining type~\ref{enum:input_1} candidates $\mathbf{A}^k_i=\hat{c}^\dagger_i \hat{c}_{i+k-1},\hat{c}_i \hat{c}^\dagger_{i+k-1},\hat{c}_i \hat{c}_{i+k-1},\hat{c}^\dagger_i \hat{c}^\dagger_{i+k-1}$ into two classes according to the symmetry between creation and annihilation operators. We call $\hat{c}_i \hat{c}_{i+k-1}$ and $\hat{c}^\dagger_i \hat{c}^\dagger_{i+k-1}$ the $cc$ type, and $\hat{c}^\dagger_i \hat{c}_{i+k-1}$ and $\hat{c}_i \hat{c}^\dagger_{i+k-1}$ the $c^\dagger c$ type. We first show that all $cc$-type coefficients vanish. We then analyze the $c^\dagger c$ type and show that step~2 leaves only these candidates, whose coefficients are related to certain $(k-1)$-support coefficients.

\paragraph*{\underline{$cc$ type}}
Now, we consider $\mathbf{A}^k_i=\hat{c}_i 
\hat{c}_{i+k-1}$.  Its commutator with the electron--phonon term produces
the first right-hand side below.  Taking into account the step~1 constraints,
the only other commutator producing the same output is the second one:
\begin{align}
[\hat{c}_i 
\hat{c}_{i+k-1},\hat{n}(0,1)_{i+k-1}]&=\hat{c}_i 
\hat{c}_{i+k-1}\hat{b}_{i+k-1}, \\   
[\hat{c}_{i+1} 
\hat{c}_{i+k-1}\hat{b}_{i+k-1},\hat{c}^\dagger_{i+1}\hat{c}_i]&=\hat{c}_i 
\hat{c}_{i+k-1}\hat{b}_{i+k-1},
\end{align}
We denote their common right-hand side by $\mathbf B_i^k$.  Thus,
\begin{align}
    r_{\hat{c}_i 
\hat{c}_{i+k-1}\hat{b}_{i+k-1}} &=gq_{\hat{c}_i 
\hat{c}_{i+k-1}}+tq_{\hat{c}_{i+1} 
\hat{c}_{i+k-1}\hat{b}_{i+k-1}}&=0.
\end{align}
In other words,
\begin{align}
    q_{\hat{c}_i 
\hat{c}_{i+k-1}}
    &=-\frac{t}{g}q_{\hat{c}_{i+1} 
\hat{c}_{i+k-1}\hat{b}_{i+k-1}}
\label{cc_L}
\end{align}
holds.

Applying the same construction at the other endpoint, starting from the
electron--phonon commutator that produces
$\hat{c}_{i+k-1}\hat{b}_{i+k-1}\hat{c}_{i+2k-2}$ and comparing all inputs
that generate this output, we obtain
\begin{align}
    q_{\hat{c}_{i+k-1} 
\hat{c}_{i+2k-2}}
    &=-\frac{t}{g}q_{\hat{c}_{i+k-1}\hat{b}_{i+k-1}
\hat{c}_{i+2k-3}}.
\label{cc_R}
\end{align}

To relate these two coefficients, we introduce
\begin{align}
\mathcal{C}^{(m)}_i
:=\hat{c}_{i+1+m}\hat{b}_{i+k-1}\hat{c}_{i+k-1+m}
\quad (m=0,1,\ldots,k-2).
\end{align}
Equivalently, this definition reads
\begin{align}
\mathcal{C}^{(m)}_i=
\begin{cases}
\hat{c}_{i+1}\hat{c}(0,1)_{i+k-1}, & (m=0),\\
\hat{c}_{i+1+m}I(0,1)_{i+k-1}\hat{c}_{i+k-1+m}, & (1\leq m\leq k-3),\\
\hat{c}(0,1)_{i+k-1}\hat{c}_{i+2k-3}, & (m=k-2).
\end{cases}
\end{align}
Here the middle line is absent when $k=3$.

For $m=0,1,\ldots,k-3$, the relevant commutators are
\begin{align}
[\mathcal{C}^{(m)}_i,\hat{c}^\dagger_{i+k-1+m}\hat{c}_{i+k+m}]
&=\hat{c}_{i+1+m}\hat{b}_{i+k-1}\hat{c}_{i+k+m},
\\
[\mathcal{C}^{(m+1)}_i,\hat{c}^\dagger_{i+2+m}\hat{c}_{i+1+m}]
&=\hat{c}_{i+1+m}\hat{b}_{i+k-1}\hat{c}_{i+k+m},
\end{align}
and have the same right-hand side. We denote this output by
$\mathbf{B}^k_{i+1+m}=\hat{c}_{i+1+m}\hat{b}_{i+k-1}\hat{c}_{i+k+m}$.
Its coefficient condition yields
\begin{align}
q_{\mathcal{C}^{(m)}_i}&=-q_{\mathcal{C}^{(m+1)}_i}
\quad (m=0,1,\ldots,k-3).
\label{cc_C_recursion}
\end{align}
Therefore,
\begin{align}
q_{\mathcal{C}^{(0)}_i}
&=(-1)^{k-2}q_{\mathcal{C}^{(k-2)}_i}.
\label{cc_C_chain}
\end{align}
Equivalently, 
\begin{align}
q_{\hat{c}_{i+1} 
\hat{c}(0,1)_{i+k-1}}&=
\begin{cases}
    -q_{\hat{c}(0,1)_{i+k-1}\hat{c}_{i+2k-3}}, & (k\colon\textnormal{odd});\\
    q_{\hat{c}(0,1)_{i+k-1}\hat{c}_{i+2k-3}}, & (k\colon\textnormal{even})
\end{cases}
\end{align}
for $k\geq 3$.

Furthermore, combining this equation with Eq.~\eqref{cc_L} and Eq.~\eqref{cc_R}, we obtain the following relation:
\begin{align}
q_{\hat{c}_{i} 
\hat{c}_{i+k-1}}&=
\begin{cases}
    -q_{\hat{c}_{i+k-1}\hat{c}_{i+2k-2}}, & (k\colon\textnormal{odd});\\
    q_{\hat{c}_{i+k-1}\hat{c}_{i+2k-2}}, & (k\colon\textnormal{even}).
\end{cases}
\end{align}

Comparing this result with Eq.~\eqref{cc_k1}, we find
\begin{align}
q_{\hat{c}_{i} 
\hat{c}_{i+k-1}}&=q_{\hat{c}_{i+k-1}\hat{c}_{i+2k-2}}=0
\end{align}
for all $i$ regardless of whether $k$ is even or odd. We can also show that $\mathbf{A}^k_i=\hat{c}^\dagger_i 
\hat{c}^\dagger_{i+k-1}$ has zero coefficients for all $i$ in a similar manner.

\paragraph*{\underline{$c^\dagger c$ type}}
Next, we treat $\mathbf{A}^k_i=\hat{c}^\dagger_i 
\hat{c}_{i+k-1}$. The electron--phonon commutators of the two endpoint
candidates, together with the competing boundary-hopping inputs, generate the
two outputs
$\mathbf{B}^k_i = \hat{c}^\dagger_i 
\hat{c}(0,1)_{i+k-1}$ and $\mathbf{B}^k_{i+k-1} = \hat{c}^\dagger(0,1)_{i+k-1} 
\hat{c}_{i+2k-2}$. Their coefficient conditions give
\begin{align}
    gq_{\hat{c}^\dagger_i 
\hat{c}_{i+k-1}}-tq_{\hat{c}^\dagger_{i+1} 
\hat{c}(0,1)_{i+k-1}}&=0, \label{cdc_1}\\
    -gq_{\hat{c}^\dagger_{i+k-1} 
\hat{c}_{i+2k-2}}+tq_{\hat{c}^\dagger(0,1)_{i+k-1} 
\hat{c}_{i+2k-3}}&=0, \label{cdc_k}
\end{align}
which directly leads to Eqs.~\eqref{cdcb1} and \eqref{cdcb2}.
Similar argument is true for $\hat{c}_i\hat{c}_{i+k-1}^\dagger$ and $\hat{c}_{i+k-1} 
\hat{c}^\dagger_{i+2k-2}$.
\end{proof}

As a byproduct of the above analysis of type~\ref{enum:input_1} inputs, we also
obtain constraints on some coefficients of $(k-1)$-support inputs
$\mathbf{A}^{k-1}$. Indeed, combining Eq.~\eqref{cdc_k1}, Eq.~\eqref{cdc_1}, and
Eq.~\eqref{cdc_k}, the boundary-hopping commutators shift the fermionic
endpoint one site at a time. The first common output in this chain is
$\mathbf{B}^k_{i+1} = \hat{c}^\dagger_{i+1}I(0,1)_{i+k-1}\hat{c}_{i+k}$ when
$k\geq 4$. Comparing successively the coefficients of these outputs gives
\begin{align}
q_{\hat{c}^\dagger_{i+1} 
\hat{c}(0,1)_{i+k-1}}
&=q_{\hat{c}^\dagger_{i+2}I(0,1)_{i+k-1}\hat{c}_{i+k}} \nonumber \\
&=\dots \nonumber \\
&=q_{\hat{c}^\dagger(0,1)_{i+k-1} 
\hat{c}_{i+2k-3}},
\label{cdc_k-1}\\
q_{\hat{c}_{i+1} 
\hat{c}^\dagger(0,1)_{i+k-1}}
&=q_{\hat{c}_{i+2}I(0,1)_{i+k-1}\hat{c}^\dagger_{i+k}} \nonumber \\
&=\dots \nonumber \\
&=q_{\hat{c}(0,1)_{i+k-1} 
\hat{c}^\dagger_{i+2k-3}}.
\end{align}

\subsubsection{Constraints on $\mathbf{A}^{k-1}$}
In addition, we derive constraints on coefficients of $(k-1)$-support inputs written as 
$\mathbf{A}^{k-1}_i=C_1(x_1,y_1)_i C_2(x_2,y_2)_{i+k-2}$ with
$(C_1,C_2)\in \{(\hat{c},\hat{c}^\dagger),(\hat{c}^\dagger,\hat{c})\}$. These constraints are useful in step 3
when analyzing inputs for $\mathbf{B}^{k-1}$.

The conclusion of this analysis is summarized in the following Lemma: 
\begin{lemma}
Assume $\hat{Q}^k$ is a $k$-local conserved quantity of the one-dimensional Holstein model~\eqref{Eq.1D_H_model} with $t\neq 0$, $g\neq 0$, and $\omega \neq0$. 
The coefficients of $\mathbf{A}^{k-1}_i=C_1(x_1,y_1)_i C_2(x_2,y_2)_{i+k-2}$ with
$(C_1,C_2)\in \{(\hat{c},\hat{c}^\dagger),(\hat{c}^\dagger,\hat{c})\}$ satisfy the following relation for any $i$:
\begin{align}
q_{\hat{c}^\dagger_{i+1} 
\hat{c}(0,1)_{i+k-1}}
&=q_{\hat{c}^\dagger(0,1)_{i+k-1} 
\hat{c}_{i+2k-3}},\label{constr:cdc_b}\\
q_{\hat{c}^\dagger_{i+1} 
\hat{c}(1,0)_{i+k-1}}
&=q_{\hat{c}^\dagger(1,0)_{i+k-1} 
\hat{c}_{i+2k-3}},\label{constr:cdc_bdagger}\\
q_{\hat{c}_{i+1} 
\hat{c}^\dagger(0,1)_{i+k-1}}
&=q_{\hat{c}(0,1)_{i+k-1} 
\hat{c}^\dagger_{i+2k-3}},\label{constr:ccd_b}\\
q_{\hat{c}_{i+1} 
\hat{c}^\dagger(1,0)_{i+k-1}}
&=q_{\hat{c}(1,0)_{i+k-1} 
\hat{c}^\dagger_{i+2k-3}},\label{constr:ccd_bdagger}\\
q_{\hat{c}^\dagger_i\hat{c}_{i+k-2}}&=q_{\hat{c}^\dagger_{i+1}\hat{c}_{i+k-1}},\label{constr:cdc_shift}\\ 
q_{\hat{c}_i\hat{c}^\dagger_{i+k-2}}&=q_{\hat{c}_{i+1}\hat{c}^\dagger_{i+k-1}},\label{constr:ccd_shift}\\
q_{C_1(x_1,y_1)_i C_2(x_2,y_2)_{i+k-2}}&=0. \quad (\textnormal{otherwise})
\end{align}
\label{lem:constraints}
\end{lemma}

\begin{proof}
As an example, we consider the case 
$\mathbf{A}^{k-1}_i=\hat{c}^\dagger(x_1,y_1)_i\hat{c}(x_2,y_2)_{i+k-2}$.
We first extend this input to the right. The commutator
\begin{align}
&[\hat{c}^\dagger(x_1,y_1)_i\hat{c}(x_2,y_2)_{i+k-2},\hat{c}^\dagger_{i+k-2}\hat{c}_{i+k-1}] \nonumber \\
=&\hat{c}^\dagger(x_1,y_1)_iI(x_2,y_2)_{i+k-2}\hat{c}_{i+k-1}
\end{align}
generates the output
\begin{align*}
\mathbf{B}^{k}_i=\hat{c}^\dagger(x_1,y_1)_iI(x_2,y_2)_{i+k-2}\hat{c}_{i+k-1}.
\end{align*}
Thus, this commutator contributes to $r_{\mathbf{B}^{k}_i}$. By Lemmas~\ref{lem:step1_1} and~\ref{lem:step1_2}, this output can have another input only in one of the following two cases:
\begin{enumerate}[label=(R\arabic*)]
\item $(x_1,y_1)=(0,0)$;
\item $(x_1,y_1)\in\{(0,1),(1,0)\}$ and $(x_2,y_2)=(0,0)$.
\end{enumerate}
Otherwise this is the unique input, and hence $q_{\mathbf{A}^{k-1}_i}=0$.

We next extend the input to the left. The commutator
\begin{align}
&[\hat{c}^\dagger(x_1,y_1)_i\hat{c}(x_2,y_2)_{i+k-2},\hat{c}^\dagger_{i-1}\hat{c}_{i}] \nonumber \\
=&-\hat{c}^\dagger_{i-1}I(x_1,y_1)_i\hat{c}(x_2,y_2)_{i+k-2}
\end{align}
generates, up to its sign, the output
\begin{align*}
\mathbf{B}^{k}_{i-1}=\hat{c}^\dagger_{i-1}I(x_1,y_1)_i\hat{c}(x_2,y_2)_{i+k-2}.
\end{align*}
Thus, it contributes to $r_{\mathbf{B}^{k}_{i-1}}$. An analogous uniqueness argument shows that this output can have another input only in one of the following two cases:
\begin{enumerate}[label=(L\arabic*)]
\item $(x_2,y_2)=(0,0)$;
\item $(x_2,y_2)\in\{(0,1),(1,0)\}$ and $(x_1,y_1)=(0,0)$.
\end{enumerate}

Combining the right- and left-extension constraints, the coefficient of $\mathbf{A}^{k-1}_i$ can be nonzero only in one of the following three cases:
\begin{enumerate}[label=(C\arabic*)]
\item $(x_1,y_1)=(x_2,y_2)=(0,0)$;
\item $(x_1,y_1)=(0,0)$ and $(x_2,y_2)\in\{(0,1),(1,0)\}$;
\item $(x_1,y_1)\in\{(0,1),(1,0)\}$ and $(x_2,y_2)=(0,0)$.
\end{enumerate}
It remains to record the relations among the coefficients in these cases.

First, when $(x_1,y_1)=(x_2,y_2)=(0,0)$, the additional input obtained by shifting the support by one site gives
\begin{align}
[\hat{c}^\dagger_{i+1}\hat{c}_{i+k-1},\hat{c}^\dagger_{i}\hat{c}_{i+1}]
&=-\hat{c}^\dagger_i\hat{c}_{i+k-1},
\end{align}
which yields
\begin{align}
q_{\hat{c}^\dagger_i\hat{c}_{i+k-2}}&=q_{\hat{c}^\dagger_{i+1}\hat{c}_{i+k-1}}.
\end{align}
Second, when $(x_1,y_1)=(0,0)$ and $(x_2,y_2)=(0,1)$ or $(1,0)$, the surviving coefficients are related to the corresponding coefficients with the nontrivial phonon factor on the left end. Using Eq.~\eqref{cdc_k-1}, we obtain
\begin{align}
q_{\hat{c}^\dagger_{i+1} 
\hat{c}(0,1)_{i+k-1}}
&=q_{\hat{c}^\dagger(0,1)_{i+k-1} 
\hat{c}_{i+2k-3}},\\
q_{\hat{c}^\dagger_{i+1} 
\hat{c}(1,0)_{i+k-1}}
&=q_{\hat{c}^\dagger(1,0)_{i+k-1} 
\hat{c}_{i+2k-3}}.
\end{align}
The third case, where $(x_1,y_1)=(0,1)$ or $(1,0)$ and $(x_2,y_2)=(0,0)$, gives silimar two relations after relabeling the lattice sites.

An analogous argument applies to the case $(C_1,C_2) = (\hat{c},\hat{c}^\dagger)$. Thus, in addition to the two relations above, we also obtain
\begin{align}
q_{\hat{c}_{i+1} 
\hat{c}^\dagger(0,1)_{i+k-1}}
&=q_{\hat{c}(0,1)_{i+k-1} 
\hat{c}^\dagger_{i+2k-3}},\\
q_{\hat{c}_{i+1} 
\hat{c}^\dagger(1,0)_{i+k-1}}
&=q_{\hat{c}(1,0)_{i+k-1} 
\hat{c}^\dagger_{i+2k-3}},\\
q_{\hat{c}_i\hat{c}^\dagger_{i+k-2}}&=q_{\hat{c}_{i+1}\hat{c}^\dagger_{i+k-1}}.
\end{align}
All remaining coefficients of $\mathbf{A}^{k-1}_i=C_1(x_1,y_1)_i C_2(x_2,y_2)_{i+k-2}$ vanish.
\end{proof}

\subsubsection{Inputs of type~\ref{enum:input_3}}
Next, we treat $\mathbf{A}^k_i$ of type~\ref{enum:input_3}.
The conclusion of step~2 analysis for this type is summarized in the following Lemma:
\begin{lemma}[Step 2 analysis for type~\ref{enum:input_3}]
Assume $\hat{Q}^k$ is a $k$-local conserved quantity of the one-dimensional Holstein model~\eqref{Eq.1D_H_model} with $t\neq 0$, $g\neq 0$, and $\omega \neq0$.
$q_{\mathbf{A}^k_i}=0$ holds for any $\mathbf{A}^k_i$ of type~\ref{enum:input_3}.
\label{lem:step2_3}
\end{lemma}

\begin{proof}
First, for
$\mathbf{A}^k_i=I(x_1,y_1)_i\big( \prod^{k-1}_{j=2} A^j(x_j,y_j)_{i+j-1} \big)I(x_k,y_k)_{i+k-1}$
with $(x_1,y_1),(x_k,y_k) \neq (0,0)$, we show that the coefficient vanishes
unless $A^2=\cdots=A^{k-1}=I$. 
As an example, let us first consider the case $k=3$. One representative choice is
$\mathbf{A}^{k=3}_i=I(x_1,y_1)_i n(x_2,y_2)_{i+1}I(x_3,y_3)_{i+2}$.
The hopping commutator
\begin{align}
&[I(x_1,y_1)_i \hat{n}(x_{2},y_{2})_{i+1} I(x_3,y_3)_{i+2},\hat{c}^\dagger_{i+1} \hat{c}_{i+2}]\nonumber \\
=&I(x_1,y_1)_i  \hat{c}^\dagger(x_2,y_2)_{i+1} \hat{c}(x_3,y_3)_{i+2}
\end{align}
generates the output on the right-hand side, which we denote by
$\mathbf{B}_i^{k=3} = I(x_1,y_1)_i \hat{c}^\dagger(x_{2},y_{2})_{i+1} \hat{c}(x_3,y_3)_{i+2}$.
Taking into account bosonic degrees of freedom at the ends,
$\mathbf{A}^{k-1=2}$ does not contribute. Type~\ref{enum:input_2} inputs such as
$\mathbf{A}^{k=3}_i=I(x_1,y_1)_i  \hat{c}^\dagger(x_{2},y_{2})_{i+1} \hat{c}(x_3,y_3+1)_{i+2}$
may also contribute through
\begin{align}
&[I(x_1,y_1)_i  \hat{c}^\dagger(x_{2},y_{2})_{i+1} \hat{c}(x_3,y_3+1)_{i+2},\hat{n}(1,0)_{i+2}]\nonumber \\
=&(y_3+1)I(x_1,y_1)_i  \hat{c}^\dagger(x_{2},y_{2})_{i+1} \hat{c}(x_3,y_3)_{i+2}\nonumber \\
&+I(x_1,y_1)_i  \hat{c}^\dagger(x_{2},y_{2})_{i+1} \hat{c}(x_3+1,y_{3}+1)_{i+2},
\end{align}
but this type $\mathbf{A}^{k=3}$ has zero coefficients by the step 1 analysis. Here we used
$[\hat{c}_j(\hat{b}^\dagger)^{x_k}_j\hat{b}^{y_k+1}_j,\hat{n}_j\hat{b}^\dagger_j]=\hat{c}_j(\hat{b}^\dagger)^{x_k}_j\hat{b}^{y_k+1}_j\hat{b}^\dagger_j =(y_k+1)\hat{c}(x_k,y_k)_j+\hat{c}(x_k+1,y_k+1)_j$. 
Therefore, the coefficient of
$I(x_1,y_1)_{i}  \hat{n}(x_{2},y_{2})_{i+1} I(x_3,y_3)_{i+2}$
is zero.

\begin{widetext}
In general, for
$\mathbf{A}^k_i=I(x_1,y_1)_i\big( \prod^{k-1}_{j=2} A^j(x_j,y_j)_{i+j-1} \big)I(x_k,y_k)_{i+k-1}$ with $A^{k-1}\neq I$, by considering 
\begin{align}
\begin{cases}
[\mathbf{A}^k_i,\hat{c}^\dagger_{i+k-2} \hat{c}_{i+k-1}],
& (A^{k-1}=\hat{c},\hat{n});\\
[\mathbf{A}^k_i,\hat{c}^\dagger_{i+k-1} \hat{c}_{i+k-2}],
& (A^{k-1}=\hat{c}^\dagger),
\end{cases}
\end{align}
one can show that the coefficient is zero. Hence, it is sufficient to consider the case $A^{k-1}=I$.

Next, we show that the coefficient of $\mathbf{A}^k_i$ is zero unless $A^{k-2}=I$ when $k\geq 4$. 
Take $\mathbf{A}^{k=4}_i=I(x_1,y_1)_i \hat{n}(x_2,y_2)I(x_3,y_3)I(x_4,y_4)$ as an example. The hopping commutator
\begin{align}
&[I(x_1,y_1)_i  \hat{n}(x_{2},y_{2})_{i+1} I(x_{3},y_{3})_{i+2} I(x_4,y_4)_{i+3},\hat{c}^\dagger_{i+1} \hat{c}_{i+2}] \nonumber\\
=&I(x_1,y_1)_i \hat{c}^\dagger(x_{2},y_{2})_{i+1}\hat{c}(x_{3},y_{3})_{i+2} I(x_4,y_4)_{i+3}
\end{align}
generates the output on the right-hand side, which we denote by
$\mathbf{B}_i^{k=4} = I(x_1,y_1)_i  \hat{c}^\dagger(x_{2},y_{2})_{i+1}\hat{c}(x_{3},y_{3})_{i+2} I(x_4,y_4)_{i+3}$.
By the preceding analysis, it is the only input whose
coefficient is not already known to vanish, and therefore
$r_{I(x_1,y_1)_i \hat{c}^\dagger(x_{2},y_{2})_{i+1}\hat{c}(x_{3},y_{3})_{i+2} I(x_4,y_4)_{i+3}}
=t q_{I(x_1,y_1)_i  \hat{n}(x_{2},y_{2})_{i+1} I(x_{3},y_{3})_{i+2} I(x_4,y_4)_{i+3}}=0$
holds. Similarly, when $A^{k-2}\neq I$, the coefficient of
\begin{align*}
\mathbf{A}^k_i=I(x_1,y_1)_i\big( \prod^{k-2}_{j=2} A^j(x_j,y_j)_{i+j-1} \big)I(x_{k-1},y_{k-1})_{i+k-2}I(x_k,y_k)_{i+k-1}\end{align*}
can be shown to zero, in general.

Repeating this argument, for any $k$, all type~\ref{enum:input_3} operators
$\mathbf{A}^k$ have zero coefficients unless
$(A_1=)A_2=\cdots=A_{k-1}(=A_k)=I$.

Now, we treat 
$\mathbf{A}^k_i=\prod^{k}_{j=1} I(x_j,y_j)_{i+j-1}$.
First, we suppose $x_j \neq 0$ with $1\leq j\leq k$. The relevant commutators are
\begin{align}
&[\prod^{k}_{l=1} I(x_l,y_l)_{i+l-1} ,\hat{n}(0,1)_{i+j-1}]\nonumber \\
=&-x_j\big(\prod^{j-1}_{l=1} I(x_l,y_l)_{i+l-1}\big)\hat{n}(x_j-1,y_j)_{i+j-1}\big(\prod^{k}_{l=j+1} I(x_l,y_l)_{i+l-1}\big),
\\
&[\big(\prod^{j-1}_{l=1} I(x_l,y_l)_{i+l-1}\big)I(x_j-1,y_j+1)_{i+j-1}\big(\prod^{k}_{l=j+1} I(x_l,y_l)_{i+l-1}\big) ,\hat{n}(1,0)_{i+j-1}] \nonumber\\
=&(y_j+1)\big(\prod^{j-1}_{l=1} I(x_l,y_l)_{i+l-1}\big)\hat{n}(x_j-1,y_j)_{i+j-1}\big(\prod^{k}_{l=j+1} I(x_l,y_l)_{i+l-1}\big).
\end{align}
\end{widetext}
The two commutators have the same right-hand side and exhaust the candidates
that may carry nonzero coefficients. Therefore,
$q_{I(x_1,y_1)_i\cdots I(x_j,y_j)_{i+j-1}\cdots I(x_k,y_k)_{i+k-1}}\propto q_{I(x_1,y_1)_i\cdots I(x_j-1,y_j+1)_{i+j-1}\cdots I(x_k,y_k)_{i+k-1}}$ holds.

By repeating this argument, we can show that for any
$\mathbf{A}^k_i=\prod^{k}_{j=1} I(x_j,y_j)_{i+j-1}$, 
$\mathbf{A'}^k_i=\prod^{k}_{j=1} I(0,x_j+y_j)_{i+j-1}$ has a proportional coefficient. 
The phonon-number commutator is
\begin{widetext}
\begin{align}
    [\prod^{k}_{j'=1} I(0,x_{j'}+y_{j'})_{i+j'-1},I(1,1)_{i+j-1}]
    &=(x_j+y_j)\prod^{k}_{j'=1} I(0,x_{j'}+y_{j'})_{i+j'-1}
\end{align}
for any $j = 1,2,\cdots,k$ with $x_j+y_j \neq 0$. 
Its right-hand side is the output
$\mathbf{B}^k=\mathbf{A'}^k_i=\prod^{k}_{j=1} I(0,x_j+y_j)_{i+j-1}$.
Hence,
\begin{align}
\omega(x_1+y_1+x_2+y_2+\cdots+x_k+y_k)q_{\prod^{k}_{j=1} I(0,x_j+y_j)_{i+j-1}}&=0
\end{align}
\end{widetext}
holds. 
By the end constraints, $x_1+y_1, x_k+y_k > 0$. Therefore, the coefficient of
$\prod^{k}_{j=1} I(0,x_j+y_j)_{i+j-1}$
vanishes, and so does the coefficient of the proportional operator
$\prod^{k}_{j=1} I(x_j,y_j)_{i+j-1}$.

Hence, all inputs of type~\ref{enum:input_3} have zero coefficients.
\end{proof}

\subsection{Proof step 3: Basic relations for products with width $k-1$.}
Finally, by analyzing inputs for $\mathbf{B}^{k-1}$, we use the step~2 constraints on $\mathbf{A}^{k-1}$ to eliminate the last remaining type~\ref{enum:input_1} candidates. More precisely, we prove that the coefficients of the related $k$-support inputs $\mathbf{A}^k_i=\hat{c}^\dagger_i\hat{c}_{i+k-1},\hat{c}_i\hat{c}^\dagger_{i+k-1}$ vanish for $3\leq k<L/2$. We discuss the cases $k=3$ and $4 \leq k < L/2$ separately, and in both cases conclude that all remaining type~\ref{enum:input_1} inputs have zero coefficients.

\subsubsection{Step 3 for $k=3$}
We first consider the case $k=3$. 

\begin{lemma}[Step 3 analysis for the case $k=3$]
Assume $\hat{Q}^{k=3}$ is a $3$-local conserved quantity of the one-dimensional Holstein model~\eqref{Eq.1D_H_model} with $t\neq 0$, $g\neq 0$, and $\omega \neq0$.
$q_{\mathbf{A}^{k=3}_i}=0$ holds for any $\mathbf{A}^{k=3}_i$ of type~\ref{enum:input_1}.
\label{lem:step3_k3}
\end{lemma}

\begin{proof}
By Lemma \ref{lem:step2_1}, $3$-support input $\mathbf{A}^{k=3}$ of type~\ref{enum:input_1} has zero coefficients in most cases.
Below, we will prove that the remaining type~\ref{enum:input_1} candidates $\mathbf{A}^{k=3}_i=\hat{c}^\dagger_i \hat{c}_{i+2},\hat{c}_i \hat{c}^\dagger_{i+2}$ also have zero coefficients. 

The relevant commutators are
\begin{align}
[\hat{c}^\dagger_{i+1}\hat{c}(0,1)_{i+2},I(1,1)_{i+2}]&=\hat{c}^\dagger_{i+1}\hat{c}(0,1)_{i+2},\\
[\hat{c}^\dagger_{i+1}\hat{c}_{i+2},\hat{n}(0,1)_{i+2}]&=\hat{c}^\dagger_{i+1}\hat{c}(0,1)_{i+2},\\
[\hat{n}(0,1)_{i+2},\hat{c}^\dagger_{i+1}\hat{c}_{i+2}]&=-\hat{c}^\dagger_{i+1}\hat{c}(0,1)_{i+2},
\end{align}
and have the same right-hand side up to its sign. We denote this output by
$\mathbf{B}^{k-1=2}_{i+1}=\hat{c}^\dagger_{i+1}\hat{c}(0,1)_{i+2}$.
As we verify below, its coefficient condition implies
\begin{align}
\omega q_{\hat{c}^\dagger_{i+1}\hat{c}(0,1)_{i+2}}
+g q_{\hat{c}^\dagger_{i+1}\hat{c}_{i+2}}
-t q_{\hat{n}(0,1)_{i+2}}&=0.
\label{eq:lem6-first-condition}
\end{align}

Besides the three contributions displayed above, the following
input may also contribute:
\begin{align}
&[\hat{n}_{i+1}\hat{b}_{i+2},
\hat{c}^\dagger_{i+1}\hat{c}_{i+2}]\nonumber\\
&\qquad=\hat{c}^\dagger_{i+1}\hat{c}(0,1)_{i+2}.
\end{align}
The right-hand side is the same output
$\mathbf{B}^{k-1=2}_{i+1}$.
The coefficient of this additional input vanishes by the
boundary-hopping commutator
\begin{align}
&[\hat{n}_{i+1}\hat{b}_{i+2},
\hat{c}^\dagger_i\hat{c}_{i+1}]\nonumber\\
&\qquad=-\hat{c}^\dagger_i\hat{c}_{i+1}\hat{b}_{i+2}.
\end{align}
Its right-hand side is a $3$-support output $\mathbf{B}^3$.
By Lemmas~\ref{lem:step1_1}, \ref{lem:step1_2}, and
\ref{lem:step2_3}, all $3$-support coefficients are already zero except
those of the bare endpoint pairs. These remaining pairs cannot produce
this output.
Since $t\neq0$, we obtain
\begin{align}
q_{\hat{n}_{i+1}\hat{b}_{i+2}}=0.
\end{align}
Consequently, the additional hopping contribution vanishes.
Higher-degree phonon-dressed endpoint pairs have already been excluded
by Lemma~\ref{lem:constraints}. Thus
Eq.~\eqref{eq:lem6-first-condition} follows from
$r_{\mathbf{B}^{k-1=2}_{i+1}}=0$.
Similarly, the analogous onsite and endpoint commutators at the other end
generate
$\mathbf{B}^{k-1=2}_{i+2}=\hat{c}^\dagger(0,1)_{i+2}\hat{c}_{i+3}$,
where the possible additional input $\hat{b}_{i+2}\hat{n}_{i+3}$
is excluded by the same boundary-hopping comparison at the right end.
Thus $q_{\hat{b}_{i+2}\hat{n}_{i+3}}=0$, and the coefficient condition
\begin{align}
\omega q_{\hat{c}^\dagger(0,1)_{i+2}\hat{c}_{i+3}}
-g q_{\hat{c}^\dagger_{i+2}\hat{c}_{i+3}}
+t q_{\hat{n}(0,1)_{i+2}}&=0
\end{align}
holds. 
Combining these equations with step~2 relation~\eqref{cdc_k-1} and $(k-1)$-support constraints~\eqref{constr:cdc_shift}, we obtain
\begin{align}
q_{\hat{c}^\dagger_{i+1}\hat{c}(0,1)_{i+2}}&=q_{\hat{c}^\dagger(0,1)_{i+2}\hat{c}_{i+3}}=0
\end{align}
and therefore, again from the step~2 relation~\eqref{cdc_1},
\begin{align}
q_{\hat{c}^\dagger_{i}\hat{c}_{i+2}}&=0
\end{align}
follows.

Similar reasoning for
$\mathbf{A}^{k=3}_i=\hat{c}_{i}\hat{c}^\dagger_{i+2}$, using coefficients such as
$\mathbf{A}^{k-1=2}_{i+1}=\hat{c}_{i+1}\hat{c}^\dagger(0,1)_{i+2}$, also yields zero
coefficients.
\end{proof}

\subsubsection{Step 3 for $4 \leq k < L/2$}
Next, consider the case $4 \leq k < L/2$. 

\begin{lemma}[Step 3 analysis for the case $4 \leq k < L/2$]
Assume $\hat{Q}^{k}$ is a $k$-local conserved quantity of the one-dimensional Holstein model~\eqref{Eq.1D_H_model} with $t\neq 0$, $g\neq 0$, and $\omega \neq 0$.
$q_{\mathbf{A}^{k}_i}=0$ holds for any $\mathbf{A}^{k}_i$ of type~\ref{enum:input_1}.
\label{lem:step3_kgeneral}
\end{lemma}

\begin{proof}
By Lemma \ref{lem:step2_1}, $k$-support input $\mathbf{A}^k$ of type~\ref{enum:input_1} has zero coefficients in most cases.
Below, we will prove that the remaining type~\ref{enum:input_1} candidates $\mathbf{A}^k_i=\hat{c}^\dagger_i \hat{c}_{i+k-1},\hat{c}_i \hat{c}^\dagger_{i+k-1}$ also have zero coefficients. 

We first look at $k=4$ as a representative example.
Again, by proving that coefficients such as
$\mathbf{A}^{k-1=3}_{i+1}=\hat{c}^\dagger_{i+1}\hat{c}(0,1)_{i+3}$ vanish, we obtain the
vanishing of the related coefficient for
$\mathbf{A}^{k=4}_{i}=\hat{c}^\dagger_{i}\hat{c}_{i+3}$.

As in the $k=3$ case, the onsite phonon and endpoint-hopping commutators generate
$\mathbf{B}^{k-1=3}_{i+1} = \hat{c}^\dagger_{i+1}\hat{c}(0,1)_{i+3}$ and
$\mathbf{B}^{k-1=3}_{i+3} =\hat{c}^\dagger(0,1)_{i+3}\hat{c}_{i+5}$. Their coefficient conditions are
\begin{align}
\omega q_{\hat{c}^\dagger_{i+1}\hat{c}(0,1)_{i+3}}
+g q_{\hat{c}^\dagger_{i+1}\hat{c}_{i+3}}
-t q_{\hat{c}^\dagger_{i+2}\hat{c}(0,1)_{i+3}}=0,\\
\omega q_{\hat{c}^\dagger(0,1)_{i+3}\hat{c}_{i+5}}
-g q_{\hat{c}^\dagger_{i+3}\hat{c}_{i+5}}
+t q_{\hat{c}^\dagger(0,1)_{i+3}\hat{c}_{i+4}}=0.
\end{align}
Furthermore, the relevant commutators are
\begin{align}
[\hat{c}^\dagger_{i+2}\hat{c}(0,1)_{i+3},\hat{c}^\dagger_{i+3}\hat{c}_{i+4}]&=\hat{c}^\dagger_{i+2}I(0,1)_{i+3}\hat{c}_{i+4},\\
[\hat{c}^\dagger(0,1)_{i+3}\hat{c}_{i+4},\hat{c}^\dagger_{i+2}\hat{c}_{i+3}]&=-\hat{c}^\dagger_{i+2}I(0,1)_{i+3}\hat{c}_{i+4},\\
[\hat{c}^\dagger_{i+2}I(0,1)_{i+3}\hat{c}_{i+4},I(1,1)_{i+3}]&=\hat{c}^\dagger_{i+2}I(0,1)_{i+3}\hat{c}_{i+4},
\end{align}
and have the same right-hand side up to its sign. We denote this output by
$\mathbf{B}^{k-1=3}_{i+2}=\hat{c}^\dagger_{i+2}I(0,1)_{i+3}\hat{c}_{i+4}$.
Its coefficient condition gives the following relation:
\begin{align}
\omega q_{\hat{c}^\dagger_{i+2}I(0,1)_{i+3}\hat{c}_{i+4}}
+t q_{\hat{c}^\dagger_{i+2}\hat{c}(0,1)_{i+3}}
-t q_{\hat{c}^\dagger(0,1)_{i+3}\hat{c}_{i+4}}&=0.
\end{align}
Based on these relations together with the step~2 relation~\eqref{cdc_k-1} and $(k-1)$-support constraints~\eqref{constr:cdc_shift}, we define
$a \coloneq \omega q_{\hat{c}^\dagger_{i+1}\hat{c}(0,1)_{i+3}}
= \omega q_{\hat{c}^\dagger(0,1)_{i+3}\hat{c}_{i+5}}
=\omega q_{\hat{c}^\dagger_{i+2}I(0,1)_{i+3}\hat{c}_{i+4}}$, $
b \coloneq g q_{\hat{c}^\dagger_{i+1}\hat{c}_{i+3}}
=g q_{\hat{c}^\dagger_{i+3}\hat{c}_{i+5}}$, $c_2 \coloneq t q_{\hat{c}^\dagger_{i+2}\hat{c}(0,1)_{i+3}}$,
$c_3 \coloneq t q_{\hat{c}^\dagger(0,1)_{i+3}\hat{c}_{i+4}}$. Then,
\begin{equation}
\renewcommand{\arraystretch}{1.2}
\begin{array}{rrrrr}
  a & +  b & - c_2 &   & =0 \\
  a &      & +  c_2 & -  c_3& = 0 \\
  a & -  b &        & +  c_3& = 0 \\ \hline
  3a&    &      &       & =0
\end{array}
\end{equation}
holds. Thus $a=0$, and hence
\begin{align}
    q_{\hat{c}^\dagger_{i}\hat{c}_{i+3}}=0.
\end{align}

For $k\geq 5$, a similar strategy applies: proving zero coefficients for terms such
as $\mathbf{A}^{k-1}_{i+1}=\hat{c}^\dagger_{i+1}\hat{c}(0,1)_{i+k-1}$ implies the
vanishing of the related coefficient of
$\mathbf{A}^{k}_{i}=\hat{c}^\dagger_{i}\hat{c}_{i+k-1}$.

In general, the onsite phonon and endpoint-hopping commutators generate
$\mathbf{B}^{k-1}_{i+1} = \hat{c}^\dagger_{i+1}\hat{c}(0,1)_{i+k-1}$ and
$\mathbf{B}^{k-1}_{i+k-1} =\hat{c}^\dagger(0,1)_{i+k-1}\hat{c}_{i+2k-3}$. Their coefficient conditions are
\begin{widetext}
\begin{align}
\omega q_{\hat{c}^\dagger_{i+1}\hat{c}(0,1)_{i+k-1}}
+g q_{\hat{c}^\dagger_{i+1}\hat{c}_{i+k-1}}
-t q_{\hat{c}^\dagger_{i+2}\hat{c}(0,1)_{i+k-1}}&=0,
\end{align}
\begin{align}
\omega q_{\hat{c}^\dagger(0,1)_{i+k-1}\hat{c}_{i+2k-3}}
-g q_{\hat{c}^\dagger_{i+k-1}\hat{c}_{i+2k-3}}
+t q_{\hat{c}^\dagger(0,1)_{i+k-1}\hat{c}_{i+2k-4}}&=0.
\end{align}
Moreover, for $k\geq 5$, successive boundary-hopping commutators shift the
fermionic endpoints through the interior and generate the outputs
$\mathbf{B}^{k-1}=\hat{c}^\dagger_{i+2}I(0,1)_{i+k-1}\hat{c}_{i+k},\hat{c}^\dagger_{i+3}I(0,1)_{i+k-1}\hat{c}_{i+k+1},\ldots,\hat{c}^\dagger_{i+k-2}I(0,1)_{i+k-1}\hat{c}_{i+2k-4}$
and give the following relations:
\begin{align}
\omega q_{\hat{c}^\dagger_{i+2}I(0,1)_{i+k-1}\hat{c}_{i+k}}
+t q_{\hat{c}^\dagger_{i+2}\hat{c}(0,1)_{i+k-1}}
-t q_{\hat{c}^\dagger_{i+3}I(0,1)_{i+k-1}\hat{c}_{i+k}}&=0,
\end{align}
\begin{align}
\omega q_{\hat{c}^\dagger_{i+3}I(0,1)_{i+k-1}\hat{c}_{i+k+1}}
+t q_{\hat{c}^\dagger_{i+3}I(0,1)_{i+k-1}\hat{c}_{i+k}}
-t q_{\hat{c}^\dagger_{i+4}I(0,1)_{i+k-1}\hat{c}_{i+k+1}}&=0,
\end{align}
\begin{align}
    \vdots \nonumber
\end{align}
\begin{align}
\omega q_{\hat{c}^\dagger_{i+k-2}I(0,1)_{i+k-1}\hat{c}_{i+2k-4}}
+t q_{\hat{c}^\dagger_{i+k-2}I(0,1)_{i+k-1}\hat{c}_{i+2k-5}}
-t q_{\hat{c}^\dagger_{i+k-1}I(0,1)_{i+k-1}\hat{c}_{i+2k-4}}&=0.
\end{align}
\end{widetext}
Again, we define
$a \coloneq \omega q_{\hat{c}^\dagger_{i+1}\hat{c}(0,1)_{i+k-1}}
=\omega q_{\hat{c}^\dagger_{i+2}I(0,1)_{i+k-1}\hat{c}_{i+k}}=\dots
= \omega q_{\hat{c}^\dagger(0,1)_{i+k-1}\hat{c}_{i+2k-3}}
$, $
b \coloneq g q_{\hat{c}^\dagger_{i+1}\hat{c}_{i+k-1}}
=g q_{\hat{c}^\dagger_{i+k-1}\hat{c}_{i+2k-3}}$, $c_2 \coloneq t q_{\hat{c}^\dagger_{i+2}\hat{c}(0,1)_{i+k-1}}$,
$c_3 \coloneq t q_{\hat{c}^\dagger_{i+3}I(0,1)_{i+k-1}\hat{c}_{i+k}}$,...,$c_{k-1} \coloneq t q_{\hat{c}^\dagger(0,1)_{i+k-1}\hat{c}_{i+2k-4}}$.
Then
\begin{equation}
\renewcommand{\arraystretch}{1.2}
\begin{array}{rrrrrrrr}
  a & +b & -c_2 &  &  &   &     & =  0 \\
  a &   &  +c_2 & -c_3 &  &  &  & =  0 \\
  a &   &   & +c_3  &  -c_4   &  &  &    = 0 \\
    &   & \vdots &   &     &   &     &       \\
  a &   &   &  &  & +c_{k-2} & -c_{k-1} & =  0 \\
  a & -b &  &   &           &  & +c_{k-1} & =  0 \\ \hline
  (k-1)a &   &   &   &           &   &  & =  0
\end{array}
\end{equation}
holds. Thus $a=0$, and hence
\begin{align}
    q_{\hat{c}^\dagger_{i}\hat{c}_{i+k-1}}=0.
\end{align}

By an analogous argument, the coefficient of
$\mathbf{A}^k_i=\hat{c}_{i}\hat{c}^\dagger_{i+k-1}$ also vanishes.
Therefore, all inputs of type~\ref{enum:input_1} have zero coefficients.
\end{proof}

Combining Lemma~\ref{lem:step1_2}, \ref{lem:step2_3}, \ref{lem:step3_k3}, and \ref{lem:step3_kgeneral}, we obtain Theorem~\ref{nontrivial}

\section{Short support conserved quantities}\label{sec:short}

We now examine trivial conserved quantities, i.e., $k$-local conserved quantities
with $k=1,2$. Obvious examples are the Hamiltonian itself $\hat{H}$, and the total
fermion number
$\hat{N}\coloneq \sum_i\hat{n}_i=\sum_i\hat{c}^\dagger_i\hat{c}_i$. As shown below, these exhaust
all independent trivial conserved quantities.

\begin{theorem}
    In the one-dimensional Holstein model~\eqref{Eq.1D_H_model} with
    $t\neq 0$, $g\neq 0$, and $\omega \neq 0$, the only independent
    $k$-local conserved quantities for $1\leq k\leq 2$ are
    the Hamiltonian $\hat{H}$ and the 
total fermion number $\hat{N}$.
\label{trivial}
\end{theorem}

Combining this theorem with Theorem~\ref{nontrivial}, we obtain Theorem~\ref{withtrivial}.

\begin{proof}[Proof of Theorem~\ref{trivial}]
We prove the statement separately for the cases $k=2$ and $k=1$.
\subsubsection{Proof for $k=2$ case}
We show that the only 2-local conserved quantity is the Hamiltonian $\hat{H}$ itself up to the freedom of adding 1-local conserved quantities.

First, we consider $\mathbf{A}^2_i = e_i^1 e_{i+1}^2$ of type~\ref{enum:input_1} and
type~\ref{enum:input_2}. Without loss of generality, we assume
$e^2\neq I(x,y)$. We first discuss explicitly the case in which the right
end is $e^2=\hat{c}(x_2,y_2)$. Namely, take
\begin{align}
\mathbf{A}^2_i=e^1_i\hat{c}(x_2,y_2)_{i+1}.
\end{align}
The commutator with the hopping term on the right gives
\begin{align}
[e^1_i\hat{c}(x_2,y_2)_{i+1},
\hat{c}^\dagger_{i+1}\hat{c}_{i+2}]
&=e^1_iI(x_2,y_2)_{i+1}\hat{c}_{i+2}.
\end{align}
If $e^1\neq \hat{c}^\dagger(0,0),\hat{c}(0,0)$, then this output cannot be
produced by any other $2$-support input whose coefficient has not already been
excluded. Hence the coefficient of
$e^1_i\hat{c}(x_2,y_2)_{i+1}$ must vanish. Thus a nonzero coefficient is
possible only when the left endpoint is $\hat{c}^\dagger(0,0)$ or
$\hat{c}(0,0)$.

The cases $e^2=\hat{n}(x_2,y_2)$ and
$e^2=\hat{c}^\dagger(x_2,y_2)$ are treated similarly.  Their right-boundary
hopping commutators produce
$e^1_i\hat{c}(x_2,y_2)_{i+1}\hat{c}_{i+2}$ and
$e^1_iI(x_2,y_2)_{i+1}\hat{c}^\dagger_{i+2}$, respectively.

These outputs again have no competing inputs unless the left endpoint is
$\hat{c}^\dagger(0,0)$ or $\hat{c}(0,0)$. An analogous argument applied to the
left hopping term shows that the right endpoint must also be
$\hat{c}^\dagger(0,0)$ or $\hat{c}(0,0)$. Therefore, coefficients of
type~\ref{enum:input_1} and type~\ref{enum:input_2} inputs $\mathbf{A}^2$
vanish unless both endpoints are $\hat{c}^\dagger(0,0)$ or $\hat{c}(0,0)$. In
particular, all type~\ref{enum:input_2} candidates have zero coefficients.

Moreover, the surviving particle-conserving inputs $\mathbf{A}^2_i=\hat{c}^\dagger_i\hat{c}_{i+1}, \hat{c}_i\hat{c}^\dagger_{i+1}$ have uniform
coefficients. To see this, take
$\mathbf{A}^2_i=\hat{c}^\dagger_i\hat{c}_{i+1}$ as a representative case.
The two boundary commutators
\begin{align}
[\hat{c}^\dagger_i\hat{c}_{i+1},
\hat{c}^\dagger_{i+1}\hat{c}_{i+2}]
&=\hat{c}^\dagger_i\hat{c}_{i+2},
\\
[\hat{c}^\dagger_{i+1}\hat{c}_{i+2},
\hat{c}^\dagger_i\hat{c}_{i+1}]
&=-\hat{c}^\dagger_i\hat{c}_{i+2}.
\end{align}
generate the same output $\hat{c}^\dagger_i\hat{c}_{i+2}$ up to sign.
Comparing the coefficient of this output gives
$q_{\hat{c}^\dagger_i\hat{c}_{i+1}}
=q_{\hat{c}^\dagger_{i+1}\hat{c}_{i+2}}$. Repeating this comparison over the
lattice yields
\begin{align}
q_{\hat{c}^\dagger_{i}\hat{c}_{i+1}}
=q_{\hat{c}^\dagger_{i+1}\hat{c}_{i+2}}=\dots .
\end{align}
The analogous comparison for the opposite hopping direction gives
\begin{align}
q_{\hat{c}_{i}\hat{c}^\dagger_{i+1}}
=q_{\hat{c}_{i+1}\hat{c}^\dagger_{i+2}}=\dots .
\end{align}

Next, consider particle-nonconserving inputs $\mathbf{A}^2_i=\hat{c}_{i}\hat{c}_{i+1}, \hat{c}^\dagger_{i}\hat{c}^\dagger_{i+1}$. For the representative
case $\mathbf{A}^2_i=\hat{c}_{i}\hat{c}_{i+1}$, the commutator
\begin{align}
[\hat{c}_{i}\hat{c}_{i+1},\hat{n}(0,1)_{i+1}]
&=\hat{c}_{i}\hat{c}(0,1)_{i+1}.
\end{align}
No other input generates
this output.  Hence $q_{\hat{c}_{i}\hat{c}_{i+1}}=0$. The Hermitian-conjugate case is treated
similarly: using
\begin{align}
[\hat{c}^\dagger_{i}\hat{c}^\dagger_{i+1},\hat{n}(0,1)_{i+1}]
&=-\hat{c}^\dagger_{i}\hat{c}^\dagger(0,1)_{i+1},
\end{align}
we also obtain
$q_{\hat{c}^\dagger_{i}\hat{c}^\dagger_{i+1}}=0$.

It remains to relate the coefficients of the hopping terms to those of the
phonon-dressed number operators. Again, take
$\mathbf{A}^2_i=\hat{c}^\dagger_i\hat{c}_{i+1}$ as a representative case.
The relevant commutators are
\begin{align}
[\hat{c}^\dagger_{i}\hat{c}_{i+1},\hat{n}(0,1)_{i+1}]
&=\hat{c}^\dagger_{i}\hat{c}(0,1)_{i+1},
\\
[\hat{n}(0,1)_{i+1},\hat{c}^\dagger_{i}\hat{c}_{i+1}]
&=-\hat{c}^\dagger_{i}\hat{c}(0,1)_{i+1}.
\end{align}
Their right-hand sides are the same output up to sign.
Comparing the coefficient of this output in $[Q,H]$ gives
$gq_{\hat{c}^\dagger_i\hat{c}_{i+1}}=tq_{\hat{n}(0,1)_{i+1}}$. The analogous
comparison with $\hat{n}(1,0)_{i+1}$ gives
$gq_{\hat{c}^\dagger_i\hat{c}_{i+1}}=tq_{\hat{n}(1,0)_{i+1}}$. Applying the
analogous argument to the opposite hopping direction yields the corresponding
relation with $\hat{c}_i\hat{c}^\dagger_{i+1}$. Therefore,
\begin{align}
gq_{\hat{c}^\dagger_{i} 
\hat{c}_{i+1}}
&=tq_{\hat{n}(0,1)_{i+1}}
=tq_{\hat{n}(1,0)_{i+1}}
=-gq_{\hat{c}_{i}\hat{c}^\dagger_{i+1} 
}.
\end{align}

Before the one-support coefficient comparison, we also exclude
higher-degree onsite number inputs. For $x+y\geq2$,
\begin{align}
&[\hat{n}(x,y)_i,\hat{c}^\dagger_i\hat{c}_{i+1}]
\nonumber\\
&\qquad=\hat{c}^\dagger(x,y)_i\hat{c}_{i+1}.
\label{eq:short-onsite-high-degree}
\end{align}
The surviving $2$-support fermionic inputs are bare endpoint pairs,
whose electron--phonon commutators have at most one bosonic factor.
Purely bosonic $2$-support inputs commute with hopping and cannot
produce this fermionic output. A one-support input at the neighboring
site cannot produce its bosonic factor at site $i$.
Hence this output has no other input whose coefficient may be nonzero,
and $q_{\hat{n}(x,y)_i}=0$ for $x+y\geq2$.

Finally, we relate $\hat{n}(0,1)_i$ and $\hat{n}(1,0)_i$ to the phonon-number
operator. For $\hat{n}(0,1)_i$, the output $\hat{n}(0,1)_i$ is generated by
\begin{align}
[\hat{n}(0,1)_{i},I(1,1)_i]&=\hat{n}(0,1)_{i},
\\
[I(1,1)_i,\hat{n}(0,1)_{i}]&=-\hat{n}(0,1)_{i}.
\end{align}
We denote their common output by $\mathbf B_i^1=\hat n(0,1)_i$.
Besides the two contributions above, the same output is also generated by
\begin{align}
[I(0,2)_i,\hat{n}(1,0)_i]&=2\hat{n}(0,1)_i,
\\
[\hat{n}(0,2)_i,\hat{n}(1,0)_i]&=2\hat{n}(0,1)_i,
\\
[\hat{n}(1,1)_i,\hat{n}(0,1)_i]&=-\hat{n}(0,1)_i.
\end{align}
The coefficients of $\hat{n}(0,2)_i$ and $\hat{n}(1,1)_i$
have already been shown to vanish by the higher-degree onsite
comparison above. To exclude $I(0,2)_i$, we use
\begin{align}
[I(0,2)_i,I(1,1)_i]&=2I(0,2)_i.
\end{align}
Its right-hand side is a different one-support output
$I(0,2)_i$, apart from the factor $2$.
The electron--phonon term produces a number operator, not a purely
bosonic output. The remaining hopping inputs have no bosonic factors,
and purely bosonic $2$-support inputs cannot lose an endpoint through
a phonon-number commutator. Thus the displayed input is the only
contribution to this output. Since $\omega\neq0$,
$2\omega q_{I(0,2)_i}=0$ implies $q_{I(0,2)_i}=0$.
Consequently, the additional contributions to $\mathbf B_i^1$ vanish.
The condition $r_{\mathbf B_i^1}=0$ gives
$\omega q_{\hat n(0,1)_i}=gq_{I(1,1)_i}$.
For $\hat{n}(1,0)_i$, the additional input $I(2,0)_i$
is excluded by the same phonon-number comparison, while
$\hat{n}(2,0)_i$ and $\hat{n}(1,1)_i$ have already been excluded.
The analogous calculation therefore
gives $gq_{I(1,1)_i}=\omega q_{\hat{n}(1,0)_i}$. Hence,
\begin{align}
\omega q_{\hat{n}(0,1)_{i}}
=g q_{I(1,1)_i}
=\omega q_{\hat{n}(1,0)_{i}}.
\end{align}

\begin{widetext}
Next, we consider type~\ref{enum:input_3} inputs $\mathbf{A}^2$.
Here, all candidates are
$\mathbf{A}^2_i=I(x_1,y_1)_iI(x_2,y_2)_{i+1}$ with
$(x_1,y_1),(x_2,y_2) \neq(0,0)$. Suppose first that $x_1\neq 0$. For the
output
$\hat{n}(x_1-1,y_1)_i I(x_2,y_2)_{i+1}$, the relevant
commutators are
\begin{align}
[I(x_1,y_1)_iI(x_2,y_2)_{i+1},\hat{n}(0,1)_i]
&=-x_1\hat{n}(x_1-1,y_1)_iI(x_2,y_2)_{i+1},
\\
[I(x_1-1,y_1+1)_iI(x_2,y_2)_{i+1},\hat{n}(1,0)_i]
&=(y_1+1)\hat{n}(x_1-1,y_1)_iI(x_2,y_2)_{i+1}.
\end{align}
Thus
\begin{align}
q_{I(x_1,y_1)_iI(x_2,y_2)_{i+1}}\propto q_{I(x_1-1,y_1+1)_iI(x_2,y_2)_{i+1}}.
\end{align}
An analogous argument at site $i+1$ gives the corresponding relation for
$x_2$. Repeating these comparisons, any coefficient of
$\mathbf{A}^2_i=I(x_1,y_1)_iI(x_2,y_2)_{i+1}$ is proportional to that of
\begin{align}
\mathbf{A'}^2_i=I(0,x_1+y_1)_iI(0,x_2+y_2)_{i+1}.
\end{align}
It remains to show that this proportional coefficient vanishes. For the
left-site phonon-number term, we have
\begin{align}
[\hat{b}^{x_1+y_1}_i\hat{b}^{x_2+y_2}_{i+1},\hat{b}^\dagger_i\hat{b}_{i}]
&=(x_1+y_1)\hat{b}^{x_1+y_1}_i\hat{b}^{x_2+y_2}_{i+1}.
\end{align}
The analogous contribution from the right-site phonon-number term is
\begin{align}
[\hat{b}^{x_1+y_1}_i\hat{b}^{x_2+y_2}_{i+1},\hat{b}^\dagger_{i+1}\hat{b}_{i+1}]
&=(x_2+y_2)\hat{b}^{x_1+y_1}_i\hat{b}^{x_2+y_2}_{i+1}.
\end{align}
Therefore,
\begin{align}
\omega(x_1+y_1+x_2+y_2)q_{I(0,x_1+y_1)_iI(0,x_2+y_2)_{i+1}}&=0,
\end{align}
which implies $q_{I(0,x_1+y_1)_iI(0,x_2+y_2)_{i+1}}=0$.
Together with the proportionality relation above, this gives
$q_{I(x_1,y_1)_iI(x_2,y_2)_{i+1}}=0$.
Hence, all type~\ref{enum:input_3} candidates have zero coefficients.

Therefore, for coefficients of operators forming $Q^2$,
\begin{align}
q_{\hat{c}^\dagger_{i}\hat{c}_{i+1}}
:q_{\hat{c}_{i}\hat{c}^\dagger_{i+1}}
:q_{\hat{n}(0,1)_i}:q_{\hat{n}(1,0)_i}
:q_{I(1,1)_i}
&=t:-t:g:g:\omega
\end{align}
\begin{align}
q_{A^2}&=0\quad (\textnormal{otherwise})
\end{align}
holds. In other words, up to the freedom of adding 1-local conserved
quantities, every 2-local conserved quantity is a constant multiple
of the Hamiltonian.
\end{widetext}

\subsubsection{Proof for $k=1$ case}
First, we consider the candidates
$\mathbf{A}^1=\hat{c}(x,y), \hat{c}^\dagger(x,y), \hat{n}(x,y)$.
We discuss $\mathbf{A}^1_i=\hat{n}(x,y)_i$ explicitly. The commutator with the
right hopping term gives
\begin{align}
[\hat{n}(x,y)_i,\hat{c}^\dagger_i\hat{c}_{i+1}]
&=\hat{c}^\dagger(x,y)_i\hat{c}_{i+1}.
\end{align}
If $(x,y)\neq(0,0)$, this output is generated only by this input, and therefore
$q_{\hat{n}(x,y)_i}=0$. When $(x,y)=(0,0)$, however, the same output can also be
generated from the neighboring input:
\begin{align}
[\hat{n}(0,0)_{i+1},\hat{c}^\dagger_i\hat{c}_{i+1}]
&=-\hat{c}^\dagger_i\hat{c}_{i+1}.
\end{align}
Thus coefficient comparison gives
\begin{align}
q_{\hat{n}_{i}}-q_{\hat{n}_{i+1}}&=0.
\end{align}
Similarly, the commutators
\begin{align}
[\hat{c}(x,y)_i,\hat{c}^\dagger_i\hat{c}_{i+1}]
&=I(x,y)_i\hat{c}_{i+1},
\\
[\hat{c}^\dagger(x,y)_i,\hat{c}^\dagger_{i+1}\hat{c}_{i}]
&=-I(x,y)_i\hat{c}^\dagger_{i+1}
\end{align}
show that the coefficients of $\hat{c}(x,y)_i$ and
$\hat{c}^\dagger(x,y)_i$ vanish for $(x,y)\neq(0,0)$; the case
$(x,y)=(0,0)$ remains for now.

Next, we consider commutators that generate $\mathbf{B}^1$. Commutators comes from inputs that
do not commute with the particle-number operator, such as
$\mathbf{A}^1=\hat{c},\hat{c}^\dagger$,
include
\begin{align}
[\hat{c}_{i},\hat{n}(0,1)_i]&=\hat{c}(0,1)_i
\\
[\hat{c}^\dagger_{i},\hat{n}(0,1)_i]&=-\hat{c}^\dagger(0,1)_i
\end{align}
and each of these is the unique input for the corrsponding output, implying that these $\mathbf{A}^1$ have zero coefficients.

Finally, we treat $\mathbf{A}^1=I(x,y)$. 
When $x\neq0$, commutators that contribute to $\mathbf{B}^1_i=\hat{n}(x-1,y)_i$ are written as
\begin{align}
[I(x,y)_i,\hat{n}(0,1)_i]&=-x\hat{n}(x-1,y)_i,
\\
[I(x-1,y+1)_i,\hat{n}(1,0)_i]&=(y+1)\hat{n}(x-1,y)_i,
\end{align}
which gives $q_{I(x,y)_i}\propto q_{I(x-1,y+1)_i}$. Iterating this relation yields $q_{I(x,y)_i}\propto q_{I(0,x+y)_i}$. In addition, for 
$\mathbf{B}^1_i=I(0,x+y)_i$,
\begin{align}
[I(0,x+y)_i,I(1,1)_i]&=(x+y)I(0,x+y)_i
\end{align}
is the unique commutator that generates the output. Hence,
$q_{I(0,x+y)_i}=0$ holds. Together with the proportionality relation above,
this gives $q_{I(x,y)_i}=0$.

Therefore, among $1$-local conserved quantities, only constant multiples of the 
total fermion number 
$\hat{N}=\sum_i\hat{n}_{i}$ remain.
\end{proof}

\section{Holstein--Hubbard model}\label{sec:HH}
We now consider the one-dimensional Holstein--Hubbard model, whose Hamiltonian is
\begin{align}
 \hat{H}   &=  \hat{H}'_{\mathrm{hop}} +\hat{H}'_{\mathrm{int}}+\hat{H}_{\mathrm{pho}}+\hat{H}_{\mathrm{e-e}}  , 
\label{Eq.1D_HH_model}
\end{align}
with
\begin{align}
\hat{H}'_{\mathrm{hop}}&= t  \sum_{i=1}^L \sum_\sigma (\hat{c}^\dagger_{i,\sigma} \hat{c}_{i+1,\sigma} +h.c. ),
\end{align}
\begin{align}
\hat{H}'_{\mathrm{int}}&= g  \sum_{i=1}^L \sum_\sigma   \hat{n}_{i,\sigma}  \bigl (\hat{b}_i^\dagger + \hat{b}_i \bigr),
\end{align}
\begin{align}
\hat{H}_{\mathrm{pho}}&=  \omega \sum_{i=1}^L  \hat{b}^\dagger_i \hat{b}_i, 
\end{align}
\begin{align}
\hat{H}_{\mathrm{e-e}}&=\frac{U}{2}\sum_{i=1}^L\sum_{\sigma \neq \sigma'} \hat{n}_{i,\sigma} \hat{n}_{i,\sigma'},
\end{align}
where $\sigma,\sigma'\in \{\uparrow,\downarrow\}$ denote spin indices, and
$\hat{c}^\dagger_{i,\sigma}$ ($\hat{c}_{i,\sigma}$) is the creation (annihilation)
operator for an electron with spin $\sigma$ at site $i$.  We write
$\hat{n}_{i,\sigma}=\hat{c}^\dagger_{i,\sigma}\hat{c}_{i,\sigma}$ for the
corresponding fermion-number operator.  The parameter $U$ denotes the strength
of the Coulomb interaction between electrons with opposite spins.  When $U=0$,
i.e., in the absence of the Coulomb interaction term, the model reduces to the
spinful Holstein model.

The Holstein--Hubbard model is also an important electron--phonon coupled system. While the Holstein model is a minimal model describing electron--phonon interactions, the Holstein--Hubbard model is a minimal model describing the competition between those interactions and electron--electron interactions. This model is also believed to be nonintegrable and exhibits diffusive transport properties~\cite{HolsteinRev,HolsteinHubbard,HHcompetition,HHsuperconductivity,HHphasediagram,HHDMRG,HHthermalization}.

Before stating the result, let us recall the conserved quantities that are
trivial from the structure of the model. The energy, represented by the
Hamiltonian itself, is conserved by definition. In addition, the model conserves
the total fermion number. In the spinful case, this particle-number conservation
appears together with the global spin symmetry: the $1$-local quantities
\begin{align}
\hat{N}_{(\sigma,\sigma')}
\coloneqq \sum_{i=1}^L \hat{c}^\dagger_{i,\sigma}\hat{c}_{i,\sigma'}
\end{align}
commute with the Hamiltonian. The diagonal components
$\hat{N}_{(\uparrow,\uparrow)}$ and $\hat{N}_{(\downarrow,\downarrow)}$ are the
particle-number operators for each spin species, whose sum gives the total fermion
number. The off-diagonal components $\hat{N}_{(\uparrow,\downarrow)}$ and
$\hat{N}_{(\downarrow,\uparrow)}$ generate global spin-flip transformations.
These operators commute with the Hamiltonian because the hopping,
electron--phonon coupling, phonon, and onsite Coulomb terms are all invariant
under the corresponding global spin rotations.

The following theorem states that, in a sense similar to the result discussed in
Sec.~\ref{sec:main_result}, these trivial conserved quantities exhaust all local
conserved quantities of the one-dimensional Holstein--Hubbard model:

\begin{theorem}\label{t:HHwithtrivial}
In the one-dimensional Holstein--Hubbard model~\eqref{Eq.1D_HH_model}
satisfying $t\neq 0$, $g\neq 0$, and $\omega \neq 0$, $k$-local conserved quantities with $k<L/2$ are
restricted to linear combinations of $\hat{H}$ and
$\hat{N}_{(\sigma, \sigma')}$.
\end{theorem}

The proof of this theorem is presented in Appendix~\ref{sec:HH_proof}. Its structure is largely parallel to the proof of Theorem~\ref{withtrivial}, with the minor modifications arising from the additional spin degree of freedom. In particular, this theorem holds regardless of whether the Coulomb interaction is present.

\section{Higher dimensions}\label{sec:higher_dimensions}

We next extend the absence result to higher-dimensional hypercubic lattices.  Let
\begin{align}
\Lambda_d=\{1,2,\ldots,L\}^d
\end{align}
with periodic boundary conditions.  A site is denoted by
$\overrightarrow{r}=(i_1,i_2,\ldots,i_d)$, and $\overrightarrow{e}_\alpha$ denotes the unit vector in the
$\alpha$ direction.  The $d$-dimensional spinless Holstein Hamiltonian is
\begin{align}
\hat{H}&=\hat{H}_{\mathrm{hop}}+\hat{H}_{\mathrm{int}}
+\hat{H}_{\mathrm{pho}},
\label{Eq.dD_H_model}
\end{align}
with
\begin{align}
\hat{H}_{\mathrm{hop}}
&=t\sum_{\overrightarrow{r}\in\Lambda_d}\sum_{\alpha=1}^d
\left(\hat c^\dagger_{\overrightarrow{r}}\hat c_{\overrightarrow{r}+\overrightarrow{e}_\alpha}+h.c.\right),
\end{align}
\begin{align}
\hat{H}_{\mathrm{int}}
&=g\sum_{\overrightarrow{r}\in\Lambda_d}\hat n_{\overrightarrow{r}}
\left(\hat b^\dagger_{\overrightarrow{r}}+\hat b_{\overrightarrow{r}}\right),
\end{align}
\begin{align}
\hat{H}_{\mathrm{pho}}
&=\omega\sum_{\overrightarrow{r}\in\Lambda_d}\hat b^\dagger_{\overrightarrow{r}}\hat b_{\overrightarrow{r}}.
\end{align}
The corresponding Holstein--Hubbard Hamiltonian is
\begin{align}
\hat{H}&=\hat{H}'_{\mathrm{hop}}+\hat{H}'_{\mathrm{int}}
+\hat{H}_{\mathrm{pho}}+\hat{H}_{\mathrm{e-e}},
\label{Eq.dD_HH_model}
\end{align}
with
\begin{align}
\hat{H}'_{\mathrm{hop}}
&=t\sum_{\overrightarrow{r}\in\Lambda_d}\sum_{\alpha=1}^d\sum_\sigma
\left(\hat c^\dagger_{\overrightarrow{r},\sigma}
\hat c_{\overrightarrow{r}+\overrightarrow{e}_\alpha,\sigma}+h.c.\right),
\end{align}
\begin{align}
\hat{H}'_{\mathrm{int}}
&=g\sum_{\overrightarrow{r}\in\Lambda_d}\sum_\sigma\hat n_{\overrightarrow{r},\sigma}
\left(\hat b^\dagger_{\overrightarrow{r}}+\hat b_{\overrightarrow{r}}\right),
\end{align}
\begin{align}
\hat{H}_{\mathrm{pho}}
&=\omega\sum_{\overrightarrow{r}\in\Lambda_d}\hat b^\dagger_{\overrightarrow{r}}\hat b_{\overrightarrow{r}},
\end{align}
\begin{align}
\hat{H}_{\mathrm{e-e}}
&=\frac{U}{2}\sum_{\overrightarrow{r}\in\Lambda_d}\sum_{\sigma\ne\sigma'}
\hat n_{\overrightarrow{r},\sigma}\hat n_{\overrightarrow{r},\sigma'}.
\end{align}
Here the operators and parameters are the same as in Eqs.~\eqref{Eq.1D_H_model}
and \eqref{Eq.1D_HH_model}, with the site index replaced by $\overrightarrow{r}\in\Lambda_d$
and the hopping taken over all coordinate directions.  In the Holstein--Hubbard model,
$\sigma\in\{\uparrow,\downarrow\}$ is the spin index.

Higher-dimensional versions of these models have also been widely studied.
In their spinful versions, electron--phonon coupling leads to competition
between charge order and superconductivity~\cite{Holstein2DScalettar}, while
the Hubbard interaction introduces competing magnetic
correlations~\cite{HH2DBerger,HH2DNowadnick}.

We define locality by the largest coordinate-wise extent of the support, following the
grading used for higher-dimensional lattice systems~\cite{2DXY,2DHubbard}.  For a nonempty
set $S\subset\Lambda_d$, let $\operatorname{Wid}_\alpha S$ be the minimum length of a
periodic interval containing the projection of $S$ onto the $\alpha$-th coordinate axis,
and set
\begin{align}
\operatorname{Wid}S=\max_{1\leq\alpha\leq d}\operatorname{Wid}_\alpha S.
\label{Eq.dD_width}
\end{align}
For an operator-basis product $\mathbf{A}$, let $\operatorname{Supp}\mathbf{A}$ be the
set of sites on which it acts nontrivially, and write
$\operatorname{Wid}\mathbf{A}=\operatorname{Wid}(\operatorname{Supp}\mathbf{A})$.
An operator is called a \textit{$k$-support operator} if the width of its support is at
most $k$. A \textit{$k$-local quantity} is then an operator that can be expressed as a
sum of $k$-support operators, but cannot be expressed as a sum of $(k-1)$-support
operators. A \textit{$k$-local conserved quantity} is a $k$-local quantity that
commutes with the Hamiltonian. No translation or other symmetry is imposed on its
coefficients. For $d=1$, this definition reduces to the
contiguous-support definition used in Sec.~\ref{sec:main_result}.

With this definition, the one-dimensional absence theorem extends as follows.
\begin{theorem}\label{t:higher_dimensions}
For every $d \geq 2$, the $d$-dimensional Holstein
model~\eqref{Eq.dD_H_model} with $t \neq 0$, $g \neq 0$, and $\omega \neq 0$
has no $k$-local conserved quantity with $3 \leq k < L/2$.  The same holds
for the Holstein--Hubbard model~\eqref{Eq.dD_HH_model} for arbitrary $U$.
\end{theorem}

The proof is presented in Appendix~\ref{sec:higher_dimensions_proof}.  The maximal-width
direction first restricts every possible input to the same endpoint forms as in one
dimension.  The remaining coefficient relations then force the transverse coordinates of
the two endpoints to coincide, after which the one-dimensional Holstein or
Holstein--Hubbard proof applies.  We give the detailed argument
for $d=2$ and explain its extension to $d \geq 3$.

\section{Discussion}\label{sec:discussion}
In this paper, we established the nonintegrability of the one-dimensional Holstein and Holstein--Hubbard models by proving the absence of $k$-local conserved quantities with $k < L/2$, apart from the Hamiltonian itself and the total fermion-number operators. We further proved the absence of nontrivial $k$-local conserved quantities for their higher-dimensional hypercubic counterparts. Our results imply that phenomena such as ballistic transport, when observed in these systems, cannot be attributed to previously unidentified local conserved quantities. To the best of our knowledge, this is the first proof of nonintegrability for fermion--boson coupled systems in terms of the absence of local conserved quantities.

One limitation of the present work is that it addresses only strictly local conserved quantities. Our elimination procedure begins with operators of maximal support, whereas quasi-local conserved quantities contain tails extending to arbitrarily large distances and therefore fall outside the present framework~\cite{ProsenQuasilocal,IlievskiCompleteGGE,IlievskiQuasilocalReview,2DXY}. Determining whether such conserved quantities exist in the Holstein and Holstein--Hubbard models will require a different approach and remains an open problem.

An important direction for future work is to extend our method to other physically relevant fermion--boson coupled systems. In the Holstein model, the fermion--boson coupling appears as an onsite term, which plays a crucial role in the structure of the proof. It would therefore be highly desirable to generalize the argument to models in which the coupling enters two-body hopping terms, as in the Su--Schrieffer--Heeger (SSH) model~\cite{SSH}. Another important extension is to go beyond the dispersionless optical phonons considered in the present work. In particular, it would be interesting to investigate models with acoustic phonons, where the phonon sector can generate spatially extended and effectively many-body bosonic terms in real space. Clarifying whether similar nonintegrability proofs can be established in such settings would broaden the scope of rigorous results on conservation laws in coupled quantum many-body systems.

\begin{acknowledgments}
The authors thank Naoto Shiraishi for fruitful discussions, Yuuya Chiba for insightful comments on the $\omega=0$ limit, and Hosho Katsura for pointing us to some useful references.
We also appreciate Kota Mitsumoto, Souta Shimozono, and Kei Oriyose for helping us perform numerical calculations.
M.Y. was supported by KAKENHI Grant No. JP25KJ0815 from JSPS.
\end{acknowledgments}

\appendix

\section{Integrable cases}\label{sec:integrable_case}

As shown in Theorem~\ref{nontrivial}, the Holstein model is nonintegrable when
$t$, $g$, and $\omega$ are all nonzero. By contrast, when at least one of these
parameters vanishes, the model acquires an extensive number of independent local
conserved quantities and becomes integrable in this sense. In particular, the
limits $t \to 0$ and $g \to 0$ correspond to the atomic limit and the weak-coupling
limit, respectively, both of which are known to be exactly solvable~\cite{LangFirsov,HolsteinETH}.
The case $\omega=0$ is less physical, but it also possesses an extensive number
of conserved quantities. As we discuss below, this case is moreover clearly
distinguished from the nonintegrable regime by standard indicators of quantum chaos.

\subsection{The case $t=0$}

When $t=0$, the Hamiltonian contains only onsite terms. In this case, it can be
diagonalized by the Lang--Firsov transformation~\cite{LangFirsov}. Introducing the operator
\begin{align}
\hat{S}&=\frac{g}{\omega}\sum_i \hat{n}_i(\hat{b}_i^\dagger-\hat{b}_i),
\end{align}
and defining the transformed Hamiltonian by
\begin{align}
\hat{H}' &= e^{\hat{S}} \hat{H} e^{-\hat{S}},
\end{align}
one obtains
\begin{align}
\hat{H}' &=
-\frac{g^2}{\omega}\sum_i \hat{n}_i^2
+\omega\sum_i \hat{b}_i^\dagger \hat{b}_i \nonumber\\
&=
-\frac{g^2}{\omega}\sum_i \hat{c}_i^\dagger \hat{c}_i
+\omega\sum_i \hat{b}_i^\dagger \hat{b}_i.
\end{align}
Thus the transformed Hamiltonian is written as a sum of independent quadratic
fermionic and bosonic terms. Correspondingly, for $\hat{H}'$, both the fermion
number operator $\hat{n}_i=\hat{c}_i^\dagger\hat{c}_i$ and the boson number
operator $\hat{b}_i^\dagger\hat{b}_i$ are conserved at each site. Applying the
inverse Lang--Firsov transformation maps these conserved quantities back to
local conserved quantities of the original Hamiltonian. To see explicitly why
locality is preserved, write $\hat{S}=\sum_i\hat{S}_i$, where the onsite
generators $\hat{S}_i$ mutually commute. The transformation then factorizes as
$e^{\pm\hat{S}}=\prod_i e^{\pm\hat{S}_i}$, so conjugating an operator supported
on site $i$ cannot enlarge its spatial support. For example, the inverse
transformation gives
$e^{-\hat{S}}\hat{b}_i^\dagger\hat{b}_i e^{\hat{S}}
=(\hat{b}_i^\dagger+g\hat{n}_i/\omega)
(\hat{b}_i+g\hat{n}_i/\omega)$, which remains an onsite operator. Hence the original
Holstein model in this limit is integrable, possessing an extensive number of
independent local conserved quantities.

\subsection{The case $g=0$}
When $g=0$, the Hamiltonian reduces to
\begin{align}
\hat{H}&=
t \sum_i (\hat{c}^\dagger_i \hat{c}_{i+1}+\hat{c}^\dagger_{i+1}\hat{c}_i)
+\omega \sum_i \hat{b}_i^\dagger \hat{b}_i,
\end{align}
namely, a sum of a free-fermion tight-binding Hamiltonian and localized phonons.
This model is therefore directly diagonalizable.

As for local conserved quantities, the boson number operator at each site $\hat{b}_i^\dagger \hat{b}_i$ is
obviously conserved. In addition, for the fermionic sector, one finds for each
$k\geq 3$ that $k$-local quantity
\begin{align}
\hat{Q}^k =& \sum_i
q_{c^\dagger c}\hat{c}^\dagger_i \hat{c}_{i+k-1}
+q_{c c^\dagger}\hat{c}^\dagger_{i+k-1}\hat{c}_i \nonumber\\
&+\sum_i (-1)^i
\left(
q_{c^\dagger c^\dagger}\hat{c}^\dagger_i \hat{c}^\dagger_{i+k-1}
+q_{c c}\hat{c}_i \hat{c}_{i+k-1}
\right)
\end{align}
is conserved, where
$q_{c^\dagger c}$, $q_{c c^\dagger}$,
$q_{c^\dagger c^\dagger}$, and $q_{c c}$ are site-independent constants.

The staggered $\hat c\hat c$- and $\hat c^\dagger\hat c^\dagger$-type pairing
terms in this expression require even $L$ under periodic boundary conditions.
The particle-conserving $\hat c^\dagger\hat c$-type terms remain conserved
for either parity of $L$.

\subsection{The case $\omega=0$}
When $\omega=0$, it is also easy to see that the system has an extensive number
of independent local conserved quantities. Indeed, at each site the operator
$\hat{X}_i\coloneq\hat{b}_i^\dagger+\hat{b}_i$, which corresponds to the phonon displacement,
commutes with the Hamiltonian. In this sense, the system is integrable when
counted by the number of independent local conserved quantities.

This limit can also be solved explicitly by using the joint spectral
decomposition of the mutually commuting operators $\{\hat{X}_i\}$. Each
$\hat{X}_i$ has the whole real line as its continuous spectrum. Thus, in a
joint generalized eigenstate specified by
$\hat{X}_i\lvert\boldsymbol{x}\rangle=x_i
\lvert\boldsymbol{x}\rangle$, the Hamiltonian acting on the fermionic degrees
of freedom reduces to
\begin{align}
\hat{H}(\boldsymbol{x})=
t\sum_i(\hat{c}_i^\dagger\hat{c}_{i+1}
+\hat{c}_{i+1}^\dagger\hat{c}_i)
+g\sum_i x_i\hat{n}_i.
\end{align}
For each fixed configuration $\boldsymbol{x}\in\mathbb{R}^L$, this is a
quadratic fermionic Hamiltonian in a static onsite potential and is therefore
reduced to a single-particle diagonalization problem. The full Hamiltonian is
the corresponding direct integral over $\boldsymbol{x}$. Consequently, its
spectrum is continuous and, for $g\neq0$ in any sector containing at least one
fermion, is unbounded both above and below as the values of $x_i$ vary.

\subsection{Numerical analysis of level statistics}
To illustrate the parameter dependence of integrability, we numerically investigate the level spacing statistics~\cite{quantumchaosRev,chaosRMT} of the Holstein model by exact diagonalization. 
For the ordinary one-dimensional Holstein model with periodic boundary
conditions, both the total fermion number
$\hat{N}=\sum_i \hat{n}_i$ and the crystal momentum $k$ associated with lattice
translation symmetry are conserved. The model also has time-reversal symmetry as
an antiunitary symmetry. Previous work~\cite{HolsteinETH} studied the energy-level statistics
in the symmetry sector with fermion number $\hat{N}=1$ and momentum
$k=2\pi/L$, and found that the resulting distribution is close to that of the
Gaussian orthogonal ensemble (GOE)~\cite{quantumchaosRev,chaosRMT}.

Following that analysis, we again focus on the $\hat{N}=1$ sector and
consider the momentum sector $k=2\pi/L$. Taking 
$t=-1.0$, $g=1/\sqrt{2}$, system size $L=7$, and a boson cutoff $M=3$, meaning that the maximal number of phonons per site is three, we restrict attention to the energy window
\begin{align}
\frac{E_{\mathrm{av}} - E_i}{E_{\mathrm{av}} - E_{\mathrm{min}}}
< \frac{2}{3},
\qquad (E_i < E_{\mathrm{av}}),
\end{align}
\begin{align}
\frac{E_i - E_{\mathrm{av}}}{E_{\mathrm{max}} - E_{\mathrm{av}}}
< \frac{2}{3},
\qquad (E_i > E_{\mathrm{av}}),
\end{align}
where $E_{\mathrm{av}}$, $E_{\mathrm{min}}$, and $E_{\mathrm{max}}$ denote the
average, minimum, and maximum eigenenergies, respectively, within the given momentum sector. 
Reference~\cite{HolsteinETH} studied the corresponding setting for $L=8$,
but here we instead used $L=7$ because of computational limitations.

For a nonintegrable case such as $\omega=0.5$,  the distribution of neighboring energy-level spacings $P(s)$~\cite{quantumchaosRev,levelspacing} and the ratio of consecutive level spacings $P(r)$~\cite{levelspacingratio,levelspacingratio2} are plotted in Fig.~\ref{fig:level_stats}(a) and (b), where $s$ denotes the unfolded neighboring energy-level spacing $s_i \propto E_{i+1}-E_i$ and $r$ denotes the ratio of consecutive level spacings
\begin{align}
    r_i \coloneq \min\{
\frac{E_{i+2}-E_{i+1}}{E_{i+1}-E_i},
\frac{E_{i+1}-E_i}{E_{i+2}-E_{i+1}}
\}.
\end{align}
Consistently with the
previous literature, both distributions are well described by the GOE prediction.

By contrast, for the integrable case $\omega=0$, the
corresponding distribution of neighboring energy-level spacings and the ratio of consecutive level spacings are shown in
Fig.~\ref{fig:level_stats}(c) and (d). In this case both distributions agree well with Poisson statistics, supporting the view that the system behaves in a manner analogous to
more conventional integrable models.

\begin{figure*}[t]
    \centering
    \begin{minipage}[t]{0.48\textwidth}
        \centering
        \includegraphics[width=\linewidth]{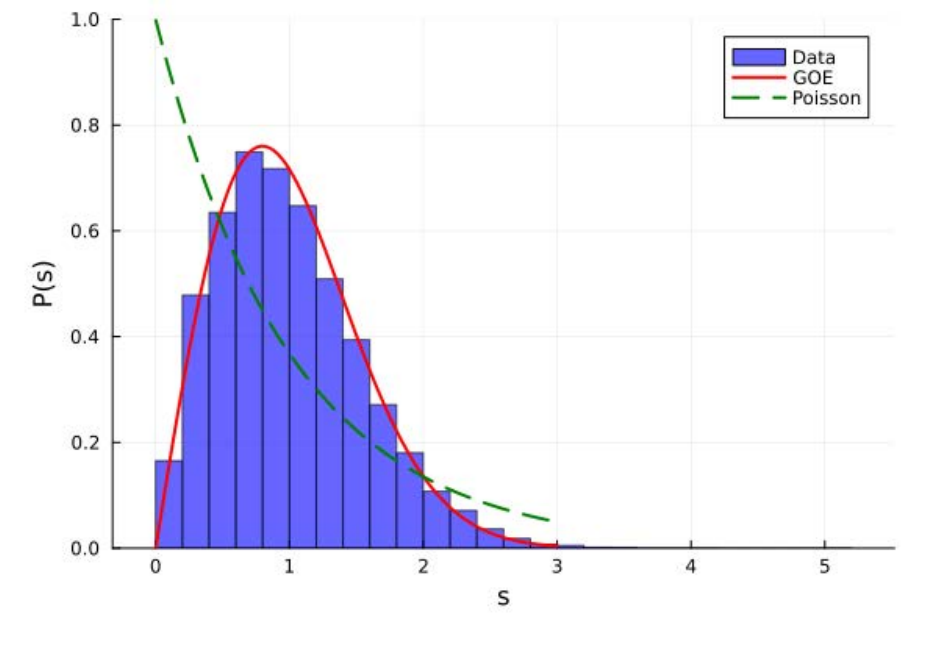}
        \\[-0.5ex]
        \small (a) Neighboring energy-level spacings for $\omega=0.5$.
        \label{fig:Ps_L7_k1}
    \end{minipage}\hfill
    \begin{minipage}[t]{0.48\textwidth}
        \centering
        \includegraphics[width=\linewidth]{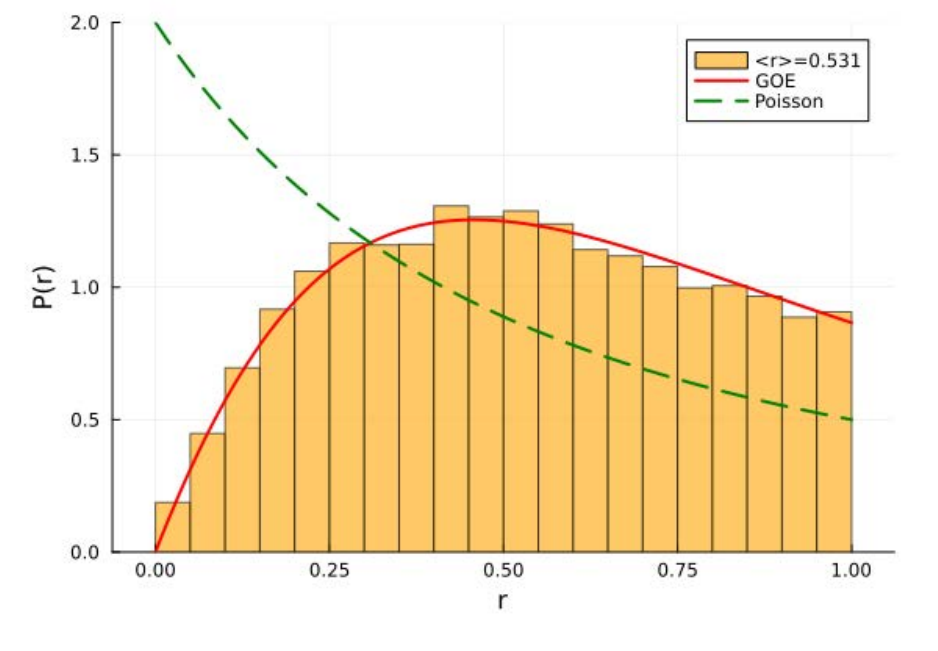}
        \\[-0.5ex]
        \small (b) Ratio of consecutive level spacings for $\omega=0.5$.
        \label{fig:Pr_L7_k1}
    \end{minipage}

    \vspace{1ex}

    \begin{minipage}[t]{0.48\textwidth}
        \centering
        \includegraphics[width=\linewidth]{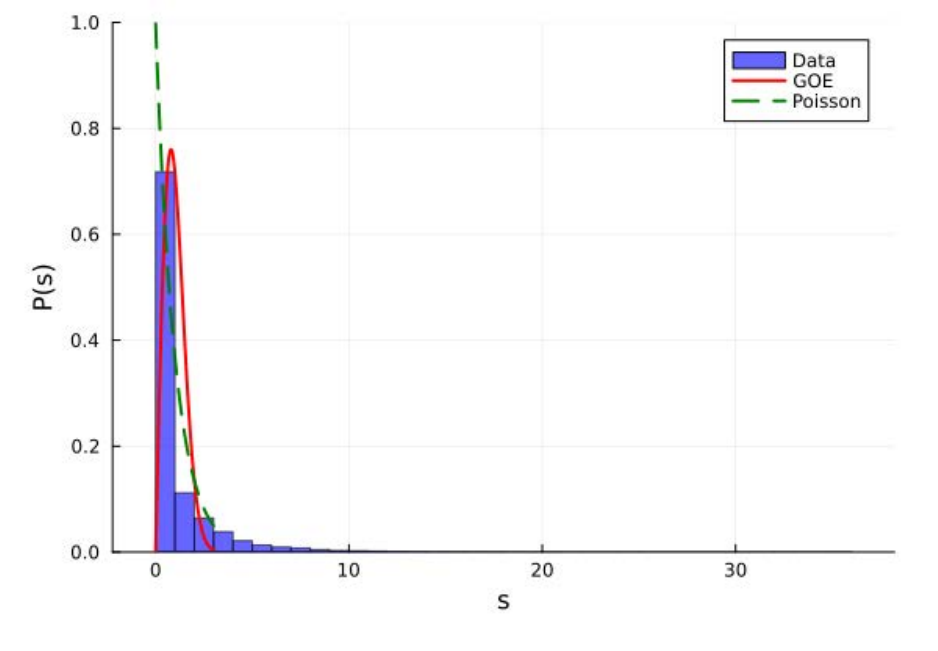}
        \\[-0.5ex]
        \small (c) Neighboring energy-level spacings for $\omega=0.0$.
        \label{fig:Ps_L7_g0_k1}
    \end{minipage}\hfill
    \begin{minipage}[t]{0.48\textwidth}
        \centering
        \includegraphics[width=\linewidth]{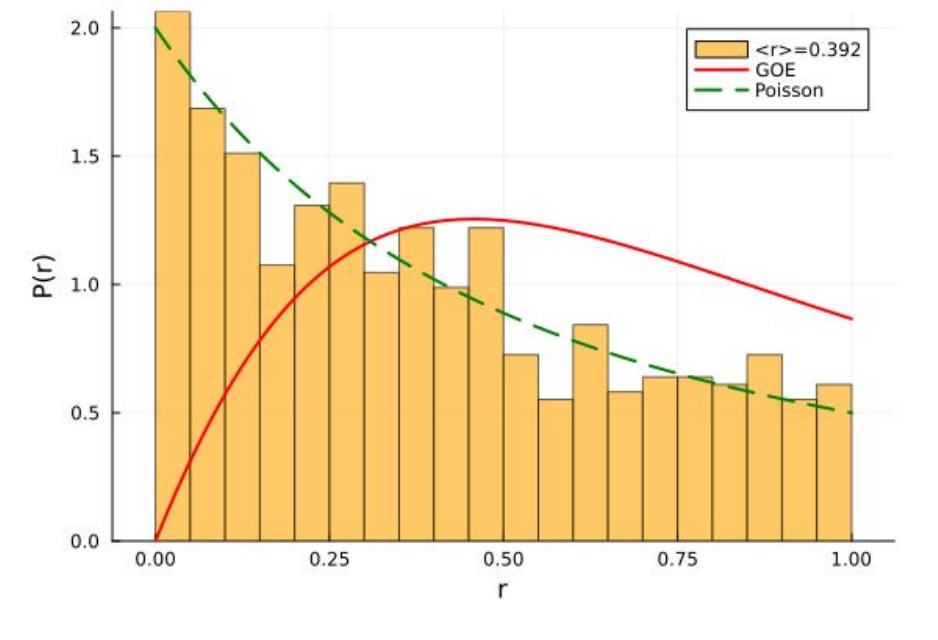}
        \\[-0.5ex]
        \small (d) Ratio of consecutive level spacings for $\omega=0.0$.
        \label{fig:Pr_L7_g0_k1}
    \end{minipage}
    \caption{Level-spacing statistics of the Holstein model~\eqref{Eq.1D_H_model}. Panels (a) and (b) show the nonintegrable case $\omega=0.5$, while panels (c) and (d) show the integrable case $\omega=0.0$.}
    \label{fig:level_stats}
\end{figure*}

\section{The model with a fermionic chemical potential}
\label{sec:withchem}
Now, we consider adding a fermionic chemical potential to the Holstein model~\eqref{Eq.1D_H_model},
\begin{align}
\hat{H}   &=  \hat{H}_{\mathrm{hop}} +\hat{H}_{\mathrm{int}}+\hat{H}_{\mathrm{pho}} +\mu \sum_i\hat{n}_i.
\label{Eq.withchem}
\end{align}
The local conserved quantities can also be classified completely in this case.

\begin{theorem}\label{t:chem}
In the one-dimensional Holstein model with fermion chemical
potential~\eqref{Eq.withchem} satisfying $t\neq 0$, $g\neq 0$, and
$\omega \neq 0$, $k$-local conserved quantities with $k<L/2$ are
restricted to linear combinations of the Hamiltonian $\hat H$ and the total
fermion number $\hat N=\sum_i\hat n_i$.
\end{theorem}

We prove the theorem separately for $3\leq k<L/2$ and for
$1\leq k\leq2$. The two parts follow the proofs of
Theorems~\ref{nontrivial} and~\ref{trivial}, respectively, with minor
modifications due to the chemical-potential term.

\begin{proof}[Proof of Theorem~\ref{t:chem}]
\medskip\noindent\underline{\textit{The case $3\leq k<L/2$.}}\quad
As in Theorem~\ref{nontrivial}, we consider three types of $\mathbf{A}^k$ and
show that the coefficients of each type vanish.

First, we consider inputs contributing to $\mathbf{B}^{k+1}$. Because the chemical
potential term is $1$-local, a procedure similar to step~1 of Theorem~1 applies,
and the remaining type~\ref{enum:input_1} candidates are
$\mathbf{A}^k_i=\hat{c}^\dagger_i \hat{c}_{i+k-1},\hat{c}_i \hat{c}^\dagger_{i+k-1},\hat{c}^\dagger_i \hat{c}^\dagger_{i+k-1},\hat{c}_i \hat{c}_{i+k-1}$.
It is also shown that type~\ref{enum:input_2} candidates have zero
coefficients.

Next, consider inputs for $\mathbf{B}^{k}$ and focus on type~\ref{enum:input_1}
candidates. In this step, the presence or absence of the chemical-potential
term does not affect the conclusion. For example, when examining inputs related
to $\mathbf{A}^k_i=\hat{c}_i\hat{c}_{i+k-1}$, one may also consider the
chemical-potential contribution
\begin{align}
[\hat{c}_i\hat{c}(0,1)_{i+k-1},\hat{n}_{i+k-1}]&=\hat{c}_i\hat{c}(0,1)_{i+k-1},
\end{align}
which outputs $\mathbf{B}^k_i=\hat{c}_i\hat{c}(0,1)_{i+k-1}$. However,
$\mathbf{A}^k_i=\hat{c}_i\hat{c}(0,1)_{i+k-1}$ is a type~\ref{enum:input_1}
candidate whose coefficient has already been shown to vanish in step~1.
Therefore, a procedure similar to step~2 of Theorem~1 remains valid, leaving
$\mathbf{A}^k_i=\hat{c}^\dagger_i \hat{c}_{i+k-1},\hat{c}_i \hat{c}^\dagger_{i+k-1}$
as type~\ref{enum:input_1} candidates, together with related
$(k-1)$-support candidates such as
$\mathbf{A}^{k-1}=\hat{c}^\dagger_{i+1}\hat{c}(0,1)_{i+k-1}$.

For the remaining type~\ref{enum:input_1} candidates, the chemical-potential
term can in principle give additional contributions because it consists of the
particle-number operators. A commutator with
$\hat{H}_{\mathrm{chem}}=\mu\sum_i \hat{n}_i$ preserves the support and the
operator structure of the input, up to a sign. Hence, when a certain output
$\mathbf{B}$ is considered, an input $\mathbf{A}$ having an operator form corresponding to $\mathbf{B}$ may also contribute through the chemical-potential term.
However, the possible outputs $\mathbf{B}^{k}$ and $\mathbf{B}^{k-1}$ relevant
to steps 2 and 3 of Theorem~1 have fermionic parts consisting of one creation
operator and one annihilation operator at two distinct sites. If this fermionic
part is represented by $\hat{c}^\dagger_i\hat{c}_j$ with $i\neq j$, the only
number operators in $\hat{H}_{\mathrm{chem}}$ that do not commute with it are
$\hat{n}_i$ and $\hat{n}_j$, and they act with opposite signs:
\begin{align}
[\hat{c}^\dagger_i\hat{c}_j,\hat{n}_i]&=-\hat{c}^\dagger_i\hat{c}_j,
\\
[\hat{c}^\dagger_i\hat{c}_j,\hat{n}_j]&=\hat{c}^\dagger_i\hat{c}_j
\end{align}
for $i\neq j$. Therefore, the possible chemical-potential contributions to such
an output cancel between the two fermionic endpoints. Equivalently, the
remaining type~\ref{enum:input_1} candidates are neutral with respect to the
total particle number. Thus, although inputs with a form corresponding to
$\mathbf{B}$ may appear as candidates, their commutators with the chemical
potential term do not produce a net contribution. The chemical-potential term
therefore gives no additional independent coefficient equations and does not
modify the coefficient relations obtained from the hopping and electron--phonon
terms. By following procedures analogous to steps 2 and 3 of Theorem~1, the
remaining type~\ref{enum:input_1} candidates are also shown to have zero
coefficients.

Finally, consider type~\ref{enum:input_3} candidates. Again, because the chemical
potential term is 1-support, a procedure similar to step~2 of Theorem~1 leaves
type~\ref{enum:input_3} candidates of the form
$\mathbf{A}^k_i=I(x_1,y_1)_i\cdots I(x_j,y_j)_{i+j-1}\cdots I(x_k,y_k)_{i+k-1}$.
These commute with the chemical-potential term. Therefore, applying an analogous
step~2 argument once more, this type is also shown to have zero coefficients.

\medskip\noindent\underline{\textit{The case $1\leq k\leq2$.}}\quad
We next adapt the proof of Theorem~\ref{trivial}.  Write
\begin{align}
\hat H_0&\coloneq \hat H_{\mathrm{hop}}+\hat H_{\mathrm{int}}
+\hat H_{\mathrm{pho}}=\hat H-\mu\hat N .
\end{align}
Because every term in $\hat H_0$ conserves the total fermion number,
$[\hat H_0,\hat N]=0$.  Hence both $\hat H_0$ and $\hat N$ commute with
$\hat H=\hat H_0+\mu\hat N$.

Consider first a 2-local conserved quantity.  The 3-support outputs generated
by its commutator with the hopping terms are exactly the same as in the
$k=2$ proof of Theorem~\ref{trivial}, since the chemical-potential term is
1-local and cannot generate such outputs.  Thus all type~\ref{enum:input_2}
coefficients vanish, and the only type~\ref{enum:input_1} inputs that can
remain at this stage have undressed endpoints
$\hat c_i,\hat c_i^\dagger$ and
$\hat c_{i+1},\hat c_{i+1}^\dagger$.

The chemical-potential term also leaves unchanged the subsequent coefficient
comparisons.  The particle-conserving hopping inputs
$\hat c_i^\dagger\hat c_{i+1}$ and
$\hat c_i\hat c_{i+1}^\dagger$ commute with $\hat N$.  For the outputs used
in the proof of Theorem~\ref{trivial} to eliminate the particle-nonconserving hopping
inputs, a chemical-potential commutator could contribute only from an input
with the same fermionic and bosonic operator structure as the output.  Such
an input contains a phonon-dressed fermionic endpoint and its coefficient has
already been eliminated by the 3-support-output argument above.  Finally,
the phonon-dressed number inputs and all type~\ref{enum:input_3} inputs commute
with $\hat N$.  Therefore no new term enters any of the independent
coefficient equations, and the relations obtained in the proof of
Theorem~\ref{trivial}
remain
\begin{align}
q_{\hat{c}^\dagger_{i}\hat{c}_{i+1}}
:q_{\hat{c}_{i}\hat{c}^\dagger_{i+1}}
:q_{\hat{n}(0,1)_i}:q_{\hat{n}(1,0)_i}
:q_{I(1,1)_i}
&=t:-t:g:g:\omega ,
\end{align}
with all other 2-support coefficients equal to zero.  Consequently, every
2-local conserved quantity is a constant multiple of $\hat H_0$, up to the
addition of a 1-local conserved quantity.

For a 1-local conserved quantity, the 2-support outputs produced by the
hopping terms are again unaffected by the 1-local chemical potential.  As in
the proof of Theorem~\ref{trivial}, they eliminate all phonon-dressed number inputs and
force the coefficients of the undressed number operators $\hat n_i$ to be
uniform.  They also eliminate the phonon-dressed single-fermion inputs.
The remaining undressed inputs $\hat c_i$ and $\hat c_i^\dagger$ are excluded
by the outputs generated with the electron--phonon term.  A possible
chemical-potential contribution to either output would require the
corresponding phonon-dressed single-fermion input, whose coefficient has
already vanished.  Purely bosonic inputs commute with the chemical-potential
term and are eliminated by the same electron--phonon and phonon-number
comparisons as before.  Hence the only 1-local conserved quantity is a
constant multiple of $\hat N$.

Combining the two cases, a short-support conserved quantity has the form
\begin{align}
a\hat H_0+b\hat N
=a\hat H+(b-a\mu)\hat N .
\end{align}
Thus the apparent freedom to add or subtract the chemical-potential term does
not yield another independent conserved quantity: the full space is spanned
by $\hat H$ and $\hat N$.  Together with the $k\geq3$ result above, this proves
the theorem.
\end{proof}

\section{Detailed argument for the Holstein--Hubbard model (Proof of Theorem~\ref{t:HHwithtrivial})}
\label{sec:HH_proof}

In this appendix, we present the proofs of the theorems on the Holstein--Hubbard model that were omitted in the main text. More precisely, we prove Theorem~\ref{t:HHwithtrivial}, which classifies all $k$-local conserved quantities with $k<L/2$ and states that there is no nontrivial $k$-local conserved quantity for $3\leq k<L/2$. We prove this theorem by dividing it into the following lemmas for the cases $k \geq 3$ and $1 \leq k \leq 2$.

The argument closely parallels that for the Holstein model in Secs.~\ref{sec:proof} and \ref{sec:short}; the only new ingredients are the spin degree of freedom and the onsite Coulomb interaction.

\begin{lemma}\label{t:HH}
The one-dimensional Holstein--Hubbard model~\eqref{Eq.1D_HH_model} with
$t\neq 0$, $g\neq 0$, and $\omega \neq 0$ has no $k$-local conserved quantity
with $3\leq k<L/2$.
\end{lemma}

Note that this theorem holds regardless of whether the Coulomb interaction is present.

\begin{lemma}\label{t:HHtrivial}
In the one-dimensional Holstein--Hubbard model~\eqref{Eq.1D_HH_model} with
$t\neq 0$, $g\neq 0$, and $\omega \neq 0$, the only independent
$k$-local conserved quantities for $1\leq k\leq 2$ are
$\hat{H}$ ($k=2$) and $\hat{N}_{(\sigma,\sigma')}$ ($k=1$), where $\sigma, \sigma'$ denote spin indices.
\end{lemma}

Combining Lemma~\ref{t:HH} and Lemma~\ref{t:HHtrivial}, we obtain Theorem~\ref{t:HHwithtrivial}.

\subsection{Proof of Lemma~\ref{t:HH}}
\begin{proof}
The proof follows a strategy similar to that of Theorem~\ref{nontrivial}. As a spinful extension of the spinless basis in Eqs.~\eqref{spinlessbasis_1} and \eqref{spinlessbasis_2}, we use the following $l$-support basis starting from site $i$:
\begin{widetext}
\begin{align}
\mathbf{A}^l_i, \mathbf{B}^l_i&=  e_ie_{i+1}\dots e_{i+l-1},
\label{spinfulbasis_1}
\end{align}
\begin{align}
e_j \in \left\{ f_{j,\uparrow} f_{j,\downarrow}
(\hat{b}^\dagger_j)^x\hat{b}_j^y \mid
f_{j,\sigma}\in\{I,\hat{c}_{j,\sigma},\hat{c}^\dagger_{j,\sigma},
\hat{n}_{j,\sigma}\},\,
\sigma\in\{\uparrow,\downarrow\},\,
x,y\in\mathbb{Z}_{\geq 0}\right\},
\label{spinfulbasis_2}
\end{align}
\end{widetext}
with $e_i,e_{i+l-1}\neq I(\hat{b}^\dagger_j)^0\hat{b}_j^0$.

Also, we denote by $f_{\uparrow}f_{\downarrow}(x,y)_j$ the product of fermionic
operators $f_{j,\uparrow}f_{j,\downarrow}$ and the bosonic operator
$(\hat{b}^\dagger_j)^x\hat{b}_j^y$. When an operator acts trivially on bosons
(i.e., when $x=y=0$), we often denote $f_{\uparrow}f_{\downarrow}(0,0)_j$
simply as $f_{\uparrow}f_{\downarrow}$.

As in the proof of Theorem~\ref{nontrivial}, we divide the analysis into three types. We first examine inputs that generate $\mathbf{B}^{k+1}$. By an argument analogous to step~1 in the proof of Theorem~\ref{nontrivial}, the remaining type~\ref{enum:input_1}
candidates are
$\mathbf{A}^k_i=\hat{c}^\dagger_{i,\sigma} \hat{c}_{i+k-1,\sigma'},\hat{c}_{i,\sigma} \hat{c}^\dagger_{i+k-1,\sigma'},\hat{c}^\dagger_{i,\sigma} \hat{c}^\dagger_{i+k-1,\sigma'},\hat{c}_{i,\sigma} \hat{c}_{i+k-1,\sigma'}$
with $\sigma,\sigma'\in\{\uparrow,\downarrow\}$. Likewise, type~\ref{enum:input_2}
candidates are shown to have zero coefficients.

We next consider inputs that generate
$\mathbf{B}^{k}=\hat{c}_{i,\sigma} \hat{c}_{\sigma'}(0,1)_{i+k-1}$. Taking into account the candidates left after the step~1 reduction, the Coulomb interaction does not affect this part of the argument. Indeed,
\begin{align}
[\hat{c}_{i,\sigma} 
\hat{c}_{i+k-1,\sigma'},\hat{n}_{\sigma'}(0,1)_{i+k-1}]&=\hat{c}_{i,\sigma} 
\hat{c}_{\sigma'}(0,1)_{i+k-1}
\\
[\hat{c}_{i+1,\sigma} 
\hat{c}_{\sigma'}(0,1)_{i+k-1},\hat{c}^\dagger_{i+1,\sigma}\hat{c}_{i,\sigma}]&=\hat{c}_{i,\sigma} 
\hat{c}_{\sigma'}(0,1)_{i+k-1}
\end{align}
exhaust all such inputs among the candidates that survived step~1. Therefore,
\begin{equation}
\label{eq:hh-type1-balance}
gq_{\hat{c}_{i,\sigma} \hat{c}_{i+k-1,\sigma'}}
+tq_{\hat{c}_{i+1,\sigma} \hat{c}_{\sigma'}(0,1)_{i+k-1}}=0
\end{equation}
follows.
Repeating a coefficient-comparison argument analogous to step~2 of Theorem~\ref{nontrivial} for the other outputs, with the spin indices carried along, we
obtain a similar conclusion: the coefficients of
$\mathbf{A}^k_i=\hat{c}^\dagger_{i,\sigma} \hat{c}^\dagger_{i+k-1,\sigma'},\hat{c}_{i,\sigma} \hat{c}_{i+k-1,\sigma'}$
vanish. Hence, among type~\ref{enum:input_1}, only
$\mathbf{A}^k_i=\hat{c}^\dagger_{i,\sigma} \hat{c}_{i+k-1,\sigma'},\hat{c}^\dagger_{i+k-1,\sigma'}\hat{c}_{i,\sigma}$
remain, and their coefficients are uniform in the site index $i$. Moreover,
for terms such as
$\mathbf{A}^{k-1}=\hat{c}^\dagger_{i+1,\sigma}\hat{c}_{\sigma'}(0,1)_{i+k-1},\hat{c}^\dagger_{i,\sigma}\hat{c}_{i+k-2,\sigma'}$, we obtain the following relations:
\begin{widetext}
\begin{align}
\frac{g}{t}q_{\hat{c}^\dagger_{i,\sigma} 
\hat{c}_{i+k-1,\sigma'}}
&=q_{\hat{c}^\dagger_{i+1,\sigma} 
\hat{c}_{\sigma'}(0,1)_{i+k-1}}
=q_{\hat{c}^\dagger_{i+2,\sigma}I(0,1)_{i+k-1}\hat{c}_{i+k,\sigma'}}
=\dots
=q_{\hat{c}^\dagger_{\sigma}(0,1)_{i+k-1} 
\hat{c}_{i+2k-3,\sigma'}}
=\frac{g}{t}q_{\hat{c}^\dagger_{i+k-1,\sigma} 
\hat{c}_{i+2k-2,\sigma'}},
\end{align}
\begin{align}
q_{\hat{c}^\dagger_{i,\sigma}\hat{c}_{i+k-2,\sigma'}}&=q_{\hat{c}^\dagger_{i+1,\sigma}\hat{c}_{i+k-1,\sigma'}}=q_{\hat{c}^\dagger_{i+2,\sigma}\hat{c}_{i+k,\sigma'}}= \cdots.
\end{align}

We now turn to inputs that generate
$\mathbf{B}^{k-1}=\hat{c}^\dagger_{i+1,\sigma} \hat{c}_{\sigma'}(0,1)_{i+k-1}$. Again, the Coulomb interaction does not modify the argument. We have
\begin{align}
[\hat{c}^\dagger_{i+1,\sigma} 
\hat{c}_{\sigma'}(0,1)_{i+k-1},I(1,1)_{i+k-1}]&=\hat{c}^\dagger_{i+1,\sigma} 
\hat{c}_{\sigma'}(0,1)_{i+k-1}\\
[\hat{c}^\dagger_{i+2,\sigma} 
\hat{c}_{\sigma'}(0,1)_{i+k-1},\hat{c}^\dagger_{i+1,\sigma}\hat{c}_{i+2,\sigma}]
&=-\hat{c}^\dagger_{i+1,\sigma} 
\hat{c}_{\sigma'}(0,1)_{i+k-1}
\\
[\hat{c}^\dagger_{i+1,\sigma} \hat{c}_{i+k-1,\sigma'},
\hat{n}_{\sigma'}(0,1)_{i+k-1}]
&=\hat{c}^\dagger_{i+1,\sigma} \hat{c}_{\sigma'}(0,1)_{i+k-1}
\end{align}
For $k=3$ and $\sigma=\sigma'$, the input
$\hat n_{i+1,\sigma}I(0,1)_{i+2}$ may also contribute to this output.
As in the proof of Lemma~\ref{lem:step3_k3}, its coefficient vanishes by
\begin{align}
&[\hat n_{i+1,\sigma}I(0,1)_{i+2},
\hat c^\dagger_{i,\sigma}\hat c_{i+1,\sigma}]
\nonumber\\
&\qquad=-\hat c^\dagger_{i,\sigma}\hat c_{i+1,\sigma}I(0,1)_{i+2}.
\end{align}
The right-hand side is a $3$-support output with no competing input
whose coefficient may be nonzero after the step~1 reduction.
Since $t\neq0$, this gives
$q_{\hat n_{i+1,\sigma}I(0,1)_{i+2}}=0$.
The reflected additional input at the other end is excluded in the same way.
Thus the three commutators above exhaust the relevant contributions.
Treating the other cases similarly, and following step~3 in the proof of Theorem~\ref{nontrivial}, we obtain
\begin{align}
q_{\hat{c}^\dagger_{i,\sigma} 
\hat{c}_{i+k-1,\sigma'}}
&=q_{\hat{c}^\dagger_{i+1,\sigma} 
\hat{c}_{\sigma'}(0,1)_{i+k-1}}=\dots=0,
\\
q_{\hat{c}^\dagger_{i+k-1,\sigma'} 
\hat{c}_{i,\sigma}}
&=q_{\hat{c}^\dagger_{i+k-1,\sigma'} 
\hat{c}_{\sigma}(0,1)_{i+1}}=\dots=0.
\end{align}
Hence, all type~\ref{enum:input_1} candidates have zero coefficients.

Finally, we treat type~\ref{enum:input_3}. We take $\mathbf{A}_i^{k} = I(x_1,y_1)_i \big( \prod^{k-2}_{j=2} A^j(x_j,y_j)_{i+j-1} \big) \hat{n}_\uparrow(x_{k-1},y_{k-1})_{i+k-2} I(x_k,y_k)_{i+k-1}$ as an example, here we suppose $k\geq 4$. Considering commutators that contribute to $\mathbf{B}_i^{k} = I(x_1,y_1)_i \big( \prod^{k-2}_{j=2} A^j(x_j,y_j)_{i+j-1} \big) \hat{c}_\uparrow^\dagger(x_{k-1},y_{k-1})_{i+k-2} \hat{c}_\uparrow(x_k,y_k)_{i+k-1}$,
\begin{align}
&[I(x_1,y_1)_i \big( \prod^{k-2}_{j=2} A^j(x_j,y_j)_{i+j-1} \big) \hat{n}_\uparrow(x_{k-1},y_{k-1})_{i+k-2} I(x_k,y_k)_{i+k-1},\hat{c}^\dagger_{i+k-2,\uparrow} \hat{c}_{i+k-1,\uparrow}] \nonumber \\
=&I(x_1,y_1)_i \big( \prod^{k-2}_{j=2} A^j(x_j,y_j)_{i+j-1} \big) \hat{c}^\dagger_\uparrow(x_{k-1},y_{k-1})_{i+k-2} \hat{c}_\uparrow(x_k,y_k)_{i+k-1}
\end{align}

is the only term related to the possible nonzero input. Therefore, $I(x_1,y_1)_i \big( \prod^{k-2}_{j=2} A^j(x_j,y_j)_{i+j-1} \big) \hat{n}_\uparrow(x_{k-1},y_{k-1})_{i+k-2} I(x_k,y_k)_{i+k-1}$ has zero coefficient.
\end{widetext}
Thus, exactly as in step~2 of the proof of Theorem~\ref{nontrivial}, coefficients vanish unless
$(A_1=)A_2=\cdots=A_{k-1}(=A_k)=I$. The remaining case is also shown to have
zero coefficients, so every type~\ref{enum:input_3} candidate has zero
coefficient. This completes the proof.
\end{proof}

\subsection{Proof of Lemma~\ref{t:HHtrivial}}
\begin{proof}
We prove the statement separately for $k=2$ and $k=1$.
\subsubsection{Proof for $k=2$ case}
We show that the only 2-local conserved quantity is the Hamiltonian $\hat{H}$ itself up to multiplication by a constant and the freedom of adding 1-local conserved quantities.

Consider $\mathbf{A}_i^2=e_i^1 e_{i+1}^2$ of type~\ref{enum:input_1} and
type~\ref{enum:input_2}. Without loss of generality, we assume $e^2 \neq I(x,y)$
We first explain the right-extension argument in a
representative case. Take
\begin{align}
\mathbf{A}^2_i=e^1_i\hat{c}_{\uparrow}f_{\downarrow}(x,y)_{i+1},
\end{align}
where $f_{\downarrow}\in\{I,\hat{c}_{\downarrow},
\hat{c}^\dagger_{\downarrow},\hat{n}_{\downarrow}\}$. Commuting this input with
the hopping term of the up-spin fermion on the right gives
\begin{align}
&[e^1_i\hat{c}_{\uparrow}f_{\downarrow}(x,y)_{i+1},
\hat{c}^\dagger_{i+1,\uparrow}\hat{c}_{i+2,\uparrow}]\notag\\
&\quad =(-1)^{p_{\downarrow}}e^1_i I_{\uparrow}f_{\downarrow}(x,y)_{i+1}
\hat{c}_{i+2,\uparrow}.
\end{align}
Here $p_{\downarrow}=0$ for $f_{\downarrow}=I,\hat{n}_{\downarrow}$ and
$p_{\downarrow}=1$ for $f_{\downarrow}=\hat{c}_{\downarrow},
\hat{c}^{\dagger}_{\downarrow}$.
If $e^1$ is neither $\hat{c}_{\sigma}(0,0)$ nor
$\hat{c}^\dagger_{\sigma}(0,0)$ for $\sigma=\uparrow,\downarrow$, then no
other $2$-support input can
produce this $3$-support output. Hence the coefficient of this input vanishes.

The other possible right-end factors are treated in the same manner. For
example, when the right endpoint contains
$\hat{c}^\dagger_{\uparrow}f_{\downarrow}(x,y)$ or
$\hat{n}_{\uparrow}f_{\downarrow}(x,y)$, one considers the corresponding outputs
obtained by the right hopping terms
$\hat{c}^\dagger_{i+2,\uparrow}\hat{c}_{i+1,\uparrow}$ or
$\hat{c}^\dagger_{i+1,\uparrow}\hat{c}_{i+2,\uparrow}$, respectively. The
same uniqueness argument shows that the coefficient is zero unless
$e^1=\hat{c}_{\sigma}(0,0)$ or $\hat{c}^\dagger_{\sigma}(0,0)$ for some
$\sigma$. Repeating this argument for down-spin hopping gives the identical
conclusion for right endpoints involving the down-spin fermionic factor. Thus a
nonzero coefficient is possible only when the left endpoint is
$\hat{c}_{\sigma}(0,0)$ or $\hat{c}^\dagger_{\sigma}(0,0)$.

Applying the analogous left-extension argument shows the same restriction on
the right endpoint. Therefore, coefficients of type~\ref{enum:input_1} and
type~\ref{enum:input_2} inputs $\mathbf{A}^2$ vanish unless both endpoints are
fermionic single operators $\hat{c}_{\sigma}(0,0)$ or
$\hat{c}^\dagger_{\sigma}(0,0)$. In particular, all type~\ref{enum:input_2}
candidates have zero coefficients.

We next compare coefficients of the remaining candidates under one-site
translations. As a representative example, consider
$\mathbf{A}^2_i=\hat{c}^\dagger_{i,\sigma}\hat{c}_{i+1,\sigma'}$ and the output
$\mathbf{B}^3_i=\hat{c}^\dagger_{i,\sigma}\hat{c}_{i+2,\sigma'}$. This output
is produced by the two translated inputs
\begin{align}
[\hat{c}^\dagger_{i,\sigma}\hat{c}_{i+1,\sigma'},
\hat{c}^\dagger_{i+1,\sigma'}\hat{c}_{i+2,\sigma'}]
&=\hat{c}^\dagger_{i,\sigma}\hat{c}_{i+2,\sigma'},
\\
[\hat{c}^\dagger_{i+1,\sigma}\hat{c}_{i+2,\sigma'},
\hat{c}^\dagger_{i,\sigma}\hat{c}_{i+1,\sigma}]
&=-\hat{c}^\dagger_{i,\sigma}\hat{c}_{i+2,\sigma'}.
\end{align}
Hence coefficient comparison gives
$q_{\hat{c}^\dagger_{i,\sigma}\hat{c}_{i+1,\sigma'}}
=q_{\hat{c}^\dagger_{i+1,\sigma}\hat{c}_{i+2,\sigma'}}$. Performing the same
comparison for the other remaining candidates yields
\begin{align}
q_{\hat{c}^\dagger_{i,\sigma}\hat{c}_{i+1,\sigma'}}
&=q_{\hat{c}^\dagger_{i+1,\sigma}\hat{c}_{i+2,\sigma'}}=\dots,
\\
q_{\hat{c}^\dagger_{i+1,\sigma}\hat{c}_{i,\sigma'}}
&=q_{\hat{c}^\dagger_{i+2,\sigma}\hat{c}_{i+1,\sigma'}}=\dots,
\\
q_{\hat{c}_{i,\sigma}\hat{c}_{i+1,\sigma'}}
&=-q_{\hat{c}_{i+1,\sigma}\hat{c}_{i+2,\sigma'}}=\dots,
\\
q_{\hat{c}^\dagger_{i,\sigma}\hat{c}^\dagger_{i+1,\sigma'}}
&=-q_{\hat{c}^\dagger_{i+1,\sigma}\hat{c}^\dagger_{i+2,\sigma'}}=\dots.
\end{align}

We first exclude the phonon-dressed onsite pair inputs that may
compete with the two-support pair inputs.
For $x+y\geq2$, the hopping commutator is
\begin{align}
&[\hat{c}_{\uparrow}\hat{c}_{\downarrow}(x,y)_i,
\hat{c}^\dagger_{i,\uparrow}\hat{c}_{i+1,\uparrow}]
\nonumber\\
&\qquad=-\hat{c}_{\downarrow}(x,y)_i
\hat{c}_{i+1,\uparrow}.
\end{align}
Its right-hand side is a two-support output.
The remaining two-support fermionic inputs are bare endpoint pairs;
their electron--phonon commutators have at most one bosonic factor,
and the onsite Coulomb term introduces no bosonic factor.
Purely bosonic inputs cannot generate a pair-annihilation output.
No other one-support input generates this output with its bosonic
factor at site $i$. Since $t\neq0$, all coefficients of
$\hat{c}_{\uparrow}\hat{c}_{\downarrow}(x,y)_i$ with $x+y\geq2$ vanish.
The same hopping comparison excludes the pair-creation counterparts.

For the onsite pair input
$\hat{c}_{\uparrow}\hat{c}_{\downarrow}(0,1)_i$, we use
\begin{align}
&[\hat{c}_{\uparrow}\hat{c}_{\downarrow}(0,1)_i,
\hat{n}_{\uparrow}(0,1)_i+\hat{n}_{\downarrow}(0,1)_i]
\nonumber\\
&\qquad=2\hat{c}_{\uparrow}\hat{c}_{\downarrow}(0,2)_i.
\end{align}
Any other onsite input producing this one-support output has bosonic degree
at least two and has just been excluded.
The bare two-support pairs cannot produce an onsite output with
two bosonic factors. Therefore coefficient comparison gives
$2gq_{\hat{c}_{\uparrow}\hat{c}_{\downarrow}(0,1)_i}=0$.
Since $g\neq0$, this coefficient vanishes.
The same argument with the output
$\hat{c}_{\uparrow}\hat{c}_{\downarrow}(2,0)_i$
excludes the $(1,0)$ input.
For pair creation the corresponding commutators have the factor $-2$
instead of $2$, so their degree-one coefficients also vanish.
Finally, the one-support output
$\hat{c}_{\uparrow}\hat{c}_{\downarrow}(0,1)_i$ has only the
contribution $2gq_{\hat{c}_{i,\uparrow}\hat{c}_{i,\downarrow}}$
after the higher-degree inputs have been excluded.
Thus $q_{\hat{c}_{i,\uparrow}\hat{c}_{i,\downarrow}}=0$, and the bare pair-creation
input is excluded in the same manner.

We next eliminate particle-nonconserving remaining candidates. Consider first
$\mathbf{A}^2_i=\hat{c}_{i,\sigma}\hat{c}_{i+1,\sigma'}$ and the output
$\mathbf{B}^2_i=\hat{c}_{i,\sigma}\hat{c}_{\sigma'}(0,1)_{i+1}$. This output is
also generated, when $\sigma\neq\sigma'$, by the onsite pair input:
\begin{align}
&[\hat{c}_{i+1,\sigma}\hat{c}_{\sigma'}(0,1)_{i+1},
\hat{c}^\dagger_{i+1,\sigma}\hat{c}_{i,\sigma}]
\nonumber\\
&\qquad=\hat{c}_{i,\sigma}\hat{c}_{\sigma'}(0,1)_{i+1}.
\end{align}
Its coefficient has just been shown to vanish; for $\sigma=\sigma'$
the onsite pair itself is zero. Among the remaining inputs,
the only contribution is
\begin{align}
[\hat{c}_{i,\sigma}\hat{c}_{i+1,\sigma'},
\hat{n}_{\sigma'}(0,1)_{i+1}]
&=\hat{c}_{i,\sigma}\hat{c}_{\sigma'}(0,1)_{i+1}.
\end{align}
Thus $q_{\hat{c}_{i,\sigma}\hat{c}_{i+1,\sigma'}}=0$. The pair-creation
candidate is eliminated in the same way: using the output
$\hat{c}^\dagger_{i,\sigma}\hat{c}^\dagger_{\sigma'}(0,1)_{i+1}$, one obtains
$q_{\hat{c}^\dagger_{i,\sigma}\hat{c}^\dagger_{i+1,\sigma'}}=0$.

It remains to relate the hopping-type candidates. We take
$\mathbf{A}^2_i=\hat{c}^\dagger_{i,\sigma}\hat{c}_{i+1,\sigma'}$ as a
representative example and consider the output
$\mathbf{B}^2_i=\hat{c}^\dagger_{i,\sigma}\hat{c}_{\sigma'}(0,1)_{i+1}$. One
commutator generating this output is
\begin{align}
[\hat{c}^\dagger_{i,\sigma}\hat{c}_{i+1,\sigma'},
\hat{n}_{\sigma'}(0,1)_{i+1}]
&=\hat{c}^\dagger_{i,\sigma}\hat{c}_{\sigma'}(0,1)_{i+1}.
\end{align}
A competing 1-support input exists for every pair $\sigma,\sigma'$:
\begin{align}
[\hat{c}^\dagger_{i+1,\sigma}\hat{c}_{\sigma'}(0,1)_{i+1},
\hat{c}^\dagger_{i,\sigma}\hat{c}_{i+1,\sigma}]
&=-\hat{c}^\dagger_{i,\sigma}\hat{c}_{\sigma'}(0,1)_{i+1}.
\end{align}
For $\sigma\neq\sigma'$, the coefficient of this additional input vanishes.
To see this, first note that the hopping commutators of
$\hat{c}^\dagger_{i,\sigma}\hat{c}_{\sigma'}(x,y)_i$ with $x+y\geq2$
produce two-support outputs carrying at least two bosonic factors.  No
two-support candidate surviving the preceding analysis can cancel such an
output, and hence all these coefficients are zero.  The one-support output in
\begin{align*}
[\hat{c}^\dagger_{i,\sigma}\hat{c}_{\sigma'}(0,1)_i,I(1,1)_i]
=\hat{c}^\dagger_{i,\sigma}\hat{c}_{\sigma'}(0,1)_i
\end{align*}
then has no competing contribution: the only possible additional contributions
from the electron--phonon term would come from the already eliminated
higher-degree inputs.  Since $\omega\neq0$, it follows that
$q_{\hat{c}^\dagger_{i,\sigma}\hat{c}_{\sigma'}(0,1)_i}=0$.  The same argument
applies to the $(1,0)$ counterpart.  Equations above therefore imply that the
hopping coefficient vanishes for $\sigma\neq\sigma'$.  When
$\sigma=\sigma'$, the additional input reduces to
$\hat{n}_{\sigma}(0,1)_{i+1}$, and its coefficient is related to the hopping
coefficient. The same
argument with $\hat{n}_{\sigma'}(1,0)_{i+1}$ gives the corresponding relation
with $q_{\hat{n}_{\sigma}(1,0)_{i+1}}$. Applying the analogous analysis to
$\hat{c}^\dagger_{i+1,\sigma'}\hat{c}_{i,\sigma}$ gives
\begin{align}
q_{\hat{c}^\dagger_{i,\sigma}\hat{c}_{i+1,\sigma'}}&=q_{\hat{c}^\dagger_{i+1,\sigma'}\hat{c}_{i,\sigma}}=0
\quad (\sigma\neq \sigma')
\\
gq_{\hat{c}^\dagger_{i,\sigma} 
\hat{c}_{i+1,\sigma}}
&=tq_{\hat{n}_{\sigma}(0,1)_{i+1}}
=tq_{\hat{n}_{\sigma}(1,0)_{i+1}}
=gq_{\hat{c}^\dagger_{i+1,\sigma} 
\hat{c}_{i,\sigma}}.
\end{align}

Before the one-support comparison, the higher-degree onsite number
inputs must also be excluded. As in
Eq.~\eqref{eq:short-onsite-high-degree}, for $x+y\geq2$ we use
\begin{align}
&[\hat{n}_{\sigma}(x,y)_i,
\hat{c}^\dagger_{i,\sigma}\hat{c}_{i+1,\sigma}]
\nonumber\\
&\qquad=\hat{c}^\dagger_{\sigma}(x,y)_i
\hat{c}_{i+1,\sigma}.
\end{align}
The output $\mathbf{B}^2_i$ on the right has at least two bosonic
factors. No remaining two-support fermionic input can produce it,
and an input at site $i+1$ cannot produce its bosonic factor at
site $i$. Onsite inputs containing the other spin's number operator
give a different fermionic factor in the output.
Thus $q_{\hat{n}_{\sigma}(x,y)_i}=0$ for $x+y\geq2$.
The purely bosonic outputs $I(0,2)_i$ and $I(2,0)_i$ have only the
phonon-number contributions
$2\omega q_{I(0,2)_i}$ and $-2\omega q_{I(2,0)_i}$, respectively,
by the same argument as in the proof of Theorem~\ref{trivial}.
Hence $q_{I(0,2)_i}=q_{I(2,0)_i}=0$.
These restrictions exclude the additional electron--phonon
contributions from $\hat{n}_{\sigma}(0,2)_i$,
$\hat{n}_{\sigma}(1,1)_i$, and $I(0,2)_i$ to
$\hat{n}_{\sigma}(0,1)_i$, and likewise for the $(1,0)$ output.

We next relate the coefficients of $\hat{n}_{\sigma}(0,1)_i$ and
$\hat{n}_{\sigma}(1,0)_i$ to that of $I(1,1)_i$. For
$\hat{n}_{\sigma}(0,1)_i$, the relevant output is generated by
\begin{align}
[\hat{n}_{\sigma}(0,1)_{i},I(1,1)_i]&=\hat{n}_{\sigma}(0,1)_{i},
\\
[I(1,1)_i,\hat{n}_{\sigma}(0,1)_{i}]&=-\hat{n}_{\sigma}(0,1)_{i}.
\end{align}
Thus coefficient comparison gives
$\omega q_{\hat{n}_{\sigma}(0,1)_{i}}=gq_{I(1,1)_i}$. The case
$\hat{n}_{\sigma}(1,0)_i$ is analogous and gives
$gq_{I(1,1)_i}=\omega q_{\hat{n}_{\sigma}(1,0)_{i}}$. Hence,
\begin{align}
\omega q_{\hat{n}_{\sigma}(0,1)_{i}}
&=g q_{I(1,1)_i}
=\omega q_{\hat{n}_{\sigma}(1,0)_{i}}.
\end{align}
Likewise, for
$\hat{c}^\dagger_{i,\sigma}\hat{c}_{i+1,\sigma},\hat{c}^\dagger_{i+1,\sigma}\hat{c}_{i,\sigma}$,
commutators with the Hamiltonian produce
\begin{align}
[\hat{c}^\dagger_{i,\sigma}\hat{c}_{i+1,\sigma},\hat{n}_{i+1,\sigma}\hat{n}_{i+1,\sigma'}]&=\hat{c}^\dagger_{i,\sigma}\hat{c}_{i+1,\sigma}\hat{n}_{i+1,\sigma'}
\end{align}
as outputs. Terms with the same outputs include
\begin{align}
[\hat{n}_{i+1,\sigma}\hat{n}_{i+1,\sigma'},\hat{c}^\dagger_{i,\sigma}\hat{c}_{i+1,\sigma}]&=-\hat{c}^\dagger_{i,\sigma}\hat{c}_{i+1,\sigma}\hat{n}_{i+1,\sigma'}
\end{align}
and combining them gives
\begin{align}
U q_{\hat{c}^\dagger_{i,\sigma}\hat{c}_{i+1,\sigma}}
&=t q_{\hat{n}_{i+1,\sigma}\hat{n}_{i+1,\sigma'}}.
\end{align}

In addition, consider type~\ref{enum:input_3} operators $\mathbf{A}^2$.
Here, all candidates are
$\mathbf{A}^2_i=I(x_1,y_1)_iI(x_2,y_2)_{i+1}$ with
$(x_1,y_1),(x_2,y_2) \neq(0,0)$. For example, when $x_1\neq 0$, consider inputs
that output
$\mathbf{B}^2_i=\hat{n}_{\sigma}(x_1-1,y_1)_iI(x_2,y_2)_{i+1}$:
\begin{align}
&[I(x_1,y_1)_iI(x_2,y_2)_{i+1},\hat{n}_{\sigma}(0,1)_i]\nonumber \\
=&-x_1\hat{n}_{\sigma}(x_1-1,y_1)_iI(x_2,y_2)_{i+1}
\\
&[I(x_1-1,y_1+1)_iI(x_2,y_2)_{i+1},\hat{n}_{\sigma}(1,0)_i] \nonumber
\\
=&(y_1+1)\hat{n}_{\sigma}(x_1-1,y_1)_iI(x_2,y_2)_{i+1}
\end{align}
these inputs exist, and
\begin{align}
q_{I(x_1,y_1)_iI(x_2,y_2)_{i+1}}\propto q_{I(x_1-1,y_1+1)_iI(x_2,y_2)_{i+1}},
\end{align}
which implies that for any
$\mathbf{A}^2_i=I(x_1,y_1)_iI(x_2,y_2)_{i+1}$ there is a proportional
operator
$\mathbf{A'}^2_i=I(0,x_1+y_1)_iI(0,x_2+y_2)_{i+1}$.
Now consider inputs that output
$\mathbf{B}^2=\mathbf{A'}^2_i=I(0,x_1+y_1)_iI(0,x_2+y_2)_{i+1}$:
\begin{align}
[\hat{b}^{x_1+y_1}_i\hat{b}^{x_2+y_2}_{i+1},\hat{b}^\dagger_i\hat{b}_{i}]&=(x_1+y_1)\hat{b}^{x_1+y_1}_i\hat{b}^{x_2+y_2}_{i+1}
\\
[\hat{b}^{x_1+y_1}_i\hat{b}^{x_2+y_2}_{i+1},\hat{b}^\dagger_{i+1}\hat{b}_{i+1}]&=(x_2+y_2)\hat{b}^{x_1+y_1}_i\hat{b}^{x_2+y_2}_{i+1}
\end{align}
these contribute, and they are all relevant terms. Therefore,
\begin{align}
\omega(x_1+y_1+x_2+y_2)q_{I(0,x_1+y_1)_iI(0,x_2+y_2)_{i+1}}&=0
\end{align}
which implies $q_{I(0,x_1+y_1)_iI(0,x_2+y_2)_{i+1}}=0$.
Together with the proportionality relation above, this gives
$q_{I(x_1,y_1)_iI(x_2,y_2)_{i+1}}=0$.
Hence, all type~\ref{enum:input_3} candidates have zero coefficients.

\begin{widetext}
Therefore, for coefficients of operators forming $Q^2$,
\begin{align}
&q_{\hat{c}^\dagger_{i,\sigma}\hat{c}_{i+1,\sigma}}
:q_{\hat{c}^\dagger_{i+1,\sigma}\hat{c}_{i,\sigma}}
:q_{\hat{n}_{\sigma}(0,1)_i}:q_{\hat{n}_{\sigma}(1,0)_i}
:q_{I(1,1)_i}:q_{\hat{n}_{i,\sigma}\hat{n}_{i,\sigma'}}=t:t:g:g:\omega:U
\end{align}
and
\begin{align}
q_{A^2}&=0\quad (\textnormal{otherwise})
\end{align}
holds. In other words, up to the freedom of adding 1-local conserved quantities, every 2-local conserved quantity is a constant multiple of the Hamiltonian.
\end{widetext}

\subsubsection{Proof for $k=1$ case}
We next consider the case $k=1$. We first examine candidates containing a single
spin component, namely
$\hat{c}_{\sigma}(x,y)$, $\hat{c}^\dagger_{\sigma}(x,y)$, and
$\hat{n}_{\sigma}(x,y)$. As a representative case, take
$\mathbf{A}^1_i=\hat{n}_{\sigma}(x,y)_i$. The commutator with the hopping term
gives
\begin{align}
[\hat{n}_{\sigma}(x,y)_i,
\hat{c}^\dagger_{i,\sigma}\hat{c}_{i+1,\sigma}]
&=\hat{c}^\dagger_{\sigma}(x,y)_i\hat{c}_{i+1,\sigma}.
\end{align}
If $(x,y)\neq(0,0)$, this output has no competing input among the remaining
candidates, so $q_{\hat{n}_{\sigma}(x,y)_i}=0$. When $(x,y)=(0,0)$, however,
the neighboring number operator gives the same output:
\begin{align}
[\hat{n}_{\sigma}(0,0)_{i+1},
\hat{c}^\dagger_{i,\sigma}\hat{c}_{i+1,\sigma}]
&=-\hat{c}^\dagger_{i,\sigma}\hat{c}_{i+1,\sigma}.
\end{align}
Thus
\begin{align}
q_{\hat{n}_{i,\sigma}}-q_{\hat{n}_{i+1,\sigma}}&=0.
\end{align}
Similarly, the commutators
\begin{align}
[\hat{c}_{\sigma}(x,y)_i,
\hat{c}^\dagger_{i,\sigma}\hat{c}_{i+1,\sigma}]
&=I(x,y)_i\hat{c}_{i+1,\sigma},
\\
[\hat{c}^\dagger_{\sigma}(x,y)_i,
\hat{c}^\dagger_{i+1,\sigma}\hat{c}_{i,\sigma}]
&=-I(x,y)_i\hat{c}^\dagger_{i+1,\sigma}
\end{align}
show that the coefficients of $\hat{c}_{\sigma}(x,y)_i$ and
$\hat{c}^\dagger_{\sigma}(x,y)_i$ vanish for $(x,y)\neq(0,0)$; the case
$(x,y)=(0,0)$ remains for now.

We next consider mixed internal-degree candidates containing a number operator,
such as $\hat{c}_{\sigma}\hat{n}_{\sigma'}(x,y)$ and
$\hat{n}_{\sigma}\hat{n}_{\sigma'}(x,y)$ with $\sigma\neq\sigma'$. For example,
\begin{align}
[\hat{c}_{\sigma}\hat{n}_{\sigma'}(x,y)_i,
\hat{c}^\dagger_{i,\sigma}\hat{c}_{i+1,\sigma}]
&=\hat{n}_{\sigma'}(x,y)_i\hat{c}_{i+1,\sigma}.
\end{align}
This output is generated only by this input, even when $(x,y)=(0,0)$. Hence its
coefficient is zero. The same argument applied to
$\hat{n}_{\sigma}\hat{n}_{\sigma'}(x,y)$, using
\begin{align}
[\hat{n}_{\sigma}\hat{n}_{\sigma'}(x,y)_i,
\hat{c}^\dagger_{i,\sigma}\hat{c}_{i+1,\sigma}]
&=\hat{c}^\dagger_{\sigma}\hat{n}_{\sigma'}(x,y)_i\hat{c}_{i+1,\sigma},
\end{align}
shows that these coefficients also vanish.

It remains to eliminate candidates that do not commute with the total
particle-number operator. Take $\hat{c}_{i,\sigma}$ as a representative example:
\begin{align}
[\hat{c}_{i,\sigma},\hat{n}_{\sigma}(0,1)_i]
&=\hat{c}_{\sigma}(0,1)_i.
\end{align}
This output has no competing input, so $q_{\hat{c}_{i,\sigma}}=0$. The creation operator gives the unique output
\begin{align}
[\hat{c}^\dagger_{i,\sigma},\hat{n}_{\sigma}(0,1)_i]
&=-\hat{c}^\dagger_{\sigma}(0,1)_i,
\end{align}
so $q_{\hat{c}^\dagger_{i,\sigma}}=0$. Likewise, the pair-annihilation
candidate gives
\begin{align}
[\hat{c}_{\sigma}\hat{c}_{\sigma'}(0,0)_i,
\hat{n}_{\sigma}(0,1)_i]
&=\hat{c}_{\sigma}\hat{c}_{\sigma'}(0,1)_i,
\end{align}
and this output is also unique; hence
$q_{\hat{c}_{\sigma}\hat{c}_{\sigma'}(0,0)_i}=0$. The pair-creation candidate $\hat{c}^\dagger_{\sigma}\hat{c}^\dagger_{\sigma'}(0,0)_i$ is
eliminated analogously.
Thus all remaining candidates that change the total particle number have zero
coefficients.

We finally consider purely bosonic candidates $I(x,y)$. Suppose first that
$x\neq0$. The relevant coefficient comparison is obtained from
\begin{align}
[I(x,y)_i,\hat{n}_{\sigma}(0,1)_i]
&=-x\hat{n}_{\sigma}(x-1,y)_i,
\\
[I(x-1,y+1)_i,\hat{n}_{\sigma}(1,0)_i]
&=(y+1)\hat{n}_{\sigma}(x-1,y)_i,
\end{align}
which gives $q_{I(x,y)_i}\propto q_{I(x-1,y+1)_i}$. Iterating this relation,
$q_{I(x,y)_i}$ is proportional to $q_{I(0,x+y)_i}$. The latter coefficient
vanishes because
\begin{align}
[I(0,x+y)_i,I(1,1)_i]&=(x+y)I(0,x+y)_i
\end{align}
is the unique contribution to this output. Hence
$q_{I(0,x+y)_i}=0$ for all $(x,y)\neq(0,0)$. Together with the proportionality
relation above, this gives $q_{I(x,y)_i}=0$.

The only remaining fermionic candidates are the bilinears conserving the total
fermion number. The diagonal ones are $\hat{n}_{i,\uparrow}$ and
$\hat{n}_{i,\downarrow}$, while the off-diagonal basis element
$\hat{c}_{i,\uparrow}\hat{c}^\dagger_{i,\downarrow}$ is rewritten as
$-\hat{c}^\dagger_{i,\downarrow}\hat{c}_{i,\uparrow}$ by the fermionic
anticommutation relation. Thus the remaining candidates can be written uniformly
as
\begin{align}
\hat{N}_{(\sigma,\sigma')}
=\sum_i\hat{c}^\dagger_{i,\sigma}\hat{c}_{i,\sigma'}
\quad (\sigma,\sigma'\in\{\uparrow,\downarrow\}).
\end{align}
Each $\hat{N}_{(\sigma,\sigma')}$ commutes with the Hamiltonian, since the
hopping, electron--phonon coupling, phonon, and onsite Coulomb terms are
invariant under global spin rotations.
Therefore, among 1-local conserved quantities, only linear combinations of
$\hat{N}_{(\sigma,\sigma')}$ remain.
\end{proof}

\subsection{Extension to an arbitrary number of fermionic components}
\label{sec:SUN_extension}

Although the model in the main text has the two spin components
$\sigma\in\{\uparrow,\downarrow\}$, the proof above does not use the fact that
the internal index takes only two values.  We may instead introduce a finite
number $N$ of fermionic components, labeled by $a=1,2,\ldots,N$, and replace
the spin sums in the hopping and electron--phonon terms by sums over $a$.  With
the normalization convention of Eq.~\eqref{Eq.1D_HH_model}, the corresponding
onsite interaction is
\begin{align}
\hat H_{\mathrm{e-e}}^{(N)}
&=\frac{U}{2}\sum_{i=1}^L\sum_{a\ne b}\hat n_{i,a}\hat n_{i,b}
=\frac{U}{2}\sum_{i=1}^L\hat n_i(\hat n_i-1),
\nonumber\\[-2pt]
&\hspace{32mm}\hat n_i=\sum_{a=1}^N\hat n_{i,a}.
\label{Eq.SUN_interaction}
\end{align}
This gives the $\mathrm{SU}(N)$ Holstein--Hubbard model; its $U=0$ limit is the
$\mathrm{SU}(N)$ Holstein model.  Representative studies of these
multicomponent Hubbard and Holstein systems can be found in
Refs.~\cite{SUNHubbardHonerkamp,SUNHubbardReview,SUNHubbardNakagawa,
SUNHubbardHashimoto,SUNHolstein}.

For the operator expansion, the onsite fermionic factor in
Eq.~\eqref{spinfulbasis_2} is replaced by
\begin{align}
\prod_{a=1}^N f_{i,a},
\qquad
f_{i,a}\in\{I,\hat c_{i,a},\hat c^\dagger_{i,a},\hat n_{i,a}\},
\label{Eq.SUN_onsite_basis}
\end{align}
where the product is taken in the fixed order $a=1,2,\ldots,N$.  In step~1,
the same boundary-output comparison leaves only the two-endpoint fermionic
operators displayed above, now with arbitrary endpoint labels $a$ and $b$.
The interaction~\eqref{Eq.SUN_interaction} is onsite and hence does not
contribute to the $(k+1)$-support outputs.  In steps~2 and 3, the labels $a$
and $b$ are merely carried through the hopping relations; no comparison
requires that the set of available labels contain only two elements.  Factors
of all inactive components remain spectators in the fixed onsite product.
Consequently, the coefficient-elimination argument for $k\geq3$ is unchanged.

The short-support analysis is unchanged for the same reason.  The only
additional notation is that the $1$-local global charges are
\begin{align}
\hat N_{(a,b)}=\sum_{i=1}^L\hat c^\dagger_{i,a}\hat c_{i,b},
\qquad a,b\in\{1,2,\ldots,N\}.
\label{Eq.SUN_charges}
\end{align}
Thus, for every fixed finite $N\geq2$, the conclusion of
Theorem~\ref{t:HHwithtrivial} extends to the $\mathrm{SU}(N)$
Holstein--Hubbard model: every $k$-local conserved quantity with $k<L/2$
is a linear combination of the Hamiltonian and the charges
\eqref{Eq.SUN_charges}.  The same flavor-label argument applies to
Appendix~\ref{sec:higher_dimensions_proof}, so the higher-dimensional absence
result likewise admits this $\mathrm{SU}(N)$ extension.

\section{Detailed argument for higher-dimensional models (Proof of Theorem~\ref{t:higher_dimensions})}
\label{sec:higher_dimensions_proof}

In this appendix, we prove Theorem~\ref{t:higher_dimensions}.  We first introduce the
operator basis and ordering used on the two-dimensional square lattice.  We then prove the
result for the two-dimensional Holstein model, treat the two-dimensional
Holstein--Hubbard model separately, and finally extend both arguments to $d\geq3$.
The only new geometric step is to show that the two fermionic endpoints left by the
maximal-width analysis must lie on the same row.

\subsection{Two-dimensional setup: Operator basis, ordering, and width}

Write a site of the square lattice as $\overrightarrow{r}=(i,j)$.  For the spinless model, the
single-site basis is the direct two-dimensional analogue of that in
Sec.~\ref{sec:proof_strategy}:
\begin{align}
f(x,y)_{\overrightarrow{r}}
={}&f_{\overrightarrow{r}}(\hat b^\dagger_{\overrightarrow{r}})^x\hat b_{\overrightarrow{r}}^y,
\nonumber\\
&f_{\overrightarrow{r}}\in
\{I,\hat c_{\overrightarrow{r}},\hat c^\dagger_{\overrightarrow{r}},\hat n_{\overrightarrow{r}}\},
\qquad x,y\in\mathbb Z_{\geq0}.
\label{Eq.2D_onsite_basis}
\end{align}
Thus, for example,
$\hat c(0,1)_{\overrightarrow{r}}=\hat c_{\overrightarrow{r}}\hat b_{\overrightarrow{r}}$ and
$I(0,1)_{\overrightarrow{r}}=\hat b_{\overrightarrow{r}}$.  To fix all fermionic signs, we order the sites as
\begin{align}
(1,1),(1,2),\ldots,(1,L),(2,1),(2,2),\ldots,(2,L),\ldots,(L,L).
\label{Eq.2D_order_sequence}
\end{align}
Equivalently, the site $(i,j)$ is placed before $(i',j')$ if
\begin{align}
i<i'\ \text{or}\ \bigl(i=i'\ \text{and}\ j<j'\bigr).
\label{Eq.2D_order}
\end{align}
An operator-basis product is obtained by multiplying its nonidentity onsite factors from
left to right in this order:
\begin{align}
\mathbf A=\prod_{\overrightarrow{r}\in S}e_{\overrightarrow{r}},
\qquad e_{\overrightarrow{r}}=f(x_{\overrightarrow{r}},y_{\overrightarrow{r}})_{\overrightarrow{r}},
\label{Eq.2D_product}
\end{align}
where every displayed factor is nonidentity. Output operators $\mathbf{B}$ are expanded in
the same ordered basis. When two odd fermionic factors are exchanged to restore the order in
Eq.~\eqref{Eq.2D_order}, the resulting sign is included in the coefficient of the ordered
product.

For the spinful model, we use the onsite basis
\begin{align}
e_{\overrightarrow{r}}={}&f_{\overrightarrow{r},\uparrow}f_{\overrightarrow{r},\downarrow}
(\hat b^\dagger_{\overrightarrow{r}})^x\hat b_{\overrightarrow{r}}^y,
\nonumber\\
&f_{\overrightarrow{r},\sigma}\in
\{I,\hat c_{\overrightarrow{r},\sigma},\hat c^\dagger_{\overrightarrow{r},\sigma},
\hat n_{\overrightarrow{r},\sigma}\},
\label{Eq.2D_spinful_basis}
\end{align}
with the up-spin factor placed before the down-spin factor at the same site.  The sites are
again ordered by Eq.~\eqref{Eq.2D_order}.  Since every term of the Hamiltonian has even
fermion parity, a factor whose support is disjoint from a Hamiltonian term can be carried
through the commutator as in Eq.~\eqref{fermioncommutation}.  All calculations below are
therefore local calculations on the affected sites, followed by restoration of the fixed
ordering.

For the spinless model, using the width in Eq.~\eqref{Eq.dD_width}, we expand a
candidate conserved quantity and its
commutator as
\begin{align}
\hat Q^k&=\sum_{l=1}^k\sum_{\mathbf A^l}
q_{\mathbf A^l}\mathbf A^l,
\qquad \operatorname{Wid}\mathbf A^l=l,
\label{Eq.2D_input_expansion}\\
[\hat Q^k,\hat H]&=\sum_{l=0}^{k+1}\sum_{\mathbf B^l}
r_{\mathbf B^l}\mathbf B^l,
\qquad \operatorname{Wid}\mathbf B^l=l.
\label{Eq.2D_output_expansion}
\end{align}
As in Sec.~\ref{sec:proof_strategy}, we refer to $\mathbf{A}$ as an \textit{input}
operator and to $\mathbf{B}$ as an \textit{output} operator. The conservation condition is
$r_{\mathbf B}=0$ for every output. In particular, if an output
is generated by a unique input of width at most $k$, the coefficient of that input is zero.

Take a $k$-support input $\mathbf A^k$ and choose a coordinate direction in which its width is $k$.
After interchanging the axes if necessary, take it to be the first direction and write its
minimal horizontal interval as $i,i+1,\ldots,i+k-1$. The columns $i$ and $i+k-1$ are
the two ends of its support in this direction. As in Sec.~\ref{sec:proof}, we classify the
input as type~\ref{enum:input_1}, \ref{enum:input_2}, or \ref{enum:input_3} according to
whether it includes $\hat c^\dagger$, $\hat c$, or $\hat n$ on both ends, on one end but
not on the other, or on neither end, respectively.

\subsection{Two-dimensional Holstein model}

We first prove the absence result for the two-dimensional Holstein model.

\subsubsection{Proof step 1: Basic relations for products with width $k+1$.}

We first use the conditions $r_{\mathbf B^{k+1}}=0$ for all $(k+1)$-support outputs
$\mathbf B^{k+1}$. The boundary-shift argument is
the fermion--boson version of that used for the higher-dimensional Hubbard model in
Ref.~\cite{2DHubbard}.

\begin{lemma}[Step 1 analysis for type~\ref{enum:input_1}]\label{lem:2D_step1}
Assume $\hat Q^k$ is a $k$-local conserved quantity of the two-dimensional Holstein
model~\eqref{Eq.dD_H_model} with $t\neq0$, $g\neq0$, and $\omega\neq0$.
The coefficient of a $k$-support input $\mathbf A^k$ of type~\ref{enum:input_1}
whose width is $k$ in the horizontal direction is zero unless $\mathbf A^k$ is
\begin{align}
&\hat c^\dagger_{(i,j)}\hat c_{(i+k-1,j')},\qquad
\hat c_{(i,j)}\hat c^\dagger_{(i+k-1,j')},
\label{Eq.2D_type1_pc}\\
&\hat c^\dagger_{(i,j)}\hat c^\dagger_{(i+k-1,j')},\qquad
\hat c_{(i,j)}\hat c_{(i+k-1,j')}.
\label{Eq.2D_type1_pnc}
\end{align}
Here all other sites are acted on by the identity.  The coefficients of these four
inputs satisfy, for any $i,j,j'$,
\begin{align}
q_{\hat c^\dagger_{(i,j)}\hat c_{(i+k-1,j')}}
&=q_{\hat c^\dagger_{(i+1,j)}\hat c_{(i+k,j')}},
\label{Eq.2D_shift1}\\
q_{\hat c_{(i,j)}\hat c^\dagger_{(i+k-1,j')}}
&=q_{\hat c_{(i+1,j)}\hat c^\dagger_{(i+k,j')}},\\
q_{\hat c^\dagger_{(i,j)}\hat c^\dagger_{(i+k-1,j')}}
&=-q_{\hat c^\dagger_{(i+1,j)}\hat c^\dagger_{(i+k,j')}},\\
q_{\hat c_{(i,j)}\hat c_{(i+k-1,j')}}
&=-q_{\hat c_{(i+1,j)}\hat c_{(i+k,j')}}.
\label{Eq.2D_shift4}
\end{align}
\end{lemma}

\begin{proof}[Proof of Lemma~\ref{lem:2D_step1}]
First, we show that a nonzero coefficient is possible only when each boundary column
contains one bare $\hat c$ or $\hat c^\dagger$.  Let the horizontal support occupy the
columns $i,i+1,\ldots,i+k-1$. We choose the site labels so that this interval
does not cross the periodic boundary, and choose the last nonidentity site
$(i+k-1,j')$ in the ordering~\eqref{Eq.2D_order} within the rightmost column.
Suppose first that its onsite factor is
$\hat c(x,y)_{(i+k-1,j')}$.  Write the input as
\begin{align}
\mathbf A_i^k=\mathbf E\hat c(x,y)_{(i+k-1,j')},
\label{Eq.2D_step1_example_input}
\end{align}
Let $(i_1,j_1),\ldots,(i_m,j_m)$ be all other sites on which
$\mathbf A_i^k$ has a nonidentity factor, listed in the order~\eqref{Eq.2D_order}.
Their coordinates satisfy $i=i_1\leq\cdots\leq i_m\leq i+k-1$ and
$1\leq j_l\leq L$ for $l=1,2,\ldots,m$.
In particular, $i_1=i$ because the input has width $k$ in the horizontal direction.
If $i_l=i+k-1$, then $j_l<j'$ by the choice of the last nonidentity site.
Then
\begin{align}
\mathbf E
={}&e_{(i_1,j_1)}e_{(i_2,j_2)}\cdots e_{(i_m,j_m)},
\nonumber\\
e_{(i_l,j_l)}
={}&f(x_l,y_l)_{(i_l,j_l)}
\quad (l=1,2,\ldots,m),
\label{Eq.2D_step1_E_expansion}
\end{align}
where the onsite basis $f(x_l,y_l)$ is given in
Eq.~\eqref{Eq.2D_onsite_basis}.  The horizontal hopping term across the right boundary
gives
\begin{align}
[\mathbf A_i^k,
\hat c^\dagger_{(i+k-1,j')}\hat c_{(i+k,j')}]
=\mathbf E I(x,y)_{(i+k-1,j')}\hat c_{(i+k,j')}.
\label{Eq.2D_step1_example_output}
\end{align}
Denote the output on the right-hand side by $\mathbf B^{k+1}$.  Any other input of width
at most $k$ that generates the same output must be supported on the columns
$i+1,\ldots,i+k$ and must recreate the entire leftmost column of
$\mathbf B^{k+1}$ by one hopping commutator across the bond from column $i+1$ to
column $i$.  This is possible only if the leftmost column contains a single nonidentity
site $(i,j)$ with a bare $\hat c$ or $\hat c^\dagger$.  If it contains a number
operator, a phonon-dressed fermion, or nonidentity factors at two or more sites, there is
no competing input.  In that case the conservation condition gives
\begin{align}
r_{\mathbf B^{k+1}}=-t q_{\mathbf A_i^k}=0,
\end{align}
and hence $q_{\mathbf A_i^k}=0$ because $t\neq0$.

The other possible factors at $(i+k-1,j')$ are treated just as in the
one-dimensional proof.  For $\hat n(x,y)_{(i+k-1,j')}$ and
$\hat c^\dagger(x,y)_{(i+k-1,j')}$, respectively, choose the outputs
\begin{align}
&\mathbf E\hat c^\dagger(x,y)_{(i+k-1,j')}\hat c_{(i+k,j')},\nonumber\\
&\mathbf E I(x,y)_{(i+k-1,j')}\hat c^\dagger_{(i+k,j')}.
\label{Eq.2D_step1_other_outputs}
\end{align}
The same uniqueness argument again requires the leftmost column to contain a single bare
$\hat c$ or $\hat c^\dagger$.  Interchanging the left and right boundaries proves the
corresponding statement for the rightmost column.  It also excludes boson dressing at
either endpoint.  Thus a possible nonzero input has precisely one bare fermion in each
boundary column.

Second, we eliminate all remaining interior factors.  Consider, as a representative
case,
\begin{align}
\mathbf A^k
=\hat c^\dagger_{(i,j)}\mathbf E_{\mathrm{mid}}
\hat c_{(i+k-1,j')},
\label{Eq.2D_step1_middle_input}
\end{align}
where $\mathbf E_{\mathrm{mid}}$ contains all nonidentity factors in the intermediate
columns.  Extending the right endpoint gives
\begin{align}
\mathbf B^{k+1}
=\hat c^\dagger_{(i,j)}\mathbf E_{\mathrm{mid}}
I_{(i+k-1,j')}\hat c_{(i+k,j')}.
\label{Eq.2D_step1_middle_output}
\end{align}
By the boundary result just proved, the only possible competing input is supported on
columns $i+1,\ldots,i+k$ and uses the hopping term
$\hat c^\dagger_{(i,j)}\hat c_{(i+1,j)}$ to recreate the left endpoint.  Such an input
exists only if column $i+1$ contains no nonidentity factor other than the shifted bare
fermion $\hat c^\dagger_{(i+1,j)}$.  When this condition holds, let $\mathbf E'_{\mathrm{mid}}$
denote the remaining ordered product, which is supported in columns
$i+2,\ldots,i+k-2$.  The competing input is then
\begin{align}
\widetilde{\mathbf A}^k
=\hat c^\dagger_{(i+1,j)}\mathbf E'_{\mathrm{mid}}
I_{(i+k-1,j')}\hat c_{(i+k,j')},
\label{Eq.2D_step1_competing_input}
\end{align}
and it gives the same output through
\begin{align}
[\widetilde{\mathbf A}^k,
\hat c^\dagger_{(i,j)}\hat c_{(i+1,j)}]
=-\mathbf B^{k+1}.
\label{Eq.2D_step1_competing_commutator}
\end{align}
If the condition does not hold, Eq.~\eqref{Eq.2D_step1_middle_output} is generated
uniquely and the coefficient of Eq.~\eqref{Eq.2D_step1_middle_input} is zero.
Repeating the same comparison after each horizontal shift makes columns
$i+2,i+3,\ldots,i+k-2$ successively into boundary columns.  Any additional onsite
factor in one of these columns would then prevent a competing input.  Hence every factor
in $\mathbf E_{\mathrm{mid}}$ is the identity.  The other three choices of the two
endpoint fermions are treated in the same way.  Thus only the four inputs in
Eqs.~\eqref{Eq.2D_type1_pc} and \eqref{Eq.2D_type1_pnc} remain.

For example, the output
\begin{align}
\mathbf B^{k+1}
=\hat c^\dagger_{(i,j)}\hat c_{(i+k,j')}
\label{Eq.2D_step1_shift_output}
\end{align}
is generated by precisely the two commutators
\begin{align}
[\hat c^\dagger_{(i,j)}\hat c_{(i+k-1,j')},
 \hat c^\dagger_{(i+k-1,j')}\hat c_{(i+k,j')}]
&=\mathbf B^{k+1},\nonumber\\
[\hat c^\dagger_{(i+1,j)}\hat c_{(i+k,j')},
 \hat c^\dagger_{(i,j)}\hat c_{(i+1,j)}]
&=-\mathbf B^{k+1}.
\label{Eq.2D_step1_shift_commutators}
\end{align}
Hence $r_{\mathbf B^{k+1}}=0$ gives
$q_{\hat c^\dagger_{(i,j)}\hat c_{(i+k-1,j')}}
=q_{\hat c^\dagger_{(i+1,j)}\hat c_{(i+k,j')}}$.
The other three outputs give Eqs.~\eqref{Eq.2D_shift1}--\eqref{Eq.2D_shift4} in the
same way, with the fermionic signs identical to those in the one-dimensional step~1
relations.
\end{proof}

\begin{lemma}[Step 1 analysis for type~\ref{enum:input_2}]
\label{lem:2D_step1_type2}
Assume $\hat Q^k$ is a $k$-local conserved quantity of the two-dimensional Holstein
model~\eqref{Eq.dD_H_model} with $t\neq0$, $g\neq0$, and $\omega\neq0$.
The coefficient $q_{\mathbf A^k}$ is zero for every $k$-support input
$\mathbf A^k$ of type~\ref{enum:input_2} whose width is $k$ in the horizontal
direction.
\end{lemma}

\begin{proof}[Proof of Lemma~\ref{lem:2D_step1_type2}]
By reflection symmetry, it is sufficient to consider an input whose left boundary
column contains no fermionic operator and whose right boundary column contains at least
one.  Since the input has width $k$, its left boundary column is nonempty and therefore
contains only purely bosonic factors.  Write their ordered product as
\begin{align}
\mathbf F_i
=\prod_{s=1}^{p} I(x_s,y_s)_{(i,j_s)},
\qquad x_s+y_s>0,
\label{Eq.2D_type2_left_boundary}
\end{align}
where $1\leq j_1<j_2<\cdots<j_p\leq L$.  In the right boundary column,
choose the site $(i+k-1,j')$ with the largest $j'$ among those whose onsite
factor contains $\hat c$, $\hat c^\dagger$, or $\hat n$.
For the representative case in which this factor is
$\hat c(x,y)_{(i+k-1,j')}$, separate the remaining ordered factors according to whether
they occur before or after it and write
\begin{align}
\mathbf A_i^k
=\mathbf F_i\mathbf E_{<}\hat c(x,y)_{(i+k-1,j')}\mathbf E_{>}.
\label{Eq.2D_type2_input}
\end{align}
Here $\mathbf E_{<}$ contains the nonidentity onsite factors at sites $(u,v)$
with $i+1\leq u\leq i+k-2$, $1\leq v\leq L$, and those with
$u=i+k-1$, $1\leq v<j'$.
The product $\mathbf E_{>}$ contains only the factors at sites $(i+k-1,v)$
with $j'<v\leq L$; these are purely bosonic by the choice of $j'$.
All products follow the order~\eqref{Eq.2D_order}.

Commuting this input with the horizontal hopping term that crosses the right boundary
gives the $(k+1)$-support output
\begin{align}
&[\mathbf A_i^k,
\hat c^\dagger_{(i+k-1,j')}\hat c_{(i+k,j')}]\nonumber\\
&=\mathbf F_i\mathbf E_{<}I(x,y)_{(i+k-1,j')}
\mathbf E_{>}\hat c_{(i+k,j')}
\nonumber\\
&=: \mathbf B^{k+1}.
\label{Eq.2D_type2_output}
\end{align}
Any other input of width at most $k$ that produces this output must be supported on
columns $i+1,\ldots,i+k$ and must recreate the left boundary column by hopping across
the bond from column $i+1$ to column $i$.  Such a hopping commutator necessarily places
a bare $\hat c$ or $\hat c^\dagger$ at the newly created site in column $i$.  It cannot
produce the purely bosonic product $\mathbf F_i$ in
Eq.~\eqref{Eq.2D_type2_left_boundary}.  Thus no competing input generates
$\mathbf B^{k+1}$, and
\begin{align}
r_{\mathbf B^{k+1}}=-t q_{\mathbf A_i^k}=0.
\end{align}
Since $t\neq0$, the coefficient of $\mathbf A_i^k$ is zero.  A right-boundary factor
$\hat n(x,y)$ or $\hat c^\dagger(x,y)$ is treated by choosing the corresponding hopping
term exactly as in the proof of Lemma~\ref{lem:step1_2}.  Reflection gives the same
conclusion when only the left boundary contains fermionic operators.  Hence every input
of type~\ref{enum:input_2} has a zero coefficient.
\end{proof}

As in the one-dimensional step~1 analysis, an input of type~\ref{enum:input_3} does not
generate a $(k+1)$-support output.  Thus step~1 alone imposes no useful restriction on
its coefficient.

\subsubsection{Proof step 2: Transverse-coordinate matching.}

The coordinates $j$ and $j'$ in Lemma~\ref{lem:2D_step1} are not yet constrained.  We now
show that they must agree.

\begin{lemma}[Step 2 analysis for type~\ref{enum:input_1}]\label{lem:2D_transverse}
Assume $\hat Q^k$ is a $k$-local conserved quantity of the two-dimensional Holstein
model~\eqref{Eq.dD_H_model} with $t\neq0$, $g\neq0$, and $\omega\neq0$.
For each of the four type~\ref{enum:input_1} inputs in
Eqs.~\eqref{Eq.2D_type1_pc} and \eqref{Eq.2D_type1_pnc}, the coefficient is zero
unless $j=j'$.
\end{lemma}

\begin{proof}[Proof of Lemma~\ref{lem:2D_transverse}]
Consider the representative input
\begin{align}
\mathbf A^k_0=\hat c_{(i,j)}\hat c_{(i+k-1,j')}
\label{Eq.2D_A0}
\end{align}
and the $k$-support output
\begin{align}
\mathbf B^k_0
=\hat c(0,1)_{(i,j)}\hat c_{(i+k-1,j')}.
\label{Eq.2D_B0}
\end{align}
The electron--phonon commutator gives
\begin{align}
[\hat c_{(i,j)}\hat c_{(i+k-1,j')},\hat n(0,1)_{(i,j)}]
=\hat c(0,1)_{(i,j)}\hat c_{(i+k-1,j')}.
\end{align}
The same output is generated by
hopping from
\begin{align}
\mathbf A^{k-1}_0
=\hat c(0,1)_{(i,j)}\hat c_{(i+k-2,j')}.
\label{Eq.2D_Achain0}
\end{align}
Consequently,
\begin{align}
gq_{\mathbf A^k_0}+tq_{\mathbf A^{k-1}_0}=0.
\label{Eq.2D_first_relation}
\end{align}

Keep the phonon at $(i,j)$ fixed and, for $m=0,1,\ldots,k-2$, define the inputs
\begin{align}
\mathbf A^{k-1}_m
=\hat c_{(i-m,j)}I(0,1)_{(i,j)}
\hat c_{(i+k-2-m,j')}.
\label{Eq.2D_Achain}
\end{align}
For $m=0$, the first two factors at the same site are understood as the single onsite
factor $\hat c(0,1)_{(i,j)}$.  For $0\leq m\leq k-3$, consider the outputs
\begin{align}
\mathbf B^k_{m+1}
=\hat c_{(i-m-1,j)}I(0,1)_{(i,j)}
\hat c_{(i+k-2-m,j')}.
\label{Eq.2D_Bchain}
\end{align}
Each $\mathbf B^k_{m+1}$ is generated by two relevant inputs: the left fermion of
$\mathbf A^{k-1}_m$ hops one site to the left, or the right fermion of
$\mathbf A^{k-1}_{m+1}$ hops one site to the right.  With the ordered-product
convention~\eqref{Eq.2D_product}, both commutators have the same sign:
\begin{align*}
[\mathbf A^{k-1}_m,
 \hat c^\dagger_{(i-m,j)}\hat c_{(i-m-1,j)}]
&=\mathbf B^k_{m+1},\\
[\mathbf A^{k-1}_{m+1},
 \hat c^\dagger_{(i+k-3-m,j')}\hat c_{(i+k-2-m,j')}]
&=\mathbf B^k_{m+1}.
\end{align*}
The condition $r_{\mathbf B^k_{m+1}}=0$ therefore gives
\begin{align}
q_{\mathbf A^{k-1}_{m+1}}
=-q_{\mathbf A^{k-1}_m}.
\label{Eq.2D_chain_relation}
\end{align}

After $k-2$ shifts, this chain reaches
\begin{align}
\mathbf A^{k-1}_{k-2}
=\hat c_{(i-k+2,j)}I(0,1)_{(i,j)}\hat c_{(i,j')}.
\label{Eq.2D_Achain_last}
\end{align}
We now determine when its coefficient can be nonzero.  Extend the left endpoint
of Eq.~\eqref{Eq.2D_Achain_last} once more and denote the
resulting $k$-support output by
\begin{align}
\mathbf C^k
=\hat c_{(i-k+1,j)}I(0,1)_{(i,j)}\hat c_{(i,j')}.
\label{Eq.2D_terminal_output}
\end{align}
The input $\mathbf A^{k-1}_{k-2}$ produces this output through
\begin{align}
[\mathbf A^{k-1}_{k-2},
 \hat c^\dagger_{(i-k+2,j)}\hat c_{(i-k+1,j)}]
=\mathbf C^k.
\label{Eq.2D_terminal_hopping}
\end{align}
The same operator $\mathbf C^k$, regarded as a $k$-support input, also reproduces
itself through the phonon term,
\begin{align}
[\mathbf C^k,I(1,1)_{(i,j)}]=\mathbf C^k.
\label{Eq.2D_terminal_phonon}
\end{align}
but its coefficient is already zero by Lemma~\ref{lem:2D_step1}, because it
contains the extra bosonic factor $I(0,1)_{(i,j)}$ in addition to its two endpoint
fermions.

It remains to check whether an input whose coefficient survived step~1 can
generate $\mathbf C^k$.  Such an input must be one of the bare two-endpoint
operators in Eqs.~\eqref{Eq.2D_type1_pc} and \eqref{Eq.2D_type1_pnc}.  The
electron--phonon commutator can produce the phonon at $(i,j)$ while retaining the
right fermion only when that fermion is at the same site, namely when $j=j'$.
In that case the additional input is, for the present representative choice,
\begin{align}
\mathbf K^k=\hat c_{(i-k+1,j)}\hat c_{(i,j)},
\qquad
[\mathbf K^k,\hat n(0,1)_{(i,j)}]=\mathbf C^k.
\label{Eq.2D_terminal_competitor}
\end{align}
For $j\ne j'$, no such step~1-surviving input exists.  Thus the coefficient
condition for $\mathbf C^k$ reduces to
$t q_{\mathbf A^{k-1}_{k-2}}=0$, and hence
$q_{\mathbf A^{k-1}_{k-2}}=0$.  Equations~\eqref{Eq.2D_chain_relation} and
\eqref{Eq.2D_first_relation} successively give $q_{\mathbf A^{k-1}_0}=0$ and
$q_{\mathbf A^k_0}=0$.

Replacing either endpoint in Eq.~\eqref{Eq.2D_A0} by $\hat c^\dagger$ changes only the
signs in the coefficient relations and does not affect the uniqueness argument when
$j\ne j'$.  Thus the same conclusion holds for all four step~1 candidates.
\end{proof}

\subsubsection{Proof step 3: Reduction to the one-dimensional proof.}

\begin{lemma}[Steps 2 and 3 analysis for type~\ref{enum:input_1}]
\label{lem:2D_type1_zero}
Assume $\hat Q^k$ is a $k$-local conserved quantity of the two-dimensional Holstein
model~\eqref{Eq.dD_H_model} with $t\neq0$, $g\neq0$, and $\omega\neq0$.
The coefficient $q_{\mathbf A^k}$ is zero for every $k$-support input
$\mathbf A^k$ of type~\ref{enum:input_1} whose width is $k$ in the horizontal
direction.
\end{lemma}

\begin{proof}[Proof of Lemma~\ref{lem:2D_type1_zero}]
By Lemmas~\ref{lem:2D_step1} and~\ref{lem:2D_transverse}, it remains only to treat the
four endpoint candidates in Eqs.~\eqref{Eq.2D_type1_pc} and
\eqref{Eq.2D_type1_pnc} with $j=j'$.  In this case, the phonon and the shifted fermion
in Eq.~\eqref{Eq.2D_Achain_last} occupy the same site and form
the onsite factor $\hat c(0,1)_{(i,j)}$.  The coefficient comparison therefore
continues exactly as in the one-dimensional proof.

Every remaining input lies on one row
\begin{align}
\Lambda_j=\{(i,j):i=1,2,\ldots,L\}.
\end{align}
For outputs supported on this row, vertical hopping produces a different ordered pair of
site coordinates and does not add a competing contribution.  The remaining coefficient
equations are therefore exactly those of the one-dimensional Holstein model, with every
site $i$ replaced by $(i,j)$.

For example, the step~1 and phonon-dressed step~2 shift relations for the $cc$-type input
$\hat c_{(i,j)}\hat c_{(i+k-1,j)}$ give, respectively,
\begin{align}
q_{\hat c_{(i,j)}\hat c_{(i+k-1,j)}}
&=(-1)^{k-1}
q_{\hat c_{(i+k-1,j)}\hat c_{(i+2k-2,j)}},
\label{Eq.2D_cc_step1}\\
q_{\hat c_{(i,j)}\hat c_{(i+k-1,j)}}
&=(-1)^k
q_{\hat c_{(i+k-1,j)}\hat c_{(i+2k-2,j)}}.
\label{Eq.2D_cc_step2}
\end{align}
Hence both coefficients vanish. The other $cc$-type input,
$\hat c^\dagger_{(i,j)}\hat c^\dagger_{(i+k-1,j)}$, is eliminated in the same way.

For the $c^\dagger c$-type input, step~2 gives, for example,
\begin{align}
gq_{\hat c^\dagger_{(i,j)}\hat c_{(i+k-1,j)}}
=tq_{\hat c^\dagger_{(i+1,j)}\hat c(0,1)_{(i+k-1,j)}}.
\label{Eq.2D_pc_relation}
\end{align}
The $(k-1)$-support outputs used in step~3 give the same finite coefficient system as in
Sec.~\ref{sec:proof}.  Writing $a$ for $\omega$ times the common phonon-dressed
coefficient, $b$ for the corresponding bare contribution times $g$, and
$c_2,\ldots,c_{k-1}$ for the successive hopping contributions times $t$, it reads
\begin{align}
a+b-c_2&=0,\nonumber\\
a+c_2-c_3&=0,\nonumber\\
&\ \vdots\nonumber\\
a+c_{k-2}-c_{k-1}&=0,\nonumber\\
a-b+c_{k-1}&=0.
\label{Eq.2D_telescoping}
\end{align}
For $k=3$, only the first and last equations are present.  Adding the equations gives
$(k-1)a=0$.  Since $k \geq 3$ and $\omega \neq 0$, the dressed coefficient vanishes, and
Eq.~\eqref{Eq.2D_pc_relation} then eliminates the $k$-support input. The other
$c^\dagger c$-type input is treated in the same way. Thus all coefficients of inputs of
type~\ref{enum:input_1} vanish.
\end{proof}

\subsubsection{Inputs of type~\ref{enum:input_3} and completion of the proof.}

It remains to eliminate inputs of type~\ref{enum:input_3}. Apply the one-dimensional
coefficient comparison for type~\ref{enum:input_3}
along every horizontal row while holding all other onsite factors fixed.  Horizontal hopping
changes only the same pair of sites as in the one-dimensional proof, while vertical hopping
produces an output with different site coordinates. Since inputs of type~\ref{enum:input_2}
have already been
eliminated, the same site-by-site argument removes all interior fermionic factors. A
remaining input is therefore purely bosonic. Writing
$S=\operatorname{Supp}\mathbf{A}^k$, it takes the form
\begin{align}
\mathbf A^k=\prod_{\overrightarrow{r}\in S}
I(x_{\overrightarrow{r}},y_{\overrightarrow{r}})_{\overrightarrow{r}}.
\label{Eq.2D_pure_boson}
\end{align}
At each site $\overrightarrow{r}\in S$, the electron--phonon commutators
\begin{align}
[I(x,y)_{\overrightarrow{r}},\hat n(0,1)_{\overrightarrow{r}}]
&=-x\hat n(x-1,y)_{\overrightarrow{r}},\\
[I(x-1,y+1)_{\overrightarrow{r}},\hat n(1,0)_{\overrightarrow{r}}]
&=(y+1)\hat n(x-1,y)_{\overrightarrow{r}}
\end{align}
relate its coefficient to that obtained by replacing $I(x,y)$ with
$I(x-1,y+1)$.  Repetition at every site relates Eq.~\eqref{Eq.2D_pure_boson} to
\begin{align}
\widetilde{\mathbf A}^k
=\prod_{\overrightarrow{r}\in S}
I(0,x_{\overrightarrow{r}}+y_{\overrightarrow{r}})_{\overrightarrow{r}}.
\end{align}
Finally,
\begin{align}
[I(0,n)_{\overrightarrow{r}},I(1,1)_{\overrightarrow{r}}]
=nI(0,n)_{\overrightarrow{r}}
\end{align}
gives
\begin{align}
\omega\sum_{\overrightarrow{r}\in S}(x_{\overrightarrow{r}}+y_{\overrightarrow{r}})
q_{\widetilde{\mathbf A}^k}=0.
\label{Eq.2D_boson_zero}
\end{align}
The sum is positive because the input has nonempty support of width $k$, and therefore its
coefficient is zero.  The preceding proportionality relations give
$q_{\mathbf A^k}=0$.  Lemmas~\ref{lem:2D_step1}--\ref{lem:2D_type1_zero}, together
with this elimination of type~\ref{enum:input_3}, show that every coefficient of a
$k$-support input vanishes.  This proves the two-dimensional Holstein-model statement of
Theorem~\ref{t:higher_dimensions}.

\subsection{Two-dimensional Holstein--Hubbard model}

\begin{lemma}
\label{lem:2D_HH}
The two-dimensional Holstein--Hubbard model~\eqref{Eq.dD_HH_model} with $t\neq0$, $g\neq0$, and
$\omega\neq0$ has no $k$-local conserved quantity with $3\leq k<L/2$.
This statement holds for arbitrary $U$.
\end{lemma}

\begin{proof}[Proof of Lemma~\ref{lem:2D_HH}]
We use the ordered onsite basis~\eqref{Eq.2D_spinful_basis} and treat the three input
types in the same order as for the spinless model.  For type~\ref{enum:input_1}, the
spinful analogue of
Lemma~\ref{lem:2D_step1}, obtained by carrying the two endpoint spin labels through
the step~1 argument as in the proof of Lemma~\ref{t:HH}, leaves only
\begin{align}
&\hat c^\dagger_{(i,j),\sigma}\hat c_{(i+k-1,j'),\tau},\qquad
\hat c_{(i,j),\sigma}\hat c^\dagger_{(i+k-1,j'),\tau},
\label{Eq.2D_HH_type1_pc}\\
&\hat c^\dagger_{(i,j),\sigma}\hat c^\dagger_{(i+k-1,j'),\tau},\qquad
\hat c_{(i,j),\sigma}\hat c_{(i+k-1,j'),\tau},
\label{Eq.2D_HH_type1_pnc}
\end{align}
with $\sigma,\tau\in\{\uparrow,\downarrow\}$.  The Coulomb term is onsite and does not
contribute to the $(k+1)$-support outputs.  Therefore the same boundary analysis applies,
and the spinful analogue of Lemma~\ref{lem:2D_step1_type2} shows that all inputs of
type~\ref{enum:input_2} have zero coefficients.

The step~2 chain now has the representative inputs
\begin{align}
\mathbf A^{k-1}_m
=\hat c_{(i-m,j),\sigma}I(0,1)_{(i,j)}
\hat c_{(i+k-2-m,j'),\tau}.
\end{align}
Its left shifts use hopping of spin $\sigma$, and its right shifts use hopping of spin
$\tau$.  The spinful analogue of Lemma~\ref{lem:2D_transverse} shows that hopping
preserves the row coordinate, so the last shift requires $j=j'$.  It does not require
$\sigma=\tau$: a phonon and a fermion of either spin form a valid onsite factor when
their spatial sites coincide.  Once
$j=j'$, the remaining equations are precisely the one-dimensional Holstein--Hubbard
equations proved in Appendix~\ref{sec:HH_proof}, with the two endpoint spin labels carried
throughout.  They eliminate Eqs.~\eqref{Eq.2D_HH_type1_pc} and
\eqref{Eq.2D_HH_type1_pnc} for arbitrary $U$.

It remains to treat inputs of type~\ref{enum:input_3}.  Their calculation is unchanged
from the spinless case: the site-by-site argument first removes every fermionic factor.
The Coulomb term then commutes with the remaining purely bosonic input, which is eliminated
by Eq.~\eqref{Eq.2D_boson_zero}.  Hence all three input types have zero coefficients,
and no $k$-local conserved quantity exists in the stated range.
\end{proof}

The two-dimensional Holstein result proved above and Lemma~\ref{lem:2D_HH} establish the
$d=2$ case of Theorem~\ref{t:higher_dimensions}.

\subsection{Extension to $d\geq3$}

Finally, consider $d \geq 3$.  After choosing a maximal-width direction, write a site as
\begin{align}
\overrightarrow{r}=(i,\overrightarrow{r}_{\perp}),
\qquad \overrightarrow{r}_{\perp}=(i_2,i_3,\ldots,i_d).
\end{align}
The step~1 conclusions in Lemmas~\ref{lem:2D_step1} and
\ref{lem:2D_step1_type2} are unchanged, with $j,j'$ replaced by the transverse
coordinate vectors $\overrightarrow{r}_{\perp},\overrightarrow{r}_{\perp}'$.  In the analogue of
Eqs.~\eqref{Eq.2D_Achain}--\eqref{Eq.2D_Achain_last}, hopping along the chosen direction
preserves every component of $\overrightarrow{r}_{\perp}$.  In the terminal comparison
corresponding to Eqs.~\eqref{Eq.2D_terminal_output}--\eqref{Eq.2D_terminal_competitor},
an additional step~1-surviving input exists only if the phonon and the shifted fermion
form one onsite endpoint, namely if
\begin{align}
\overrightarrow{r}_{\perp}=\overrightarrow{r}_{\perp}'.
\end{align}
Thus the analogue of Lemma~\ref{lem:2D_transverse} requires equality of all transverse
coordinates.  If any transverse component differs, the terminal coefficient condition
has no additional nonzero contribution and the auxiliary input coefficient vanishes.
Hopping in the other $d-1$ directions produces outputs with a
different pair of site coordinates and does not modify this chain.  The surviving inputs
therefore lie on a one-dimensional line parallel to the selected direction, where the
one-dimensional Holstein or Holstein--Hubbard proof applies.  For the spinful model, this
geometric step again imposes no equality of the endpoint spin labels. Thus no coefficient
of a $k$-support input survives for any $d \geq 2$, proving
Theorem~\ref{t:higher_dimensions}.

\bibliography{holstein}

\end{document}